\documentclass[11pt]{article}

\usepackage[margin=1in]{geometry}

\usepackage{natbib}
\bibpunct[, ]{(}{)}{,}{a}{}{,}%
\usepackage{setspace}

\usepackage{amsmath,amssymb,amsthm,mathtools}
\usepackage[normalem]{ulem} 
\usepackage[makeroom]{cancel}

\usepackage{graphicx}
\usepackage[hypcap=false]{caption}
\usepackage{subcaption}

\usepackage{tikz}
\usetikzlibrary{decorations.pathreplacing}

\usepackage{stackengine}

\usepackage{xcolor}

\usepackage{booktabs}

\usepackage{enumerate}

\usepackage{hyperref}

\usepackage{appendix}

\usepackage{soul}

\newcommand{\change}[1]{{\color{black} #1}}

\newcommand{\E}{\mathbb{E}}
\newcommand{\R}{\mathbb{R}}
\renewcommand{\P}{\mathbb{P}}
\newcommand{\A}{\mathcal{A}}
\newcommand{\M}{\mathcal{M}}
\newcommand{\F}{\mathcal{F}}

\newtheorem{theorem}{Theorem}[section]
\newtheorem{lemma}[theorem]{Lemma}
\newtheorem{prop}[theorem]{Proposition}
\newtheorem{corollary}[theorem]{Corollary}
\newtheorem{assumption}{Assumption}
\newtheorem{remark}{Remark}[section]

\makeatletter
\newcommand*{\barfix}[2][.175ex]{%
	\mathpalette{\@barfix{#1}}{#2}%
}
\newcommand*{\@barfix}[3]{%
	\vbox{%
		\kern#1\relax
		\hbox{$#2#3\m@th$}%
	}%
}
\makeatother

\counterwithin{figure}{section}
\counterwithin{equation}{section}
\counterwithin{table}{section}

\mathtoolsset{showonlyrefs=true}

\begin{document}


\vspace{-2cm}
\renewcommand{\UrlFont}{\footnotesize}
\title{Formation of Optimal Interbank Networks under Liquidity Shocks}

\author{Daniel E. Rigobon\thanks{Department of Operations Research \& Financial Engineering (ORFE), Princeton University.  \href{mailto:drigobon@princeton.edu}{\texttt{drigobon@princeton.edu}}}  \\
\vspace{-0cm}
\and Ronnie Sircar\thanks{Department of Operations Research \& Financial Engineering (ORFE), Princeton University. \href{mailto:sircar@princeton.edu}{\texttt{sircar@princeton.edu}}} \\
\vspace{-0cm}
}
\date{\today}

\maketitle
\vspace{-0.8cm}


\begin{center}
\renewcommand{\UrlFont}{\normalsize}
\begin{abstract}
\change{We study the formation of an optimal interbank network in a model where banks control both their supply of liquidity, through cash reserves, and their exposures to other banks' risky projects. 
The value of each bank's project may suddenly decline depending on their cash reserves and both the occurence and magnitude of liquidity shocks. In two distinct settings, we solve the system-wide optimal control problem and obtain explicit formulas for the unique optimal allocations of capital. In the first decentralized setting, banks seek only to maximize their own expected utility. Second, a central planner aims to maximize the sum of all banks' expected utilities. Both of the resulting networks exhibit a `core-periphery' structure in equilibrium. However, in the decentralized setting, banks elect to hold less cash reserves compared to the centralized setting, leading to greater susceptibility to liquidity shortages. We characterize the behavior of the planner's optimal allocation as the size of the system grows. Surprisingly, the relative welfare gap is of constant order. Finally, we derive co-investment requirements that allow the decentralized system to achieve the planner's optimal level of risk. In doing so, we find that banks in the network's core are subjected to the highest co-investment requirements -- ensuring that they are sufficiently incentivized to hold significant cash reserves. Our analysis may inform regulators' requirements on banks' liquidity reserves, as have been debated in the wake of the 2023 regional banking crisis in the US.}
\end{abstract}
\end{center}

\begin{minipage}{0.7\textwidth}
	\small
	\textbf{Keywords:} Financial Institutions, Optimal Control, Investment\\ 
\end{minipage}


\newpage


\section{Introduction}
\label{sec:intro}

\change{In this paper, we construct a continuous-time model of a financial system in which banks are exposed to both their own and possibly each others' liquidity shocks. We analyze the equilibria in this model in two distinct settings: decentralized and centralized. In contrast to much of the existing literature on banking networks, which assumes that the linkages between financial institutions are exogenous, we allow this network to form endogenously based on optional decisions taken by participating institutions.}

The model proceeds as follows. Banks may specialize in different activities; some collect a large number of deposits, whereas others specialize in revenue generation. This heterogeneity is modeled by unique, proprietary, investment opportunities (e.g. a portfolio of commercial loans) available to each bank. We assume that these opportunities are scalable, and accessible to other banks within the system. For example, a deposit-collecting bank $i$ can obtain the large returns of investment bank $j$'s unique revenue-generating operations by directly investing into project $j$. We note that this construction is similar to the single-period models of~\citet{Rochet1996}, \citet{acemoglu2015systemic}, \change{and~\citet{erol2017network}}, and we also refer to these unique investment opportunities as `projects'. In these previous models and ours, the riskiness of these projects is tied to some decision taken by the associated bank.

Although these projects may accrue large rates of return, they are subject to a degree of risk. More precisely, a bank's project may be impacted by unpredictable liquidity shocks of random magnitude. The shocks are assumed to represent, for example, additional liquidity required for a bank's project to avoid losses, such as occurs in~\citet{Rochet1996} and \citet{acemoglu2015systemic}, \change{or the withdrawal of a large number of deposits, which triggered the decline and ultimate failure of Silicon Valley Bank in 2023.} If the size of such a shock exceeds the bank's current supply of liquidity (i.e. cash reserves), then the project's value instantaneously drops. \change{We refer to such an event as a liquidity shortage. When this occurs, all banks investing in this project suffer losses proportional to their stake in it.} Therefore, conditioned on the arrival of a liquidity shock, a bank's supply of cash determines their project's level of risk. Conversely, if a bank has sufficient cash reserves when a liquidity shock arrives, their project suffers no losses. Figure~\ref{fig:ex_liq_shock} illustrates the relationship between a bank's liquidity supply (i.e. cash reserves), the distribution of a liquidity shock's size, and the probability of liquidity shortage conditioned on shock arrival. 
\begin{figure}[!h]
	\centering
	\begin{tikzpicture}[xscale = 1.5, yscale=3, domain=0:5]	
		\draw[thick, ->](0,0)--(5,0);
		\draw[thick, ->](0,0)--(0,1.1);
		\draw[thick] (0.05,0) -- (-0.05,0) node [left, xshift = -2pt] {\footnotesize 0};
		\draw[thick] (0.05,1) -- (-0.05,1) node [left, xshift = -2pt] {\footnotesize 1};
		\draw[red, very thick] plot(\x,{1-exp(-\x)}) node [xshift = 35pt, yshift = -2pt] {\footnotesize \shortstack{CDF, size of \\ liquidity shock}};
		\draw[blue, very thick] (1,1) -- (1,-0.1) node [below] {\footnotesize \shortstack{Liquidity\\supply of firm}};
		\draw [decorate,decoration={brace,amplitude=5pt}, cyan, thick, xshift = -30pt, yshift = 0pt](1,0) -- (1,0.63) node [midway, xshift = -50pt] {\footnotesize \shortstack{Probability of \\ sufficient liquidity}};	
		\draw [decorate,decoration={brace,amplitude=5pt}, violet, thick, xshift = -30pt, yshift = 0pt](1,0.63) -- (1,1) node [midway, xshift = -50pt] {\footnotesize \shortstack{Probability of \\a liquidity shortage}};
		\draw [dashed, thick] (0,0.63) -- (5,0.63);
	\end{tikzpicture}
	\captionsetup{width=.8\linewidth, font = small, justification=justified}
	\caption{For a single bank, the relationship between the cumulative distribution function (CDF) of the size of a liquidity shock, their supply of cash, and the conditional probabilities of (in)sufficient liquidity.}
	\label{fig:ex_liq_shock}
\end{figure}
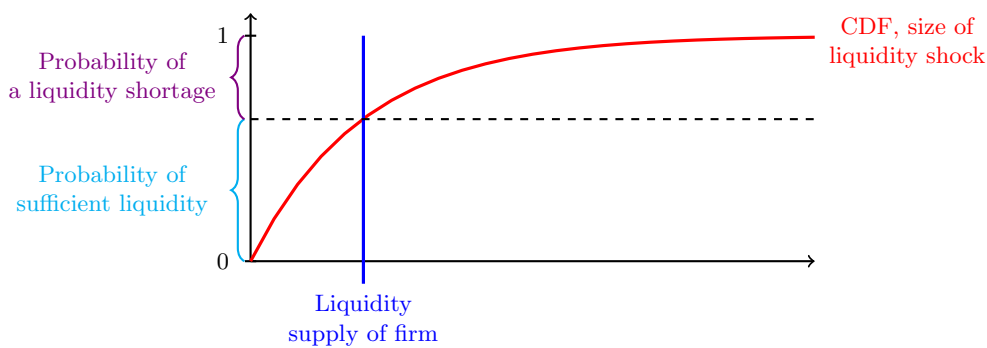

\change{Since the original draft of this paper, there have been some notable bank runs on mid-size US banks in 2023, particularly Silicon Valley Bank (SVB), First Republic Bank (FRB) and Signature Bank. These seem to have been initiated by interest rate increases by the Federal Reserve in its efforts to reduce inflation, leading to losses in value to these banks' bond portfolios. Rising interest rates also led to higher financing costs on many of their large institutional clients within the Silicon Valley tech sector, who then began withdrawing their deposits to meet their own liquidity needs. This impact on banks was then exacerbated, as other large depositors became concerned about the banks going into distress and their inability to recover beyond the $\$250$K insured by the FDIC. They withdrew large amounts very quickly, typically using online banking apps, and often outside of business hours and over weekends. For instance, on March 8, 2023, depositors withdrew $\$42$ billion from SVB in one day. 

However, these banks did not have sufficient cash reserves to meet depositors' demand, with significant amounts of capital tied up in long-term treasury bonds. These bonds had been purchased during an earlier low interest rate era and were liquidated at considerable loss. Ultimately this was not enough to restore depositor confidence, and the banks ended up insolvent due to lack of liquidity. The run on SVB likely lowered confidence in similar tech industry concentrated banks such as FRB, which led to a run on these banks. While our model was not designed to capture this kind of contagion effect, it does indeed focus on the impact and prevention of liquidity shocks through cash reserves in an inter-connected banking network. 
}

A key focus of this paper is that each bank endogenously chooses to allocate its capital between cash reserves (i.e. supply of liquidity) and other banks' risky projects. To that end, we will study the optimal capital allocations for two extreme settings of the financial system. First, we consider the \emph{decentralized} case -- wherein each bank freely allocates their capital with pure self-interest. They seek only to maximize their utility of wealth at a given terminal time. We note that this setting reflects a game-theoretic equilibrium. Second, we consider a \emph{centralized} setting -- where a single social planner makes the allocation decisions for all banks concurrently and aims to maximize the sum of individual banks' utilities.

In both decentralized and centralized cases, we derive the dynamic programming equations for the respective value functions, and explicitly compute the optimal allocations, \change{which yield a core-periphery network structure in both cases.} Under stricter technical conditions, we also prove uniqueness. We observe a discrepancy between the optimal allocations computed in both settings: the social planner chooses to hold a greater supply of liquidity. This occurs because our model captures a simple negative externality: when individualistically determining how much cash to hold in reserves, a bank sets the level of risk experienced by external investors in its project. An individual bank, operating in its own interests, fails to internalize these external investors' exposures when choosing their supply of cash. In contrast, the planner is cognizant of the network-wide welfare and acts accordingly by reducing the level of risk through increased cash reserves according to each project's degree of external investment. By design, the social planner achieves the welfare-maximizing (i.e. first-best) allocation for the financial system. As a consequence of this discrepancy, liquidity shortages are less likely to occur under the social planner than under decentralized behavior. However, we also observe that interbank exposures are larger in the centralized setting, and hence each liquidity shortage becomes more damaging to the system. This tension between the likelihood and severity of extreme events bears a resemblance to the `robust yet fragile' observation made by~\citet{Gai2010} and~\citet{acemoglu2015Bsystemic}. In particular, we find that this feature is associated with the socially optimal allocation of capital in the financial system.

We also study how the two optimal allocations differ as the financial system's size increases. Two natural points of comparison are: 1) the difference in, and 2) the ratio of, social welfare between both settings. The former comparison measures the nominal size of the inefficiency, and the latter its relative size (which we refer to as the relative welfare gap and was dubbed the `price of anarchy' by~\citet{papadimitriou2001algorithms}). Perhaps counter-intuitively, we find that the relative welfare gap remains bounded by a constant as the size of the system grows. Namely, the nominal size of the system's inefficiency grows at the same rate as the social welfare itself. These results are first derived theoretically, and also verified in simulations. Finally, we show that it is possible to alter banks' investment constraints to replicate the social planner's optimal allocation, \change{which requires increasing restrictions only on banks in the network's core.}

There are several interesting takeaways from this paper. \change{First, we are able to explicitly solve a continuous-time network optimal control problem and characterize the loss in efficiency due to decentralized behavior. Our model contains: i) multiple interacting agents, each solving their own utility maximization Merton problem, ii) uncertainty driven by jump processes rather than standard Brownian motion, and iii) agents' ability to control the intensity of their jump process. To the best of our knowledge, explicit solutions in a model with multiple controlled stochastic intensities is novel.}
Second, we find that the planner's optimal allocation leads to low-frequency and high-magnitude losses in the system. This may imply that the `robust-yet-fragile' feature of financial networks is, in some sense, socially optimal. However, we see that the planner compensates for the larger-magnitude losses by ensuring they are less likely. As a result, the centralized equilibrium is associated with larger interbank investments throughout the system. \change{Additionally, in both settings we see that the endogenous financial networks exhibit a `core-periphery' structure, where only a subset of banks' projects are invested in by other banks. Intuitively, we show that only banks in the core of the network must be subjected to stricter regulation to replicate the planner's optimal allocation. This suggests there is greater value in regulation of banks in the endogenously-determined core of the network, over those in the periphery.}

\change{Indeed, the 2023 failures of banks such as SVB spurred a regulatory response in the US to raise capital requirements on banks by $19\%$. This led to an intensive lobbying effort from the banking sector, which termed the new proposals ``Basel Endgame", and the increase has subsequently been halved to around $9\%$ at the time of writing. Analysis of the centralized setting can inform regulatory debates and about what constitutes an adequate supply of liquidity. We refer, for instance, to the speech \citet{barr-speech} by the Vice Chair for Supervision at the Federal Reserve, which states that: ``{\em To address the lessons about liquidity learned last spring, we are exploring targeted adjustments to our current liquidity framework. Over the last year, many firms have taken steps to improve their liquidity resilience, and the regulatory adjustments we are considering would ensure that all large banks maintain better liquidity risk management practices going forward. They would also complement the capital requirements by improving banks' ability to respond to funding shocks.}''}

This paper is organized as follows. Section~\ref{ssec:intro-lit} reviews several related branches of literature and background. Section~\ref{sec:mod} introduces the model of the financial system, the dynamics of each instrument, and the control problem posed to each bank. In the first part of our main results, Section~\ref{sec:opts} derives the endogenous equilibria for interbank networks in the decentralized (Section~\ref{ssec:opts-dec}) and centralized (Section~\ref{ssec:opts-cen}) settings. We compare these two optimal allocations in Section~\ref{sec:diff}, including an asymptotic analysis of the \change{relative welfare gap}. Finally, Section~\ref{sec:disc} concludes with a discussion of our results and directions for future work.

\subsection{Related Literature}
\label{ssec:intro-lit}

A strong motivation for this paper follows from the systemic risk literature, where much of the existing work assumes a given or exogenous network structure for the financial system. An early paper by~\cite{Allen2000} studies several stylized structures of interbank claims, and finds that the structure determines whether or not a local liquidity shock propagates throughout the system. Later papers seek to answer similar questions with distinct models. For instance~\citet{Gai2010} and~\citet{Gai2011} find that systemic liquidity crises can emerge in highly interconnected financial networks, albeit with low probability. \citet{Caccioli2014} present a model in which firms' overlapping portfolios can lead a single default to cause mark-to-market losses throughout the system -- perhaps leading to additional defaults. In~\citet{Elliott2014}, firms directly own claims on each others' assets and suffer sudden bankruptcy losses if their valuation falls below a threshold. \citet{Battiston2012a} studies a continuous-time process representing financial robustness, and allows its evolution to depend on a given financial network. Finally, several papers including~\citet{amini2016inhomogeneous,detering2019managing,detering2020financial} and \citet{detering2021integrated} seek to characterize the asymptotic behavior of contagion cascades in random inhomogeneous networks as the system's size grows. In their respective studies, these different mechanistic models are investigated both theoretically and in simulations. 
\change{In addition, the two-period model of~\citet{ramirez2022regulating} shows that under extreme levels of uncertainty, banks fail to internalize their effect on contagion cascades.}
However, the explicit or implicit networks in these papers share one common feature -- they are fixed or generated according to canonical random graph models. As previously highlighted, we believe this assumption may not be realistic; institutions in the financial system make optimal investment decisions, and the resulting network is endogenous -- not random or otherwise pre-specified. In contrast to this branch of the literature, our model enables us to investigate how the organization and fundamental parameters of the financial system can lead to the emergence and scale of its inefficiencies.

\change{We use several tools from continuous-time portfolio optimization in this work, which begins with the foundational papers of}~\citet{merton1969lifetime,merton1971optimum}. Merton studies the optimal portfolio allocation between risk-free and risky assets for a investor who maximizes their expected discounted utility of consumption. In these models, the returns of each risky asset are driven by correlated Brownian motions. Following from Merton's seminal papers, there is a wealth of literature on extensions of the original problem; see~\citet{rogers2013optimal} and references therein. The stochastic control techniques we use in this paper for deriving the optimal allocations will be similar to Merton's original work and its subsequent branch of literature, but here we will be studying a financial system in which \emph{all} participants are simultaneously determining their optimal allocations of wealth -- not only an individual. Moreover, to the best of our knowledge, the ability for multiple players to control the jump intensity of their the risk to their asset's returns has not been previously studied in the area of portfolio optimization. \change{A related mean field game with a continuum of cryptocurrency miners controlling the intensity of their discovery arrival process is analyzed in~\citet{MFGbitcoin}.}

There are, however, several papers that study an individual who incurs a cost to control the intensity of a jump process, such as~\citet{biais2010large}, \citet{pages2014mathematical}, \citet{capponi2015dynamic}, \citet{hernandez2020bank}, and~\citet{bensalem2020continuous}. These studies focus on Principal-Agent models and largely analyze the optimal contract and behavior. Moreover, they focus on the presence of moral hazard, where the Principal is unable to observe the Agent's efforts. Our mathematical approach for determining a bank's optimal cash reserves is similar to the models used in these papers. However, there are a few important differences. First, we study these optimizations performed simultaneously within a large system, and second, we focus on the inefficiencies that arise when banks optimize only in their self-interest. Additionally, our setting assumes perfect information, so there is no moral hazard.

We note that the high-level ideas in this paper are similar to the literature on optimal network formation. For early work in this area, see~\cite{jackson1996strategic} and~\cite{bala2000noncooperative}, where the authors present a process by which individuals choose to create edges with each other in a game-theoretic model. In these studies, individuals must balance a trade-off between the cost of forming an edge and the rewards associated with the edge. Our paper differs primarily from these studies through our emphasis on the financial features of the model, and the fact that edges are cost-less to form. A more realistic extension of our work would certainly incorporate these kinds of fixed costs.

Most closely related to this paper is the study of endogenous financial networks, including~\citet{erol2017network,zawadowski2013entangled, bluhm2014endogenous, acemoglu2015systemic, acemoglu2021systemic, babus2016formation} and~\citet{farboodi2021intermediation}. 
\change{\citet{erol2017network} studies the formation of linkages along which liquidity can flow between banks in a system, where each bank has access to a proprietary project. They find that excessive restrictions on banks' liquidity reserves increase the level of systemic risk. Our model primarily differs from this work in that linkages between banks do not provide liquidity to either counterparty, and there is no mechanism for contagion.}
The work of \citet{zawadowski2013entangled} shows that individual banks may fail to achieve the socially-optimal outcome by not buying insurance against their counterparties' default. While the author's model differs greatly from ours, we similarly find that individual banks' optimal behavior fails to take into account an externality on the system. A model by~\citet{babus2016formation} presents an extension of~\citet{Allen2000}. Her model allows banks to make optimal lending and borrowing decisions to redistribute liquidity throughout the system, and a highly-connected network is again found to be the most resilient to contagion. We share the idea of idiosyncratic liquidity shocks, but also study the planner's optimal allocation and compare it to the case where banks make individualistically optimal decisions.

The papers most similar to our own are~\citet{bluhm2014endogenous},~\citet{acemoglu2015systemic,acemoglu2021systemic} and~\citet{farboodi2021intermediation}. Our model is fundamentally different from those in these studies -- which are either static or consist of three distinct time periods. \change{In contrast, we analyze the endogenous financial network problem in a continuous-time environment, which allows us to leverage the powerful mathematical tools of stochastic optimal control.} First,~\citet{bluhm2014endogenous} construct a model of optimal interbank lending where banks face both liquidity and capital requirement constraints. In their model, both the interbank lending amounts and the market prices are determined endogenously. The authors show that contagion can occur (1) directly as a result of counterparty losses in the event of a default, or (2) indirectly through the mark-to-market losses incurred by a bank's portfolio in the event of a fire sale. The authors largely focus on numerical and simulation results, while we are able to provide explicit formulas for the equilibria we study. Moreover, our model endogenizes the initial sources of disruption to the system.

The contribution of~\citet{farboodi2021intermediation} characterizes how banks optimally lend to each other within a financial system where there is a strong incentive to serve as intermediaries within the chain of lending. In her model, as in ours, an interbank exposure will also allow the originating bank to access the surplus generated by a risky investment of an associated bank. She shows that the resulting network can have a core-periphery structure, and that due to the benefit of intermediation, banks' private incentives can fail to achieve the socially optimal outcome. Although there are similarities between this paper and ours, we do not focus on the incentive of intermediation, but instead on banks' optimal decisions to reduce the riskiness of their investments. Our results also replicate the core-periphery feature in her paper -- a small subset of banks with highly profitable investment opportunities form the financial network's core.

Finally,~\citet{acemoglu2015systemic, acemoglu2021systemic} endogenize both the decision of interbank lending and also the interbank interest rates. 
\change{\citet{acemoglu2021systemic} build a model whereby linkages between banks provide liquidity throughout the system, and the anticipation of contagion leads to scarce or non-existent credit. They study the relationship between the distribution of liquidity shocks and the occurence of systemic credit freezes, along with the role of different policy interventions in curbing these inefficiencies. In contrast to this work, our paper focuses on the investment decisions of participating banks in the financial system, where their liquidity shocks must be satisfied by their own cash reserves.}
\citet{acemoglu2015systemic} mirrors the work of~\citet{Rochet1996}, where banks exchange deposits to finance a project that yields high rate of return if run to conclusion, or low returns if liquidated prematurely. A bank faces external liabilities that may require them to liquidate these projects -- thereby passing losses onto its creditors. The authors find that the optimal contracts do indeed consider the first-order network effects, wherein a risk-taking bank must pay large interest rates to its creditors. However, these do not account for the `financial network externality', which can negatively affect banks that are not party to the contract. It follows that the resulting financial network may not be efficient (i.e. welfare-maximizing). While their models of interbank dynamics are similar to ours, the authors' analysis is largely focused on stylized networks in which equilibria are shown to exist. In this paper, we will instead allow the sparsity structure of the financial network to be endogenously determined by the interbank investment opportunities and other banks' decisions.

\section{Model}
\label{sec:mod}

Consider a financial system consisting of $n$ different banks. Let $(\Omega, \mathcal{E},\P)$ be a probability space, containing $n$ mutually independent Poisson processes $\tilde N^1_t,...,\tilde N^n_t$, $t\ge 0$, each of which has corresponding intensity $\theta_1,...,\theta_n > 0$. These counting processes will be used to indicate the arrival times of liquidity shocks to each respective bank. Define $\mathcal{F}$ to be the filtration generated by the full set of jump processes. Hence, we obtain the filtered probability space $(\Omega, \mathcal{E}, \F, \P)$.
The net capitalization (i.e. net value or wealth) of bank $i$ is given by the non-negative stochastic process $\{X_t^i\}_{t\ge 0}$, whose dynamics are described in the following.

\subsection{Investment Opportnities}
\label{sec:2.1}

First, the financial system contains a \change{long-term bond}, which accumulates a constant, fixed rate of return $\change{r \geq0}$ over all time. Therefore its \change{value at time $t$}, denoted $S^0_t$, evolves according to the ordinary differential equation $\frac{dS^0_t}{S^0_t} = r\,dt$. \change{The rate of return $r$ will be referred to in this paper as the `risk-free rate', since it is unaffected by any stochastic factors. However, we note that capital held in the bond is not free from \textit{liquidity} risk. Namely, we will assume that capital invested in the bond cannot be used to satisfy short-term liquidity needs. We think of the rate $r$ as capturing the rate of return for long-term and less liquid fixed income instruments, such as 5- or 10-year Treasury bills. Appendix~\ref{subapp:partially_liq_rf_asset} contains an extension of our model where this assumption is weakened and the bond can be liquidated for a fraction of its face value. Computationally, we see no significant differences in large enough banking networks.}

Each bank $i$ has access to a unique set of investments henceforth referred to as a `project', e.g. a collection of commercial loans. These projects are made available for investment to all other banks $j\neq i$ in the system. However, \change{these projects cannot be shorted, and any capital invested in them cannot be rapidly liquidated}. More precisely, \change{no bank} can use this capital to meet any \change{short-term} liquidity needs. While these projects accumulate large constant rates of return for investors, they will incur losses when the \change{associated bank suffers a liquidity shortage}. If such an event occurs, then the value of the project immediately drops \change{ to some non-zero fraction of its prior value}. For instance, it is plausible that a bank's revenue-generating operations intermittently require additional liquidity to cover a position or meet regulatory requirements. Failure to do so may lead to the bank's inability to realize the investment's gains, or may even directly cause losses.

Let $S^i_t$ denote the time-$t$ value of a single unit of capital invested in bank $i$'s project. Its dynamics are given by
\begin{equation}
	\label{eq:risky_asset_sde}
	\frac{dS_t^i}{S_t^i} = (\mu_i + r)\,dt - \phi_i\,dN_t^i, \quad i=1,\dots,n.
\end{equation}
\change{We assume that $\mu_i > 0$, so this project} has deterministic rate of return greater than $r$. The jump process $N^i_t$ is obtained by performing a thinning of the shock arrival process $\tilde N^i_t$, and is described in the next paragraphs. The increment $dN_t^i$ takes on values in $\{0,1\}$, and is \change{equal to one} if and only if bank $i$ \change{experiences a liquidity shortage at time $t$}. Finally, $0 < \phi_i < 1$ represents the \change{fraction of project value lost when such an event occurs.}


All \change{banks} in the system may experience liquidity shocks; if sufficiently large, these shocks \change{ adversely affect the value of their project}. A key feature of this paper is each bank's ability to control their project's susceptibility to such events -- by holding a greater supply of liquidity \change{ in the form of cash reserves}. In our model, this is represented through bank $i$'s ability to influence the intensity of the jump \change{process} $N_t^i$ that appears in~\eqref{eq:risky_asset_sde}.

A bank may hold a non-negative amount of their capital as \change{cash reserves, which do not accrue any interest}. \change{However, these cash reserves} are the \emph{only} source of \change{short-term} liquidity within the system, and are the sole tool by which a bank can hedge against the arrival of liquidity shocks. Namely, if a liquidity shock exceeds bank $i$'s cash reserves, their project experiences a \change{ substantial loss in value, which impacts both their own wealth and that of other banks investing in $i$'s project.} The jump increment $dN_t^i$ in~\eqref{eq:risky_asset_sde} equaling one represents the arrival of a shock that overwhelms bank $i$'s supply of cash, and its construction follows from a probabilistic model of liquidity shocks and a bank's \change{ cash reserves}.

Recall that our filtered probability space contains $n$ independent time-homogeneous Poisson processes $\tilde N_t^i$, with rates $\theta_i > 0$. At time $t$, if $\tilde N_t^i$ jumps, then bank $i$ experiences a liquidity shock of size $\zeta^i_t \cdot X^i_t$, where the random variable $\zeta^i_t$ is $\F_t$-measurable. We assume that these shocks are proportional to a bank's wealth, and each $\zeta^i_t$ is independently and identically distributed according to the cumulative distribution function (CDF) $F_i(\cdot)$.  For convenience of notation, the complementary CDF of $\zeta^i_t$ is denoted by $\bar F_i(\cdot) = 1-F_i(\cdot)$.

Let $c^i_t \ge 0$ denote the fraction of bank $i$'s capital held \change{ as cash reserves} at time $t$. When the shock to bank $i$ is larger than these reserves (i.e. $\zeta^i_t > c^i_t$), \change{ they undergo a liquidity shortage and} all investors in their project $i$ suffer an instantaneous return of $-\phi_i$ on their investment. In particular, if $c^i_t = 0$, then any liquidity shock to bank $i$ at time $t$, no matter how small, results in \change{losses}.

The jump process $N^i_t$ is constructed by independently flipping a coin at every \change{jump} arrival time of $\tilde N_t^i$, with failure probability given by $p^i_t=\bar F_i(c^i_t)$. Observe that $p^i_t = \P\left(\zeta^i_t > c^i_t  \left| d\tilde N^i_t = 1 \right.\right)$. \change{We let $dN_t^i = 1$ if and only if the flip is a failure}. It follows \change{from the thinning properties of Poisson processes} that the instantaneous rate at time $t$ of the process $N_t^i$ is equal to $\theta_i \bar F_i\left(c^i_t\right)$.\footnote{This result is a consequence of the thinning properties of Poisson processes. See, for instance, Theorem 1 in~\citet{lewis1979simulation}. If $\{c^i_t\}_{t\ge 0}$ is adapted (as we will require), then conditioned on time $t$, the previous jump process $\{N^i_s\}_{s\in [0,t]}$ has the desired rate function.} The second component of the rate, $\bar F_i(c^i_t) = \P(\zeta^i_t > c^i_t)$, is the probability that bank $i$ \change{ suffers a liquidity shortage}, conditioned on the time-$t$ arrival of a liquidity shock with CDF $F_i(\cdot)$. See Figure~\ref{fig:ex_liq_shock} for an illustration.

\change{
\begin{remark}
	In the main text, we assume that liquidity shocks smaller than the bank's cash supply (i.e. when $\zeta^i_t \le c^i_t$ and $d\tilde N^i_t = 1$) do not affect the bank's overall wealth. Appendix~\ref{subapp:liq_loss_shock} relaxes this assumption, and this extended model is still solvable up to first-order conditions, which must be solved numerically. As long as the expected shock size is small, there is no significant difference caused by this relaxation.
	\label{rmk2.1}
\end{remark}
}

Finally, we will require a few technical conditions on $F_i$:
\begin{assumption}
	\label{assn}
	We assume that each $F_i$ is absolutely continuous with respect to the Lebesgue measure. Its density is given by $f_i(\cdot) = F_i'(\cdot)$, which is assumed to be fully supported on $\R_+$, and monotonically decreasing (i.e. $f_i'(\cdot) < 0$).
\end{assumption}
If $f_i(\cdot)$ had compact support, then it would be possible for a bank's project to be riskless with large enough \change{ cash reserves}. Since the return of this project would be greater than the risk-free rate, this would lead to all other banks profiting infinitely by borrowing at the risk-free rate and investing in this riskless project. While the problem may remain analytically tractable, this outcome is not of practical interest. Our assumption that the density is monotonically decreasing will be used to establish uniqueness of the optimal controls.

\subsection{Dynamics of Wealth}
\label{ssec:model-wealth}

In this model, \change{ each bank may invest in any ther banks' projects.} Let $w^{ji}_t \ge 0$ denote the fraction of bank $j$'s capital invested in bank $i\neq j$'s project. The return \change{on} this \change{investment} is \change{described} by~\eqref{eq:risky_asset_sde}. Recall that $c^i_t$ equals the fraction of bank $i$'s wealth held as \change{ cash reserves}, which accumulates no return over time.

The final component influencing bank $i$'s wealth is their degree of investment in their own project, \change{which is equal to a given and fixed fraction of their current wealth}. Unlike \change{exposures to other banks' projects}, we assume that this quantity cannot be \change{optimized by the bank. For instance, it may be the case that bank $i$ is required by a regulator to be an investor in its own project.
In principle, we could imagine allowing bank $i$ to also control their exposure to their own project. However, doing so introduces the possibility for multiple equilibria, as shown in Appendix~\ref{subapp:endog_eta}.}

We will use $\frac{\eta_i}{\phi_i}$ to denote the fraction of bank $i$'s wealth that is invested in their own project. This implies that bank $i$ loses a constant fraction $\eta_i$ of its total wealth whenever \change{they suffer a liquidity shortage}. The parameter $\eta_i$ captures the severity of \change{losses on the distressed} bank -- in the extreme case of $\eta_i = 1$, a single liquidity shortage will wipe out bank $i$. Conversely, if $\eta_i = 0$, then bank $i$ has no stake in their project and is wholly unaffected by its dynamics. We will take $0 < \eta_i < 1$, away from the two extreme cases.

Due to self-financing constraints, the remaining $X_t^i(1 - c^i_t - \sum_{j\neq i} w^{ij}_t - \frac{\eta_i}{\phi_i})$ units of wealth are held as \change{long or short positions in the interest-bearing bond}. 
Putting together the dynamics for each component of bank $i$'s wealth, we see that $X_t^i$, follows 
\begin{equation}
	\frac{dX^i_t}{X^i_t} = \left(1 - c^i_t - \sum_{j\neq i}w^{ij}_t - \frac{\eta_i}{\phi_i} \right)\frac{dS^0_t}{S^0_t} + \sum_{j\neq i}w^{ij}_t \,\frac{dS^j_t}{S^j_t} + \frac{\eta_i}{\phi_i} \,\frac{dS^i_t}{S^i_t}\, , \quad i=1,\cdots,n.
\end{equation}
Using~\eqref{eq:risky_asset_sde} and the dynamics of $S^0_t$, we obtain the following simplified expression:
\begin{equation}
	\label{eq:wealth_sde}
	\frac{dX^i_t}{X^i_t} = \left( (1 - c^i_t) r  + \sum_{j\neq i}w_t^{ij} \mu_j + \frac{\eta_i }{\phi_i} \mu_i\right)dt - \sum_{j\neq i} w_t^{ij} \phi_j \,dN^j_t - \eta_i \,dN^i_t \, , \quad i=1,\cdots,n.
\end{equation}
A novel contribution of this paper is the control $c^i_t$; while there is no return accumulated by this capital held as cash, it serves to reduce the likelihood that bank $i$ \change{suffers a liquidity shortage}, which would cause them to lose a fraction $\eta_i$ of their wealth.

We say that $(c^i_\cdot, w^{i \cdot}_\cdot )$ is in $\A^i_{s,t}$, the set of admissible controls for bank $i$ between times $s$ and $t$, if it is adapted to the filtration $\F$ and satisfies both $c^i_u \in \R_+$ and $w^{ij}_u \in [0,\phi_j^{-1})$ for all $u\in[s,t]$ and $j\neq i$. The upper bound on $w^{ij}_u$ ensures that wealth will always remains positive.

In the first decentralized setting we study, all banks seek to maximize, over their cash reserves and investments in other banks' projects, their own expected utility of wealth at a common terminal time $T< \infty$: equal to $\E \left[ U_i(X_T^i)\right]$. As is relatively standard in the literature, a bank's utility function $U_i \in \mathcal{C}^\infty (\R_+)$ is assumed to have constant relative risk aversion of $\gamma_i>0$, which yields the following family of utility functions:
\begin{equation}
	\label{eq:utility}
	U_i(x) = \begin{cases} 
		\frac{x^{1-\gamma_i}}{1-\gamma_i} & \gamma_i > 0, \gamma_i \neq1\\
		\log x & \gamma_i = 1.
	\end{cases}
\end{equation}
For generality, we will allow banks to have heterogeneous risk-aversion coefficients $\gamma_i$ in the decentralized network formation of Section \ref{ssec:opts-dec}, while in Section \ref{ssec:opts-cen} a central planner will require all banks to have logarithmic utility.

\section{Decentralized and Centralized Financial Networks}
\label{sec:opts}

We consider two distinct organizations of the financial system. In the first, banks operate only in their self-interest -- seeking to maximize their own expected terminal utility. We call this the \emph{decentralized} setting, as there is no coordination between banks. Instead, each bank's optimal allocation reflects their best response to the others' decisions. On the other hand, the \emph{centralized} setting in Section~\ref{ssec:opts-cen} will consider the perspective of a single social planner who determines all banks' allocations to maximize welfare -- as measured by the sum of all banks' utilities.

Both allocations are important to consider. The decentralized optimum reflects a game-theoretic or market-mediated equilibrium of the financial system, where each bank chooses their controls optimally given all others' actions. Therefore, from the perspective of individual banks this is a stable allocation. In contrast, the centralized optimum reflects the maximum total welfare that could exist in the financial system if banks' capital allocations were coordinated by a central planner. We will study the differences between these two optimal allocations, which reflect the severity of our model's externality, in Section~\ref{sec:diff}. Finally, the optimal allocations yield a financial network of interest, whose edges represent direct exposures between banks.

\subsection{\change{Financial Frictions}}
\label{ssec:friction}

\change{
	It is essential to note that the fundamental financial friction in this model is that markets are incomplete. Bank $i$ is unable to fully hedge the risk associated with another bank $j$'s distress process $N^j_t$. A bank may only partially insure against this risk by holding capital in the risk-free bond. In contrast, however, bank $i$ can control the risk of their \textit{own} jump process $N^i_t$ by choosing to hold cash reserves. However, bank $i$'s motivation for doing so is entirely self-preservational. They are not punished for the negative externality inflicted on other firms in the system because: i) there is no contract between banks allowing them to transfer the risk of a jump, and ii) bank $i$ is itself never responsible for losses to external investors. Given that such a friction exists, it is not a novel result that the equilibrium obtained by individual agents operating selfishly is distinct from the `welfare-maximizing' equilibrium obtained by a social planner.
}

\subsection{Decentralized Network}
\label{ssec:opts-dec}

Let us define the value function of bank $i$ to be the supremum over all admissible controls of their expected utility at the terminal time: 
\begin{equation}
	\label{eq:dec_val}
	V_i(t,x) = \sup_{(c^i_\cdot, w_\cdot^{i\cdot}) \in \A^i_{t,T}} \E \left[ U_i(X_T^i) \big | X_t^i = x \right].
\end{equation}
Recall that $A^i_{t,T}$ denotes the set of admissible controls for bank $i$ -- defined in Section~\ref{ssec:model-wealth}. Note also that each bank simultaneously solves their own optimization problem, and therefore the value function of bank $i$ in~\eqref{eq:dec_val} may depend on the capital allocations chosen by other banks within the system. In this sense, \change{ the decentralized setting} reflects a stable game-theoretic equilibrium.

Our first result derives the non-local dynamic programming (often referred to as the Hamilton-Jacobi-Bellman or HJB) equation for the value function under regularity.

\begin{prop}
	\label{prop:dec_opt_hjb}
	If there exist optimal controls and the value function in~\eqref{eq:dec_val} is $\mathcal{C}^{1,1}([0,T),\mathbb{R}_+)$, then it solves the following non-local partial differential equation (PDE):
	\begin{equation}
		\label{eq:dec_hjb}
		\begin{alignedat}{23}
			 0&= \partial_t V_i  + \sup_{c_i,  w_{i\cdot}} \Bigg\{ 	&&\left[\left(1-c_i\right)r  + \sum_{j\neq i}w_{ij} \mu_j  + \frac{\eta_i \mu_i}{\phi_i}\right]x  \partial_x V_i + \theta_i \bar F_i(c_i)\Big[ V_i(t,x(1 - \eta_i)) - V_i \Big] \\
			 & && + \sum_{ j\neq i} \theta_j \bar F_j(c_j)\Big[ 	V_i(t,x(1-\phi_j w_{ij})) - V_i \Big] \Bigg\},
		\end{alignedat}
	\end{equation}
	with terminal condition $V_i(T,x) = U_i(x)$. Where unspecified, the value function and its derivatives are evaluated at $(t,x)$.
\end{prop}

The proof is contained in Appendix~\ref{ssec:pfs-dec}, and follows from applying It\^{o}'s formula to the value function between $t$ and an appropriately defined sequence of stopping times. The assumption that optimal controls exist is verified by a subsequent result in Corollary~\ref{cor:dec_verification}.

Fortunately, it is possible to find a separable solution to~\eqref{eq:dec_hjb}, and explicit solutions for the optimal allocations. It is convenient to introduce the following notation:
\begin{equation}
	\label{eq:gamma}
	\Gamma(\delta;\gamma) = 
	\begin{cases}
	\frac{1-(1-\delta)^{1-\gamma}}{1-\gamma} & \gamma>0, \gamma\neq1\\
	-\log(1-\delta) & \gamma=1,
	\end{cases}
\end{equation}
for any $\delta \in[0,1)$. Within this range, we note that $\Gamma \ge 0$. There is a natural interpretation of this object; for a utility function of the form in~\eqref{eq:utility}, $\Gamma(\delta; \gamma)$ is proportional to the loss in utility caused by losing a fraction $\delta$ of one's wealth. More precisely, $\Gamma(\delta; \gamma_i) = x^{\gamma_i-1} \left[U_i(x) - U_i(x(1-\delta))\right] $ for any $x > 0$.

We can now state our second main result, which presents a solution to~\eqref{eq:dec_hjb} and computes the optimal decentralized allocation of capital for bank $i$.

\begin{prop}
	\label{prop:dec_hjb_soln}
	The unique optimal \change{ cash reserves and project investment} amounts for the maximization problem in~\eqref{eq:dec_hjb} are given by
	\begin{equation}
		\label{sol:dec}
		\begin{alignedat}{3}
			\hat c_i  &= \begin{cases}
				f_i^{-1}\left(\frac{r}{\theta_i\Gamma(\eta_i;\gamma_i)}\right) 	& 	\mbox{ if } \frac{r}{\theta_i\Gamma(\eta_i;\gamma_i)} \le f_i(0)
				\\
				0 & \text{ otherwise}
			\end{cases}	  \quad &&  \forall i = 1,\cdots,n
			\\
			\hat w_{ij} &= \begin{cases}
				\frac{1}{\phi_j}\left(1-\left(\frac{\phi_j\theta_j 		\bar F_j(\hat c_j)}{\mu_j}\right)^{1/\gamma_i}\right) & \text{ if } \frac{\phi_j\theta_j \bar F_j(\hat c_j)}{\mu_j} < 1 
				\\
				0 & \text{ otherwise.}
			\end{cases}   \quad &&  \forall j\neq i.
		\end{alignedat}
	\end{equation} 
	Furthermore, with the notation
	\begin{equation*}
		J_i^* = \frac{\eta_i \mu_i}{\phi_i} + (1 - \hat c_i )r - \theta_i 	\bar F_i(\hat c_i) \Gamma(\eta_i;\gamma_i) + \sum_{j\neq i}\hat w_{ij}\mu_j   - \theta_j \bar F_j(\hat c_j)\Gamma(\phi_j \hat w_{ij};\gamma_i),
	\end{equation*}
	the following are explicit solutions to~\eqref{eq:dec_hjb}: 
	\begin{enumerate}[(i)]
		\item if $\gamma_i = 1$ and 
		$U_i(x) = \log x$, we have $V_i(t,x) = g_i(t) + \log x$, where 
		$g_i(t) = (T-t)J_i^*$
		
		\item for $\gamma_i\neq1$ and $U_i(x) = \frac{x^{1-\gamma_i}}{1-\gamma_i}$, then we have $V_i(t,x) = g_i(t)U_i(x)$, where $g_i = e^{(1-\gamma_i)(T-t)J_i^*}$.
	\end{enumerate}
\end{prop}

The proof, which is again given in Appendix~\ref{ssec:pfs-dec}, follows from plugging in the proposed solution, simplifying, and then analyzing the necessary and sufficient conditions for optimality of the resulting maximization problem. A key observation in this proof is that the maximization problem in~\eqref{eq:dec_hjb} is additively separable between each of the controls $c_i, w_{i\cdot}$.

\begin{remark}
	The optimal interbank investment $\hat w_{ij}$ depends explicitly on $\hat c_j$ through the function $\bar F_j(\hat c_j)$. Moreover, for any choice of bank $j$'s cash \change{reserves}, there exists a corresponding optimal value of $w_{ij}$. In a game-theoretic sense, this would be bank $i$'s best response to $j$'s decision. However, bank $j$'s optimal value of $\hat c_j$ depends only on fixed model parameters. This ensures that $\hat c_j$ is bank $j$'s best response to \emph{any} decisions made by the other banks, and is therefore a dominant strategy. Hence, the `game' is trivialized -- one can compute every other banks' optimal $\hat c_j$, after which the corresponding $\hat w_{ij}$'s can be easily found.
\end{remark}

The final result of this subsection verifies that the solution given in Proposition~\ref{prop:dec_hjb_soln} is indeed equal to the value function.

\begin{corollary}
	\label{cor:dec_verification}
	The value function in~\eqref{eq:dec_val} is given by
	\begin{equation*}
		V_i(t,x) = \begin{cases}  g_i(t) + \log x &\text{ if }\gamma_i = 1 \\ g_i(t)\frac{x^{1-\gamma_i}}{1-\gamma_i} &\text{ otherwise,}\end{cases}
	\end{equation*}
	where $g_i(t)$ and the optimal controls are given in Prop.~\ref{prop:dec_hjb_soln}.
\end{corollary}

The proof in Appendix~\ref{ssec:pfs-dec} uses a verification argument. We show that any solution to~\eqref{eq:dec_hjb} that is once continuously differentiable in both time and space is equal to the value function. Since the proposed solutions in Proposition~\ref{prop:dec_hjb_soln} satisfy this regularity condition, we conclude the desired claim. Finally, we note that this result verifies the assumption made in Proposition~\ref{prop:dec_opt_hjb} regarding the existence of optimal controls.

\subsubsection*{Analysis of Decentralized Equilibrium}

\change{In analyzing our results, we first} note that each unit of additional cash provides some quantifiable marginal benefit by lowering the risk of a \change{liquidity shortage.} 
\change{The quantity $\Gamma(\eta_i,\gamma_i)$, defined in ~\eqref{eq:gamma}, is closely related to the loss in utility when a bank with risk aversion $\gamma_i$ loses a fraction $\eta_i$ of their wealth. However, the instantaneous hazard rate of an event in which $i$ suffers this loss in utility is controlled by $i$'s cash reserves through the intensity $\theta_i \bar F_i(c_i)$ of the thinned Poisson process.}
From the proof of Proposition~\ref{prop:dec_hjb_soln}, the optimal choice of $\hat c_i$ will solve the following:
$$\max_{c_i \ge 0} \big\{ -rc_i - \theta_i \bar F_i(c_i) \Gamma(\eta_i;\gamma_i) \big\},$$
which indicates that the resulting $\hat c_i$ achieves the optimal tradeoff between the cost of \change{maintaining liquidity reserves versus the instantaneous expected losses in utility due to a shortage}.
\change{In particular, the optimal $\hat c_i$ from the first-order condition of $r=-\theta_i f_i(\hat c_i) \Gamma(\eta_i;\gamma_i)$ ensures that the marginal cost of holding a single additional unit of liquidity $(r)$ equals the marginal reduction in expected losses provided by that unit of liquidity $\left(\theta_i f_i(\hat c_i) \Gamma(\eta_i;\gamma_i)\right)$.}
In the extreme case where $r$ is large, it may be too costly (relative to the potential losses) for a bank to hold any amount of \change{cash reserves}, leading to $\hat c_i = 0$.

\change{With explicit solutions for the optimal allocations given in Proposition~\ref{prop:dec_hjb_soln}, it is possible to analyze their dependence on the exogenous parameters of the system. Due to the monotonicity assumption on $f_i(\cdot)$ in Assumption~\ref{assn}, we can see that the optimal cash reserves $\hat c_i$ are decreasing in the liquidity risk premium $r$: the greater the opportunity cost of holding liquidity, the lower cash reserves are held by banks. We also see that $\Gamma(\eta_i; \gamma_i)$ is increasing in $\eta_i$. It follows that bank $i$'s optimal cash reserves, $\hat c_i$ increase with the fraction of wealth they stand to lose if a liquidity shortage occurs. In Section~\ref{ssec:res-replication} we will discuss how a regulator can achieve a more socially-optimal allocation by stipulating new values for $\eta_i$, each bank's degree of self-investment.}

In addition, the optimal \change{project investment by bank $i$ in project $j$}, $\hat w_{ij}$, depends on $i$ only through their risk aversion parameter $\gamma_i$. Hence, if $\gamma_i = \gamma_k$ then $\hat w_{ij} = \hat w_{kj}$. Although the fractional amount of these investments are equal, the nominal amounts may differ \change{because banks $i$ and $k$ may have different levels of wealth.} However, the optimal investment amount is decreasing in the investing bank's risk aversion coefficient $\gamma_i$, as we might expect.

The quantity $\frac{\mu_j}{\barfix[0.2ex]{\phi_j\theta_j \bar F_j(\hat c_j)}}$, which appears in~\eqref{sol:dec} for $\hat w_{ij}$, is similar to the well-known Sharpe ratio. However, there is one main difference: the variance of bank $j$'s project returns can be controlled by $j$ itself. Nonetheless, notice that the optimal investment $\hat w_{ij}$ grows with this `Sharpe-like' ratio. 
In particular, note that if $\frac{\mu_j}{\barfix[0.2ex]{\phi_j\theta_j \bar F_j(\hat c_j)}} < 1$, bank $i$ would prefer to short project $j$. Since this is not permitted in our model, bank $i$ resorts to an investment of zero. 
As a direct result, notice that network's sparsity structure is dictated by this quantity -- a project $j$ has external investors if and only if $\frac{\mu_j}{\barfix[0.2ex]{\phi_j\theta_j \bar F_j(\hat c_j)}} > 1$. This implies a `core-periphery' network structure, such that \change{only a subset of banks' projects are invested in} -- an example of such a financial network can be seen in Figure~\ref{fig:dec_nwk}.
\begin{figure}[htbp]
	\centering
	\includegraphics[width = 0.5\linewidth]{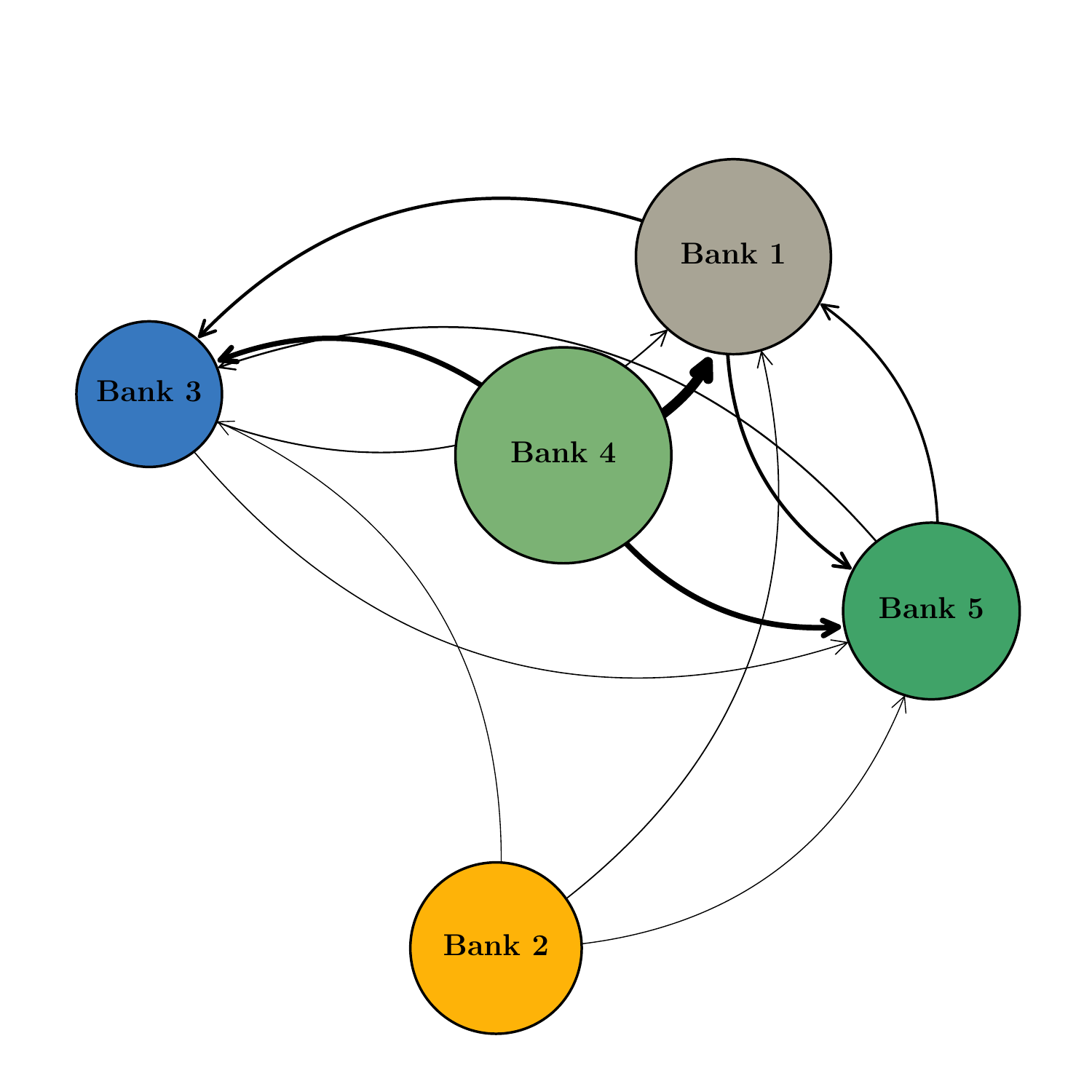}
	\captionsetup{width=.8\linewidth, font = small, justification=justified}
	\caption{Sample financial network generated by the decentralized optimum. Edges point from investing banks to projects, and their width indicates the nominal size of the exposures. Node size indicates to total capitalization.}
	\label{fig:dec_nwk}
\end{figure}
\begin{table}[h]
	\centering
	\captionsetup{width=.8\linewidth, font = small, justification=justified}
	\caption{Parameters for the financial network in Figure~\ref{fig:dec_nwk}. The risk-free rate is $r = 5\%$. Relevant code can be found \href{https://github.com/drigobon/rigobon-sircar-2022}{here}.}
	\label{tab:dec_nwk_params}
	\begin{tabular}{ c | c c c c c c }
		\toprule
		Bank & $\mu_i$ (\%) & $\phi_i$ & $\eta_i$ &  $\theta_i$ & $\gamma_i$ & $\bar F_i(x)$ \\
		\midrule
		1 & 0.9 & 0.2  & 0.5  & 0.04 & 0.5 & $e^{-0.5 x}$\\  
		2 & 1 &  0.3  &  0.6  &  0.08 & 1.7  & $e^{-0.6 x}$\\
		3 & 1.5 & 0.9  & 0.7  &  0.12 &  1 & $e^{-0.7 x}$\\
		4 & 1.3 & 0.6  & 0.4  & 0.05 & 0.3  & $e^{-2 x}$\\
		5 & 1.3 & 0.82 &  0.9  &  0.02 &  0.87 & $e^{-2.4 x}$\\
		\bottomrule
	\end{tabular}
\end{table}

\change{
	The resulting `core-periphery' structure reflects the fact that in equilibrium, only a subset of banks' projects have expected \textit{excess} return greater than zero. The expected return associated with the jump term in Eq.~\eqref{eq:risky_asset_sde} is equal to $-\phi_i \theta_i \bar F_i(\hat c_i)$, whereas the drift is $\mu_i + r$. Hence, the condition $\frac{\mu_i}{\barfix[0.2ex]{\phi_i\theta_i \bar F_i(\hat c_i)}} > 1$ is equivalent to $\mu_i + r -\phi_i \theta_i \bar F_i(\hat c_i) > r$, which implies that only banks whose projects have expected return greater than $r$ will attract external investors. This is an intuitive consequence of the model, as otherwise investors would not be rewarded for their exposure to the riskiness of these projects. Although the core-periphery structure is determined only by the expected return of a bank's project, we note that both the volatility and return of a project are impacted through a bank's cash reserves. Namely, the optimal $w_{ji}$ decreases in $\phi_i \theta_i \bar F_i(\hat c_i)$ due to the effects of \textit{both} lower expected return and increased volatility. 
	
	A key feature driving the core-periphery structure is the fact that investing banks are identical up to their risk aversion level, which affects their decision of how much capital is allocated to a project. However, whether or not bank $i$ chooses to invest in another project $j$ depends \textit{only} on the expected profitability of project $j$. This reflects a system in which banks operate as `price-takers' in other projects, without the ability to alter the qualities of their investment opportunities. 
}

\subsection{Centralized Network}
\label{ssec:opts-cen}

Consider now the perspective of a single social planner of the financial system. In contrast with Section~\ref{ssec:opts-dec}, we will see that the planner has two different incentives for maintaining bank $i$'s \change{cash reserves}. The first is identical -- bank $i$ stands to lose wealth \change{when they suffer a liquidity shortage}. The second incentive is concerned with the entire financial system -- other banks face losses to their \change{investments} on the very same event. Therefore, we expect the planner to have stronger incentives to hold \change{cash reserves}, and elect for a greater supply of liquidity throughout the system.

We assume that the planner seeks to maximize social welfare in the system -- defined as the sum of all banks' utilities. The planner's value function is therefore given by the following:
\begin{equation}
	\label{eq:cent_val}
	V(t,x_1,\dots,x_n) = \sup_{(c^\cdot_\cdot, w_\cdot^{\cdot \cdot}) \in \A_{t,T}} \mathbb{E} \left[ \sum_{i=1}^n U_i(X_T^i) \Bigg | (X^1_t,\dots X^n_t) = (x_1,\dots x_n)\right],
\end{equation}
where $\A_{t,T} = \prod_{i} \A^i_{t,T}$ is the Cartesian product of each bank's admissible controls.

\begin{remark}
	\label{rmk:swf}
	It is important to note that there are many possible social welfare functions for the planner to consider. In this section, we will see that using the sum of utilities allows for separable solutions to the value function when all banks have logarithmic utility, i.e. $\gamma_i = 1$ for all $i$. We note that if the planner maximized the \emph{product} of utilities, then we can also find an explicit solution and optimal controls in the case where $\gamma_i \in (0,1)$ for all $i$, but we omit these calculations for conciseness.
\end{remark}

We can relate the planner's value function to those of individual banks in~\eqref{eq:dec_val}. The optimal decentralized allocation from Section~\ref{ssec:opts-dec} is always feasible for the planner, and therefore the planner's value function is bounded from below by the sum of each bank's value function as follows:
\begin{equation}
	\label{eq:val_gap}
	V(t,x_1,\dots, x_n) \ge \sum_{i=1}^n V_i(t,x_i).
\end{equation}
This inequality reflects \change{a potential} inefficiency of the decentralized setting; the planner's allocation is the first-best (i.e. welfare-maximizing) outcome for the system, and decentralized banks may not achieve this outcome. \change{However, it is not yet clear that the inequality in Eq~\eqref{eq:val_gap} is strict. In Section~\ref{sec:diff}, we will show that in most cases the decentralized optimum is not welfare-maximizing, and establish technical conditions under which we can explicitly compute the gap in welfare.}

In what follows, we analyze the planner's optimal allocation by deriving the dynamic programming equation and analyzing the resulting optimization problem. As in the previous section, we first derive the non-local PDE solved by the planner's value function.

\begin{prop}
	\label{prop:cent_opt_hjb}
	If there exist optimal controls, and the value function in~\eqref{eq:cent_val} is \\
	$\mathcal{C}^{1,1,\dots, 1}([0,T),\mathbb{R}_+,\dots,\mathbb{R}_+)$, then it solves
	\begin{equation}
		\label{eq:cent_hjb}
		\begin{alignedat}{2}
			0 &= \partial_t V  + \sup_{c_\cdot,  w_{\cdot\cdot}} \Bigg\{ \sum_{i=1}^n \Bigg( \left[\left(1-c_i\right)r + \sum_{j\neq i}w_{ij}\mu_j + \frac{\eta_i \mu_i}{\phi_i}\right] x_i \partial_{x_i} V 
			\\
			&  \quad + \theta_i \bar F_i(c_i)\Big[ V(t,x_1(1-\phi_i w_{1i}),\dots, 	x_i(1-\eta_i),\dots, x_n(1-\phi_i w_{ni})) - V \Big] \Bigg)\Bigg\},
		\end{alignedat}	
	\end{equation}
	with terminal condition $V(T,x_1,\dots,x_n) = \sum_{i=1}^n U_i(x_i)$. Where unspecified, the value function and its derivatives are evaluated at $(t,x_1,\dots,x_n)$.
\end{prop}

The proof of this result can again be found in Appendix~\ref{ssec:pfs-cent}. 
Proposition~\ref{prop:cent_opt_hjb} yields an $n+1$ dimensional non-local PDE for the planner's value function. There is one key difference between Equations~\eqref{eq:cent_hjb} and~\eqref{eq:dec_hjb} -- when a \change{ bank suffers a liquidity shortage}, the planner's value function is affected by losses occurring throughout the \emph{entire} financial system. This is not true in the decentralized setting; an individual bank's value function only depends on their own losses caused by such an event.

With specific choices of utility functions, it is possible to find a separable solution to~\eqref{eq:cent_hjb}, and prove existence of an optimal allocation. However, in order to establish uniqueness, we will need the following technical assumption.

\begin{assumption}
	\label{assn:cent_uniq}
	Let each shock density $f_i(\cdot)$ satisfy
	\begin{equation}
		\label{eq:cent_f_cond}
		\frac{f_i(c)}{\bar F_i(c)} + 3\frac{f_i'(c)}{f_i(c)} - \frac{f_i''(c)}{f_i'(c)} < 0, \ \forall c\ge 0.
	\end{equation}
	and, with the notation $\tilde c_i = F_i^{-1} \left(\left[1 - \frac{\mu_i}{\phi_i \theta_i}\right]_+\right)$, assume that the following holds for all $i$
	\begin{equation}
		\Gamma(\eta_i;1) > 
		\begin{cases}
			\min\left\{ (n-1) \left[ \log\left(\frac{\phi_i \theta_i}{\mu_i}\right) - \frac{f_i(0)^2}{f_i'(0)}\right] , \ \frac{r}{\theta_i f_i(0)} + (n-1) \log\left(\frac{\phi_i \theta_i}{\mu_i}\right) \right\} & \text{ if } \tilde c_i = 0 \\[2ex]
			\min\left\{- (n-1) \frac{\phi_i \theta_i f_i(\tilde c_i)^2}{\mu_i f_i'(\tilde c_i)} , \ \frac{r}{\theta_i f_i(\tilde c_i)}  \right\} & \text{ otherwise,}
		\end{cases}
		\label{eq:cent_gamma_cond}
	\end{equation}
	\change{where recall from Eq.~\eqref{eq:gamma} that $\Gamma(\eta_i;1) = \log(1)-\log(1-\eta_i) = -\log(1-\eta_i)$ represents the loss in utility associated with a loss of $\eta_i$ fraction of wealth.}
\end{assumption}

Assumption~\ref{assn:cent_uniq} \change{provides sufficient -- but not necessary -- conditions} for uniqueness of the planner's optimal allocation. Numerically, we have observed that the optimal solution is often unique, but the optimization problem in~\eqref{eq:cent_hjb} is (generally) not convex, and therefore proving uniqueness is non-trivial. We do, however, note that the inequality~\eqref{eq:cent_f_cond} is always satisfied by exponential and power distributions. \change{At a glance, Eq.~\eqref{eq:cent_gamma_cond} requires that each bank's losses upon a liquidity shortage not be too small. If bank $i$ loses only a small amount of capital when a liquidity shortage occurs, then it is possible for there to be two equilibria for the planner: i) project $i$ is profitable and bank $i$'s cash reserves are high; and ii) project $i$ is not profitable and bank $i$'s cash reserves are low. Since the same rationale holds for all $n$ projects, we could have up to $2^n$ distinct equilibria. In this paper, we will focus on settings where unique optimal controls can be provably obtained, and hence require Assumption~\ref{assn:cent_uniq} to hold.}

Analogous to Section~\ref{ssec:opts-dec}, we show there exists a separable solution to the PDE~\eqref{eq:cent_hjb}. Additionally, we show that the optimal solution will solve a system of algebraic equations.

\begin{prop} 
	\label{prop:cent_hjb_soln}
	Let each bank have a logarithmic utility function (i.e. $\gamma_i = 1 \ \forall i$). Then, there exist optimal cash \change{reserves} and \change{project investment} amounts for the planner, which solve the following system of equations:
	\begin{equation}
		\label{sol:cent}
		\begin{alignedat}{2}
			c_i^* &= 
				\begin{cases}
				f_i^{-1}\left(\frac{r}{\theta_i \left[ \Gamma(\eta_i;1) + (n-1) \Gamma(\phi_i w_{\cdot i}^*;1)\right]}\right) &\text{ if } f_i(0) \le \frac{r}{\theta_i \left[ \Gamma(\eta_i;1) + (n-1) \Gamma(\phi_i w_{\cdot i}^*;1)\right]} \\
				0 & \text{ otherwise,}
				\end{cases}  
				\\
			w_{\cdot i}^* &= 
				\begin{cases}
				\frac{1}{\phi_i}\left(1-\frac{\phi_i\theta_i 			\bar F_i(c_i^*)}{\mu_i}\right) &\text{ if } \frac{\phi_i \theta_i\bar F_i(c_i^*)}{\mu_i} < 1\\
				0 & \text{ otherwise.}
				\end{cases}
		\end{alignedat}  
	\end{equation}
	We define 
	\begin{equation}
		\label{eq:cent_opt}
		J^*_C =\sum_{i=1}^n \Bigg( \left[\left(1-c_i^*\right)r + (n-1)w_{\cdot i}^*\mu_i + \frac{\eta_i \mu_i}{\phi_i} \right]   - \theta_i \bar F_i(c_i^*)\Big[ \Gamma(\eta_i;1) + (n-1) \Gamma(\phi_iw_{\cdot i}^*;1) \Big] \Bigg),
	\end{equation}
	and $g(t) = (T-t)J^*_C$. The solution to~\eqref{eq:cent_hjb} is given by
	\begin{equation}
		\label{eq:cent_sol}
		V(t,x_1,\dots,x_n) = g(t) + \sum_{i=1}^n \log x_i.
	\end{equation}	
	Furthermore, under Assumption~\ref{assn:cent_uniq}, the optimal controls $(c_i^*,w_{\cdot i}^*)$ are unique.	
\end{prop}

The proof is again given in Appendix~\ref{ssec:pfs-cent}. We note that a separable solution using logarithmic utility functions is only possible because the planner aims to maximize the sum of banks' expected utilities. See Remark~\ref{rmk:swf} for a brief discussion of other settings where a separable solution can be obtained.

In contrast to the decentralized setting, the maximization in~\eqref{eq:cent_hjb} is not additively separable between each optimization variable. Nonetheless, each of the $i$ subsets $\{c_i, w_{1i}, \dots, w_{ni}\}, \ i = 1,\dots, n$ can be analyzed separately, which greatly simplifies our analysis. However, the coupling between $c_i$ and $w_{\cdot i}$ leads to the need for additional assumptions to establish uniqueness.

The system of equations in~\eqref{sol:cent} admits a block coordinate descent approach. Namely, for any fixed $c_i$, the maximization problem for $w_{\cdot i}$ is strictly concave and admits a unique solution (these can be seen in the proof of Proposition~\ref{prop:cent_hjb_soln}). Conversely, for given values of $w_{\cdot i}$, the maximizing of $c_{i}$ shares these features. As a result, we can iteratively update these variables to solve for the planner's optimum numerically. Upon convergence, we are guaranteed to have found the uniquely optimal allocation.

Since we have shown existence of an optimal allocation, and the proposed solution in~\eqref{eq:cent_sol} is continuously differentiable, then we are able to verify that it is indeed equal to the planner's value function.
\begin{corollary}
	\label{cor:cent_verification}
	The planner's value function in~\eqref{eq:cent_val} is given by~\eqref{eq:cent_sol}. Furthermore, the optimal \change{ project investment and cash reserves}  solve~\eqref{sol:dec}.
\end{corollary}
\change{As with previous results, the proof is presented in Appendix \ref{ssec:pfs-cent}.}

\subsubsection*{Analysis of Centralized Equilibrium}

There is one main difference between the system of equations in~\eqref{sol:cent} and the optimal solutions from the decentralized setting in~\eqref{sol:dec}. Here, the expression for the planner's optimal cash reserves $c_i^*$ contains an additional term of $(n-1)\Gamma(\phi_i w_{\cdot i}^*;1)$. This term directly captures the externality -- when \change{bank $i$ suffers a liquidity shortage and their project incurs losses}, the planner is congnizant of losses in utility experienced by all banks. As a result, with more banks the planner maintains larger \change{cash reserves} to compensate for greater system-wide losses. In contrast, bank $i$'s decentralized optimization problem considers only changes to their own wealth, and therefore their optimal $\hat c_i$ will be indifferent to the system's size.

Since we will have $w_{\cdot i}^* \ge 0$ in~\eqref{sol:cent}, the planner has no weaker an incentive to hold liquidity than an individual bank.\footnote{This observation may not be the case if, for example, short-selling were allowed. Qualitatively, the planner may choose to have a single bank $i$ hold zero cash, while others in the system maintain large, short positions in $i$'s project. In this case, the total utility of the system may actually increase when bank $i$'s project incurs losses. However, clearly this result may not align with the best outcome for bank $i$ itself.} Hence, the planner will hold \change{no smaller cash reserves than banks in the decentralized setting} -- this difference will be studied more closely in the following section. Finally, we also notice that given some fixed \change{cash reserves}, the optimal project investments $w^*_{\cdot\cdot}$ and $\hat w_{\cdot \cdot}$ are computed identically in both settings. It follows that any differences between the optimal \change{project investments} in~\eqref{sol:dec} and~\eqref{sol:cent} must be driven only by differences in optimal cash holdings, \change{which we proceed to analyze in Section~\ref{sec:diff}.}

\change{Finally, we note that \textit{only} in the case where $w^*_{\cdot i}$ equal zero for all $i$ will the planner's optimal allocation coincide with that of the decentralized system. In such a case, bank $i$'s cash reserves are dictated only by their own losses, which implies that $\hat c_i = c_i^*$. In other words, as long as at there is least one project within the planner's optimum that has positive expected excess return, the market-mediated equilibrium from Section~\ref{ssec:opts-dec} is inefficient, and the inequality in Eq.~\eqref{eq:val_gap} is strict. This suggests that there is value to be obtained from regulating the system.}

\section{\change{Differences in Equilibria}}
\label{sec:diff}

It is natural to compare the two optimal allocations from Sections~\ref{ssec:opts-dec} and~\ref{ssec:opts-cen}. In particular, we may be interested in computing the gap in welfare from the inequality~\eqref{eq:val_gap}. Figure~\ref{fig:dec_vs_cent} illustrates a sample path for the wealth of three banks, where in Fig.~\ref{fig:dec_wealth} the controls are given by~\eqref{sol:dec}, and in Fig.~\ref{fig:cent_wealth} by~\eqref{sol:cent}. Qualitatively, there are higher-frequency jumps in~\ref{fig:dec_wealth}, but the jumps are of larger size in~\ref{fig:cent_wealth}. With the remainder of this section, we study these differences more precisely. 
\begin{figure}[htbp]
	\centering
	\begin{subfigure}{0.4\linewidth}
		\includegraphics[width = \linewidth]{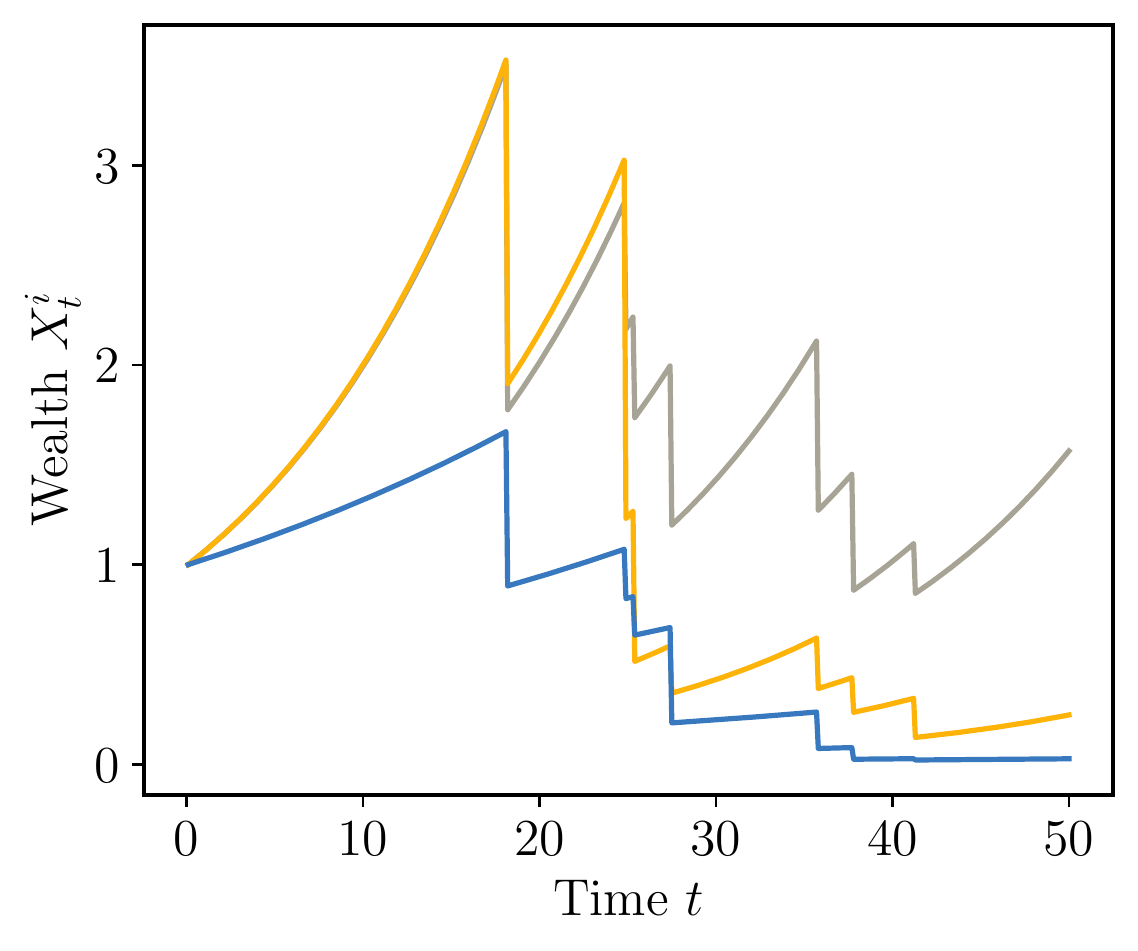}
		\captionsetup{font = small}
		\caption{Decentralized}
		\label{fig:dec_wealth}
	\end{subfigure}
	\hspace{2em}
	\begin{subfigure}{0.4\linewidth}
		\includegraphics[width = \linewidth]{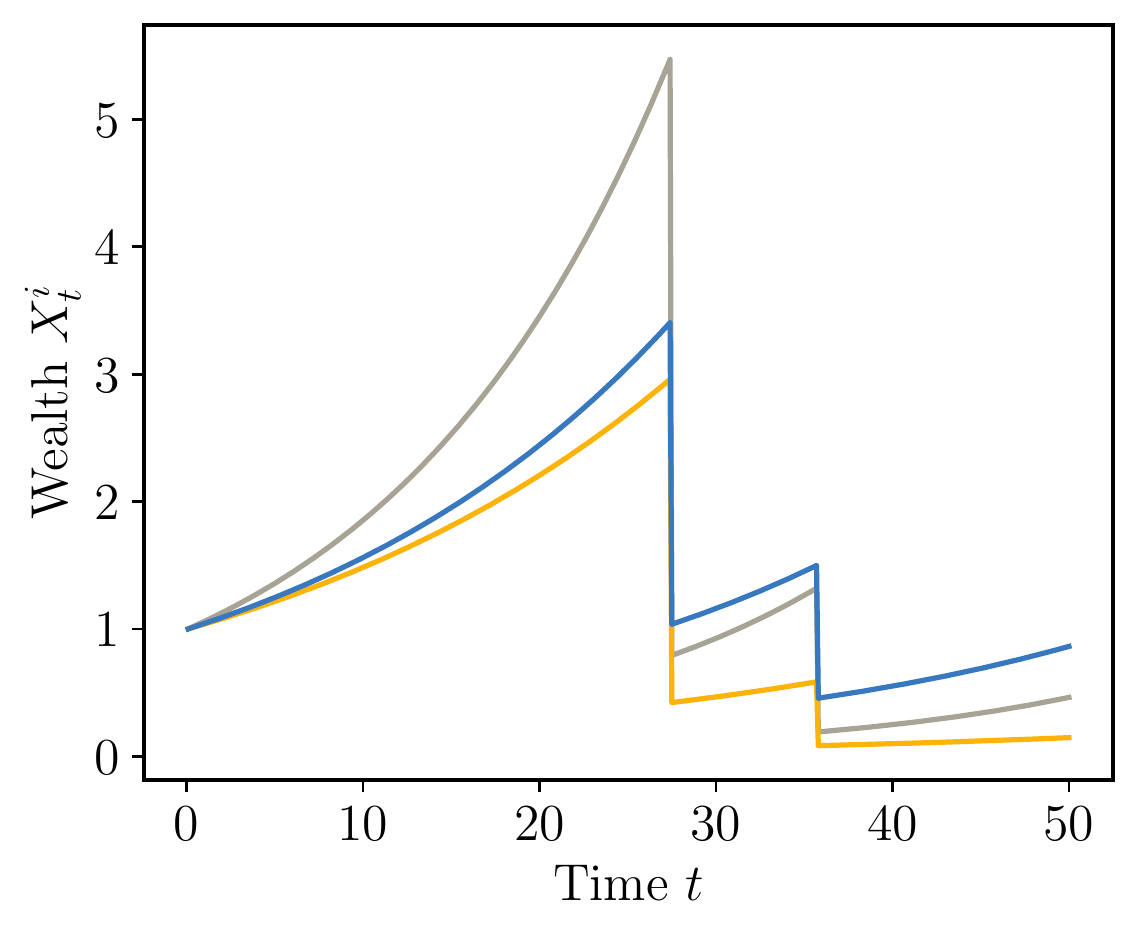}
		\captionsetup{font = small}
		\caption{Centralized}
		\label{fig:cent_wealth}
	\end{subfigure}
	\captionsetup{width=.8\linewidth, font = small, justification=justified}
	\caption{An example of wealth dynamics under the both optimal allocations for a system of $n=3$ banks. The same random seed is used in both simulations, so that the size and arrival times of liquidity shocks are identical. For conciseness, we do not include the parameters, but the code to reproduce these figures can be found \href{https://github.com/drigobon/rigobon-sircar-2022}{here}.}
	\label{fig:dec_vs_cent}
\end{figure}

In what follows, we will assume that all banks have logarithmic utility (i.e. $\gamma_i = 1$ for all $i$). Recall that $\hat c_i, \hat w_{ji}$ denote the optimal decentralized allocations given in~\eqref{sol:dec}. Note that for all $j,k\neq i$ we will have $\hat w_{ki} = \hat w_{ji}$, so we will denote these fractional amounts to be $\hat w_{\cdot i}$ (this follows from $\gamma_j = 1$ for all $j$). Additionally, recall that $c_i^*, w_{\cdot i}^*$ denotes the optimal solution from~\eqref{sol:cent}. Finally, we use the asymptotic notation $g(n) = \Theta(h(n))$ to denote that there exist positive constants $A_1,\, A_2$ such that $A_1 \le \lim_{n\to\infty} \frac{g(n)}{h(n)} \le A_2 $. If $A_1 = A_2$, then we will write $g(n) \asymp h(n)$.

Comparing the two optimal allocations, since $w_{\cdot i}^* \ge 0$, it will necessarily be the case that $c_i^* \ge \hat c_i$. Our first core result establishes the asymptotic rate at which the planner's optimal \change{cash reserves} grow as the size of the system increases, \change{when liquidity shocks follow an exponential distribution.} More precisely, we show that the planner's cash holdings must grow at logarithmically in the system size $n$. In contrast, if $w_{\cdot i}^* = 0$, then we would have $c_i^* = \hat c_i$, which is of constant order.

\begin{prop}
	\label{prop:PoA_exp_exact_c}
	If $F_i(x) = 1-e^{-\frac{x}{\lambda_i }}$, then 
	$$\lambda_i \log\left(\frac{\theta_i (n-1)}{\lambda_i r} \left[ \log(n-1) - \log\left(\frac{\Gamma(\eta_i;1)}{\Gamma(\phi_i \hat w_{\cdot i};1)}\right) \right] \right) \le c_i^* \le \lambda_i \log\left(\frac{\theta_i C_U (n-1)}{\lambda_i r}\log(n)\right),$$
	where $\Gamma(\cdot\, ;1)$ is defined in Eq.~\eqref{eq:gamma}.
	
	In particular, $c_i^* \asymp \lambda_i \log(n).$
\end{prop}

The proof follows from Lemma~\ref{lem:PoA_exp_c} in Appendix~\ref{sec:proofs}, which establishes upper and lower bounds on super- and sub-exponentially distributed liquidity shocks.

\change{This result highlights a substantial discrepancy between the planner's optimum and that of individual banks. Since the losses to an individual bank are not affected by the number of external investors, their cash reserves are dictated only by their own losses and are of constant order. In contrast, the planner observes losses throughout the system, which grow at least linearly in the system size, and hence chooses to hold greater cash reserves to compensate for the greater losses upon a liquidity shortage to any particular bank.

Proposition~\ref{prop:PoA_exp_exact_c} is a useful tool for comparing the two optimal allocations, as all differences are driven by the distinct cash reserves. The following result compares the asymptotic rates for various quantities in both the centralized and decentralized equilibrium.

\begin{corollary}
	\label{cor:asymp_comparison}
	Under the conditions of Proposition~\ref{prop:PoA_exp_exact_c}, and when $\hat w_{\cdot i} > 0$, we have:
	\begin{enumerate}[(i)]
		\item \textbf{Likelihood of Liquidity Shortage:} 
		\begin{equation}
			\bar F_i(c_i^*) = \Theta\left(\frac{1}{n\log(n)}\right),
			\qquad
			\bar F_i(\hat c_i) = \Theta(1).
		\end{equation}
		
		\item \textbf{Optimal Project Investment:}
		\begin{equation}
			w_{\cdot i}^* = \frac{1}{\phi_i} - \Theta\left(\frac{1}{n\log(n)} \right),
			\qquad 
			\hat w_{\cdot i} = \Theta(1).
		\end{equation}
		
		\item \textbf{Expected System-wide Loss in Utility (due to bank $i$):}
		\begin{equation}
			\bar F_i(c_i^*) (n-1) \Gamma(w_{\cdot i}^*;1) = \Theta(1),
			\qquad
			\bar F_i(\hat c_i) (n-1) \Gamma(\hat w_{\cdot i};1) = \Theta(n).
		\end{equation}

	\end{enumerate}
\end{corollary}

The proofs require only plugging in the asymptotic rate for $c^*_i$ from Proposition~\ref{prop:PoA_exp_exact_c} into $\bar F_i(\cdot)$, the expression for $w_{\cdot i}^*$, and $\Gamma(w_{\cdot i};1)$.}

Notice that we must have $w_{\cdot i} < \phi_i^{-1}$ to ensure wealth remains positive, yet we can still pin down the rate at which the \change{optimal project investments approach their} upper bound. \change{The term $\bar F_i(c_i^*) (n-1) \Gamma(w_{\cdot i}^*;1)$ appears in Eq.~\eqref{eq:cent_opt}, the value function for the planner. It represents the expected loss in utility to all of bank $i$'s external investors when $i$ suffers a liquidity shortage: $\Gamma(w_{\cdot i}^*;1)$ captures the loss in utility to a single investing bank, of which there are $n-1$ within the system, and $\bar F_i(c_i^*)$ is proportional to the likelihood of such an event. It is particularly interesting that this quantity is of constant order, as it shows that the planner perfectly compensates for larger possible losses in utility in the system through its reduction of the probability of such an event.}

The results in Corollary~\ref{cor:asymp_comparison} allow us to analyze differences in welfare between the two settings in the following section.

\subsection{Welfare Gap}
\label{ssec:welfare_gap}

We now turn to the gap between value functions from~\eqref{eq:val_gap}. It will be useful to have $\M_n$ denote the set of banks \change{whose projects are invested in by other firms} in the planner’s optimal allocation, i.e. $\M_n = \left\{i \in \{1,\dots,n\} : w_{\cdot i}^* > 0\right\}$. Banks in $\M_n$ form the `core' of the financial network. If for some $i$ we have $w_{\cdot i}^* = 0$, then it must be the case that $c_i^* = \hat c_i$ and $\frac{\phi_i\theta_i \bar F_i(c_i^*)}{\mu_i} > 1$. For such a bank $i$, the planner’s optimal $c_i^*$ would remain constant at $\hat c_i$, even as $n$ grows.

In this model, we define \change{the welfare ratio as:} 
$$\mathrm{WR} = \frac{V}{\sum_{i=1}^n V_i}.$$
\change{This quantity is occasionally referred to as} the `price of anarchy', and reflects how greedy decentralized behavior leads to lesser welfare in the system~\citep{papadimitriou2001algorithms}. In the following result, we characterize its asymptotic behavior.

\begin{prop}
	\label{prop:PoA_order}
	Assume that $\gamma_i = 1$ and $F_i(x) = 1 - e^{-\frac{x}{\lambda_i}}$ for all $i$. Then, as $n\to\infty$, we have
	
	$$\mathrm{\change{ WR}} = \Theta(1).$$
\end{prop}

The proof is found in Appendix~\ref{ssec:pfs-diff}, and uses the results from Corollary~\ref{cor:asymp_comparison}.

It is particularly interesting that the \change{welfare ratio} does not grow with the system size $n$, nor the remaining time horizon $(T-t)$. A more precise result can be obtained if banks are sufficiently homogeneous, where we can compute the limiting value.

\begin{corollary}
	\label{cor:PoA_limit}
	\change{In the setting of Proposition~\ref{prop:PoA_order},} assume further that all banks in $\M_n$ are identical (i.e. $\mu_j = \mu$, $\phi_j = \phi$, $\theta_j = \theta$, $\eta_j = \eta$, and $\lambda_j = \lambda$ for some given constants $\mu, \phi, \theta, \eta$ and $\lambda$). If $|\M_n| \underset{n\to \infty}{\to} \infty$, then 
	\begin{alignat}{3}
		\frac{V_i}{|\M_n|(T-t)} &\underset{n\to \infty}{\to} \frac{\mu}{\phi} + \theta \bar F(\hat c))\left[\log\left(\frac{\phi \theta \bar F(\hat c)}{\mu}\right) - 1\right], \quad && \forall i=1,\dots, n \\
		\frac{V}{n|\M_n|(T-t)} &\underset{n\to \infty}{\to} \frac{\mu}{\phi}. &&
	\end{alignat}
	where $\hat c$ is given in~\eqref{sol:dec} and $\bar F(\hat c) = e^{-\frac{\hat c}{\lambda}}$. As a result, we have:
	\begin{equation}
		\label{eq:PoA_limit}
		\mathrm{\change{WR}} \underset{n\to \infty}{\to} \frac{1}{1 + \barfix{\frac{\phi \theta \bar F(\hat c)}{\mu}} \left[\log\left(\barfix{\frac{\phi \theta \bar F(\hat c)}{\mu}}\right) - 1\right]}.
	\end{equation}
\end{corollary}

Corollary~\ref{cor:PoA_limit} verifies that the \change{welfare ratio} is of constant order in $n$, and the proof is found in Appendix~\ref{ssec:pfs-diff}. Of particular interest, the rate at which $|\M_n|$ grows in $n$ does not appear in our result. This implies that the limiting \change{welfare ratio} is independent from the fraction of the system that operates as its `core'. Notice also that $\phi \theta \bar F(\hat c) < \mu$, and hence the right-hand side in~\eqref{eq:PoA_limit} is greater than one. Moreover, the limiting \change{welfare ratio} is increasing in $\frac{\phi \theta \bar F(\hat c)}{\mu}$. Therefore, as the profitability of projects in the decentralized setting is reduced, the limiting \change{welfare ratio} grows to infinity.

Corollary~\ref{cor:PoA_limit} is verified numerically. Using the parameters in Table~\ref{tab:PoA_sim_params}, we compute the individual and collective value functions. The resulting \change{welfare ratio} is plotted in Figure~\ref{fig:PoA_sim}, along with the limiting value in~\eqref{eq:PoA_limit}. We see that this quantity quickly converges to its limit.
\begin{figure}[htbp]
	\centering
	\includegraphics[width=0.5\textwidth]{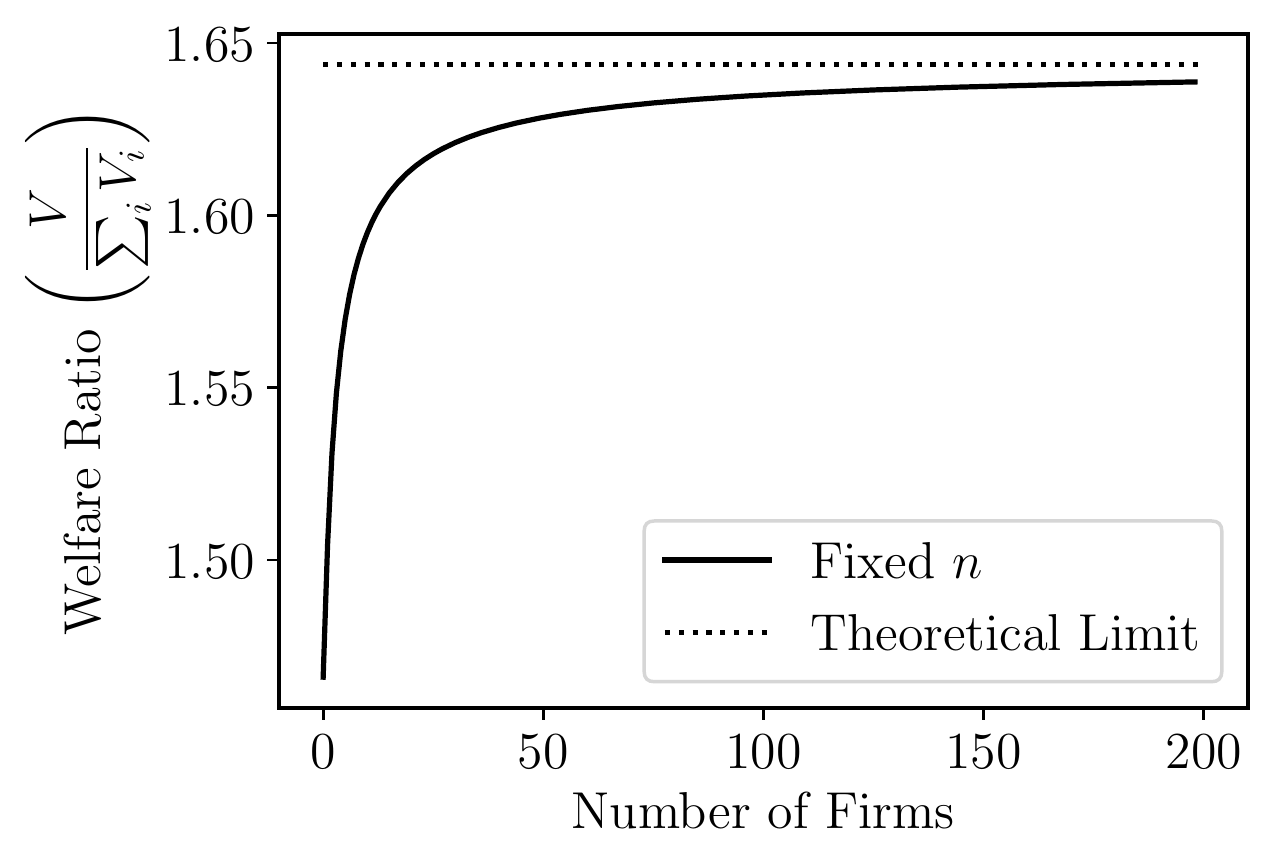}
	\captionsetup{width=.8\linewidth, font = small, justification=justified}
	\caption{Simulating the Welfare Ratio for a system of identical firms as $n$ grows. Uses parameters shown in Table~\ref{tab:PoA_sim_params}.}
	\label{fig:PoA_sim}
\end{figure}
\begin{table}[h]
	\centering
	\captionsetup{width=.8\linewidth, font = small, justification=justified}
	\caption{System parameters used in simulations for Figures~\ref{fig:PoA_sim} and~\ref{fig:Liq_Loss_Shock_Arrival}. Relevant code is available \href{https://github.com/drigobon/rigobon-sircar-2022}{here}.}
	\label{tab:PoA_sim_params}
	\begin{tabular}{ c  c  l }
		\toprule
		Notation & Value & Description \\
		\midrule
		$r$ & 0.01 & Risk-free rate \\
		$\mu$ & 0.045 & Excess drift\\
		$\phi$ & 0.4 & Losses to External Investors\\
		$\eta$ & 0.5 & Losses to Associated Bank\\
		$F(x)$ & $1 - e^{-\frac{x}{\lambda}}$ & CDF of shock size\\
		$\lambda$ & 1 & Parameter of $F(\cdot)$ \\
		$\theta$ & 0.1 & Shock arrival rate\\
		$\gamma$ & 1 & Relative risk aversion coefficient\\
		\bottomrule
	\end{tabular}
\end{table}

\subsection{Replicating the Centralized Allocation}
\label{ssec:res-replication}

Finally, we may be interested in studying how banks in the decentralized setting can be incentivized or forced to replicate the planner's optimal allocation. \change{We could approach this regulatory question in one of two ways -- either directly enforcing a liquidity requirement on individual banks (corresponding to the planner's optimal $c_i^*$), or allowing regulators to control the previously exogenous parameter $\eta_i$, which dictates how much of bank $i$'s wealth must be invested in their own project. In this section, we choose to pursue the latter approach.} The rationale for this is twofold: first, $\eta_i$ plays the fundamental role in decentralized banks' choice \change{of optimal cash reserves. From the analysis near the end of Section~\ref{ssec:opts-dec}, we observed that by increasing $\eta_i$, banks will increase their cash holdings and therefore reduce their projects' riskiness}, which can lead them to eventually match the planner's optimal liquidity reserves. Second, we can imagine that \change{externally investing} banks are permitted to write a contract stipulating the associated bank's degree of co-investment. This kind of contracting is not a focus of our paper, but is instead analyzed in more detail with Principal-Agent problems such as~\citet{hernandez2020bank}. Nonetheless, the co-investment contract can be designed to ensure that individual banks hold sufficient \change{ reserves of cash}.

\change{There are several other strategies used by regulators in practice to reduce inefficiencies in the financial system, such as bail-in or bail-out policies, or government intervention~\citep{bernard2022bail, kanik2019lombard}. While these strategies are fruitful directions for future work, we will focus on how co-investment requirements can be used to reduce the riskiness of core banks in the system, since these parameters are already integrated into the model and directly imply equilibirum levels of liquidity reserves.}

Let $\eta_i^C$ (respectively $\eta_i^D$) denote the fraction of bank $i$'s wealth lost upon liquidity shortage in the centralized (resp. decentralized) setting. We would like to choose $\eta_i^D$ so that decentralized banks replicate the planner's optimum with given values $\eta_i^C$. More precisely, we seek to find $\eta_i^D$ solving $c_i^*(\eta_i^C) = \hat c_i(\eta_i^D)$ for all $i$, where we write the optimal controls in a way that highlights their dependence on the underlying values of $\eta$. Even though the optimal allocations are identical, however, we note that the decentralized optimum is still inefficient (with respect to the optimal centralized allocation corresponding to $\eta_i^D$). Using equations~\eqref{sol:dec} and~\eqref{sol:cent}, we find that:

$$\eta_i^D = 1 - (1 -\eta_i^C)\left(1 - \phi_i w_{\cdot i}^*(\eta_i^C)\right)^{n-1}.$$

First, notice that whenever $w_{\cdot i}^*(\eta_i^C) > 0$, the resulting value of $\eta_i^D$ will grow with $n$ towards its upper bound of $1$. This is intuitive -- banks whose projects are highly invested in require the strongest incentive to reduce their project's riskiness. \change{It is precisely these banks that comprise the `core' of the financial network, and will be subjected to the strictest regulations. In contrast, banks in the periphery (with $w_{\cdot i}(\eta_i^C) = 0$) are no more strictly regulated ($\eta_i^D = \eta_i^C$), since their liquidity shortages cause only localized losses.}
Second, we see that for banks to replicate the planner's optimum, it is necessary for bank $i$'s degree of co-investment to depend on their liabilities throughout the system. It is therefore necessary to know the complete structure of the financial network to determine the value of $\eta_i^D$, which may not be known to individual \change{investing banks}. 
Finally, an interesting case occurs when we choose $\eta_i^C = 0$. In this case, the value of $\eta_i^D$ is only non-zero if $w_{\cdot i}^*(0) > 0$. Namely, banks \change{in the periphery would not be required to hold any} stake in their own projects.

\change{This discrepancy suggests that while stricter regulations on core banks can improve the efficiency of the financial system, there is little if any value to be gained by increasing restrictions on banks in the periphery. When increasing the self-investment requirement of some bank in the periphery, one of two will occur: i) this increase causes $i$ to have large enough cash reserves for other banks in the system to invest in project $i$, or ii) the increase in $\eta_i$ is not sufficiently large for $i$'s project to be profitable for other banks in the system. In case i), bank $i$ moves into the `core' of the financial system, as their project is now capitalized by other firms. In this case, it is not possible to conclusively state whether or not this increase in $\eta_i$ yielded higher or lower social welfare. However, in case ii), the small increase to $\eta_i$ serves only to push bank $i$'s capital towards less efficient investments -- they are obligated to hold additional capital in a project whose expected return is less than $r$, and they choose to hold greater cash reserves to compensate for this heightened exposure -- which also provides lower returns than $r$. In this case, both effects serve only to reduce $i$'s wealth and hence the total welfare in the financial system.}

\section{Conclusion}
\label{sec:disc}

In this paper, we present a model in which banks in a financial system control both their own levels of cash reserves, and their investment in each others' risky projects. 
We compute the unique optimal allocations of capital for two distinct organizations of the system, and study their differences qualitatively and quantitatively. First, we analyze the setting where each bank acts with pure self-interest. We compute explicitly the optimal allocation, and find that the size of project investments are closely related to a Sharpe-like ratio -- which is controlled by \change{the bank associated with a particular project}. In particular, the equilibrium financial network exhibits a `core-periphery' structure, in which only a subset of banks serve as \change{valid investment targets}. Second, we formulate the optimization problem of a social planner, who seeks to maximize the total welfare in the system. Under a few technical assumptions, we are able to prove the existence of a unique optimal allocation. \change{A substantial technical contribution of this paper is to provide explicit solution of a multi-player game of intensity control in a CRRA-utility maximization setting. We also also able to solve explicitly the associated social planner's problem, and characterize the asymptotics of the system's inefficiencies. Finally, we can also extend the model in various directions described in the Appendices, which can be solved up to numerical solution of systems of algebraic equations.}

In particular, we find that the planner's optimum exhibits low-frequency and high-severity events of distress, which aligns with the `robust-yet-fragile' feature observed by~\citet{Gai2010}. The difference in these two optimal allocations is driven by a negative externality, where individual banks are excessively risky given the potential losses that they may induce.

In the case where shocks are exponentially distributed, we can precisely compute how the externality's severity depends on the system's size. We see that the planner compensates for an increased number of \change{investors} by reducing the risk of a bank's project. The planner perfectly balances the two effects, so that the expected losses in utility remain of constant order -- regardless of the system size. We are also able to see that the absolute loss in welfare due to decentralized behavior grows with the size of both the financial system and its core. However, and perhaps counterintuitively, the \change{welfare \textit{ratio}, also known} as the price of anarchy, is of constant order. Finally, we show that it is possible, through regulation or contracting between banks, to replicate the planner's optimal interbank allocation.

\change{Banks in the core of the system must be subjected to the strictest requirements, and will therefore have the strongest incentive to reduce their project's riskiness. This highlights both i) the danger of government bailouts, which can cause perverse incentives for individual banks, and ii) the heterogeneous impact of regulation on social welfare.}

\change{In the wake of the regional banking crisis in 2023, the widespread financial distress was addressed with forms of bailouts and takeovers of, for instance, Silicon Valley Bank and First Republic Bank. Subsequently, regulators are discussing revisions to capitalization and liquidity requirements of banks contained in Basel III \citep{basel3}, which led to a debate and substantial lobbying from large, central banks such as JP Morgan. This resulted in US regulators lowering their initial proposed increases to capital and liquidity requirements. Our work strongly supports the idea of different degrees of regulation imposed on banks in the core and the periphery.}

We believe there are several interesting continuations of this work. First, a notable limitation of this model is that it does not contain a mechanism of contagion. For instance,~\citet{ait2015portfolio} consider a portfolio optimization problem where assets' jump components are self- and mutually exciting. An immediate extension of our work may be to incorporate jump processes with these features directly into the model. It may also be possible to show that self- and mutually exciting jumps can endogenously emerge, e.g. if an investing bank suffers losses to liquidity when \change{the bank associated with this project suffers}. 
\change{Second, we assume that the liquidity risk premium $r$ is fixed and exogenous, so that banks may borrow as much as they like at this rate. In practice, this rate would be determined endogenously by the relative supply and demand of liquidity, and substantially affected by market conditions.}
Additionally, financial crises are heavily destabilizing, and it is natural to assume that it is challenging (or impossible) to quickly rebalance a portfolio in the wake of such an event. Therefore, it is practical to prevent banks from instantaneously re-weighting their portfolios. This feature may lead to further inefficiencies caused by banks' inability to establish an optimal allocation of wealth shortly after a shock occurs. Furthermore, our model differs from the literature on strategic network formation in that creating an `investment linkage' to another bank is costless. It is natural to incorporate these costs into banks' optimization problems, for example, as there is a cost to performing due diligence \change{on a potential investment opportunity to assess its profitability}.



\section*{Acknowledgments}
We are grateful to Emma Hubert and Roberto Rigobon for helpful comments and discussion.


\bibliographystyle{plainnat}
\bibliography{library}

\newpage
\appendix

\section{Proofs}
\label{sec:proofs}

\subsection{Decentralized Network}
\label{ssec:pfs-dec}

\begin{proof}[Proof of Proposition~\ref{prop:dec_opt_hjb}]	
	First, we use the dynamic programming principle to consider only the optimal control over the time interval $[t,\tau]$, for a stopping time $\tau < T$ to be defined later. We can write the value function recursively as

	\begin{equation}
		\label{eq:dec_val_DP}
		V_i(t,x) = \sup_{(c^i_\cdot,w^{i\cdot}_\cdot)\in \A^i_{t,T}} \ \E 	\Big[V_i\left(\tau, X^i_{\tau}\right) \Big| X_t^i = x\Big],
	\end{equation}
	which holds for all $t<T$ and $\tau \le T$.

	Next, we for each bank $k$ we fix some admissible control $(c^k_\cdot,w^{k\cdot}_\cdot)\in \A^k_{t,T}$. By assumption, $V_i$ is once differentiable in both time and space, and using It\^{o}'s formula (see for instance~\cite{tankov2003financial}) we can write:	
	\begin{equation}
		\label{eq:dec_val_ito_int}
		\begin{split}
			V_i\left(\tau, X^i_{\tau}\right) -   V_i\left(t,X_t^i\right) = \ & 	\int_{t}^{\tau} \left[\partial_t V_i(s,X_s^i)  + \partial_x V_i\left(s,X_s^i\right) b_i({c}^i_s, {w}^{i\cdot}_s) X_s^i \right]ds \\
			&+ \sum_{j=1}^n \int_{t}^{\tau} \left[ V_i\left(s,X_s^i\right) - 	V_i\left(s,X_{s-}^i\right) \right] dN_s^j.
		\end{split}
	\end{equation}
	where $b_i({c}^i_t, {w}^{i\cdot}_t)$ is the coefficient on the $dt$ term in~\eqref{eq:wealth_sde}. 
	
	Recall that the jump process $N_t^j$ has instantaneous intensity $\theta_j \bar F_j(c^j_t) $. Therefore, the compensated process $M_t^j = N_t^j - \int_{0}^{t} \theta_j \bar F_j(c^j_s) ds$ is a martingale. Rewriting the integrals in~\eqref{eq:dec_val_ito_int} in terms of $dM_t^j$ and taking expectation conditioned on $X^i_t = x$ (denoted $\mathbb{E}_{t,x}$) of both sides yields:
	\begin{equation}
		\label{eq:dec_val_exp}
		\begin{split}
			\mathbb{E}_{t,x} \left[ V_i\left(\tau, X^i_{\tau}\right)\right]  -   V_i\left(t,X_t^i\right) = \ & \mathbb{E}_{t,x} \Bigg[ \int_{t}^{\tau} \mathcal{L}^{c^i_s,w^{i\cdot}_s}V_i(s,X^i_{s-}) ds \Bigg] \\
			&+ \mathbb{E}_{t,x} \left[ \int_{t}^{\tau} \left[ 	V_i\left(s,X_{s-}^i - \eta_i X_{s-}^i\right) - V_i\left(s,X_{s-}^i\right) \right] dM_s^i \right]\\
			&+ \sum_{j\neq i} \mathbb{E}_{t,x} \left[ \int_{t}^{\tau} \left[ 	V_i\left(s,X_{s-}^i - \phi_j {w}^{ij}_s X_{s-}^i\right) - V_i\left(s,X_{s-}^i\right) \right] dM_s^i \right],
		\end{split}
	\end{equation}
	where the generator $\mathcal{L}^{c_i,w_{i\cdot}}$ is defined to be
	\begin{equation}
		\label{eq:dec_generator}
		\begin{split}
			\mathcal{L}^{c_i,w_{i\cdot}}\psi(t,x) = \ &\partial_t \psi(t,x)  +  \left( (1 - c_i) r  + \sum_{j\neq i}w_{ij}  \mu_j + \frac{\eta_i \mu_i}{\phi_i}\right) x \partial_x \psi \\
			&+ \theta_i \left(1 - F_i(c_i)\right) \Big[ \psi(t, x (1 - \eta_i)) - \psi(t,x) \Big]  \\
			&+ \sum_{j\neq i} \theta_j \left(1 - F_j(c_j)\right) \Big[ \psi(t,x(1 - \phi_j {w}_{ij})) - \psi(t,x) \Big],
		\end{split}
	\end{equation}
	for any $\psi \in  \mathcal{C}^{1,1}([0,T), \mathbb{R}_+)$.

	Next, we need to show that the expectation of the stochastic integrals with respect to $dM_s^k$ are equal to zero. To do so, it is sufficient to have the integrand bounded for $s\in [t,\tau]$. Define the stopping time $\tau$ to be:
	\begin{equation}
		\label{eq:tau_def}
		\tau = (t+\delta) \wedge \inf \left\{s\in [t,T], X^i_s \le \epsilon  \text{ or } X^i_s \ge \frac{1}{\epsilon}	\right\},
	\end{equation}
	for some small $\delta > 0 $ and $\epsilon > 0$. Then, since $X^i_s$ is bounded away from zero in $[t,\tau]$, the size in the jump of the value function is bounded.	Therefore the stochastic integrals in~\eqref{eq:dec_val_exp} have zero expectation. We obtain:
	\begin{equation}
		\label{eq:dec_val_pf_pre_dpp}
		\begin{split}
			\mathbb{E}_{t,x} \left[ V_i\left(\tau, X^i_{\tau}\right)\right] -  	V_i\left(t,X_t^i\right) = \ & \mathbb{E}_{t,x} \Bigg[ \int_{t}^{\tau} \mathcal{L}^{c^i_s,w^{i\cdot}_s}V_i(s,X^i_{s-}) ds \Bigg].
		\end{split}
	\end{equation}

	Take the supremum on both sides over the admissible controls $(c^i_\cdot,w^{i\cdot}_\cdot)\in \A^i_{t,T}$. Recall that the dynamic programming principle in~\eqref{eq:dec_val_DP} implies that for any stopping time $\tau$, we have 
	$$\sup_{(c^i_\cdot,w^{i\cdot}_\cdot)\in \A^i_{t,\tau}} \mathbb{E}_{t,x} \left[ V_i\left(\tau, X^i_{\tau}\right)\right] = V_i\left(t,X_t^i\right).$$ Therefore, we arrive at:
	\begin{equation}
		\label{eq:dec_val_pf_pre_lim}
		\begin{split}
			0 = \sup_{(c^i_\cdot,w^{i\cdot}_\cdot)\in \A^i_{t,\tau}} \mathbb{E}_{t,x} \Bigg[&  	\int_{t}^{\tau} \mathcal{L}^{c^i_s,w^{i\cdot}_s}V_i(s,X^i_{s-}) ds \Bigg].
		\end{split}
	\end{equation}

	We note that this step required existence of an optimal control. For small enough $\delta$ and $\epsilon$ in~\eqref{eq:tau_def}, we will have $\tau = t+\delta$. Therefore,~\eqref{eq:dec_val_pf_pre_lim} yields
	\begin{equation}
		\label{eq:dec_val_pf_pre_dct}
		\begin{split}
			0 =  \sup_{(c^i_\cdot,w^{i\cdot}_\cdot)} \lim_{\delta\to 0} \frac{1}{\delta}  \mathbb{E}_{t,x} \Bigg[&  \int_{t}^{t+\delta} \mathcal{L}^{c^i_s,w^{i\cdot}_s}V_i(s,X^i_{s-}) ds \Bigg].
		\end{split}
	\end{equation}

	Finally, applying the Dominated Convergence Theorem gives
	\begin{equation}
		\label{eq:dec_val_pf_end}
		\begin{split}
			0 =  \sup_{(c^i_\cdot,w^{i\cdot}_\cdot)} \mathcal{L}^{c_i,w_{i\cdot}}V_i(t,x),
		\end{split}
	\end{equation}
	which equals~\eqref{eq:dec_hjb} after plugging in the definition of $\mathcal{L}^{c_i,w_{i\cdot}}$ from~\eqref{eq:dec_generator}.
\end{proof}

\begin{proof}[Proof of Proposition~\ref{prop:dec_hjb_soln}] 
	Both parts of this Proposition are proved nearly identically. For conciseness, full detail is only provided for case $(i)$ where $\gamma_i = 1$.

	$(i)$: We first show that~\eqref{eq:dec_hjb} has a separable solution. Next, the internal optimization problem is shown to be convex, and its objective function strictly concave. Finally, we show that the proposed solution is optimal.

	\paragraph{Separability of the PDE:}
	First we show the value function is separable. Plugging the ansatz $V_i(t,x) = g_i(t) + \log x$ into~\eqref{eq:dec_hjb} and performing some simplification, we have:
	\small
	\begin{equation}
		\label{eq:dec_hjb_separability}
		0 = g_i'(t)+  \sup_{c_i,  w_{i\cdot}} \left\{ \left(1-c_i\right)r + \sum_{j\neq i}w_{ij}\mu_j  + \frac{\eta_i \mu_i}{\phi_i} + \theta_i \bar F_i(c_i)  \log\left(\frac{x-\eta_i x}{x}\right) \sum_{ j\neq i} \theta_j \bar F_j(c_j) 	\log\left(\frac{x-\phi_j w_{ij}x}{x}\right)\right\}.
	\end{equation}
	\normalsize
	Observe we can cancel out all remaining $x$'s, and obtain the following ODE for $g_i$:
	\small
	\begin{equation}
		\label{eq:dec_g_ode_log}
		0 = g_i'(t) + \frac{\eta_i \mu_i}{\phi_i} + \sup_{c_i, w_{i\cdot}} \left\{ (1 - c_{i})r + \sum_{j\neq i}w_{ij}\mu_j + \theta_i \bar F_i(c_i) \log(1-\eta_i) + \sum_{j\neq i} \theta_j \bar F_j(c_j) \log(1-\phi_j w_{ij}) \right\} \\
	\end{equation}
	\normalsize
	with terminal condition $g_i(T) = 0$. If $\hat c_i$ and $\hat w_{ij}$ are indeed the optimal solutions to the maximization in~\eqref{eq:dec_g_ode_log}, $g_i$ solves $g_i'(t) = -J_i^* $ with $g_i(T) = 0$, to which the solution is $g_i(t) = (T-t) J_i^*$ as desired.

	\paragraph{Strict Concavity:}
	Now we analyze the resulting optimization problem for $c_i, w_{i\cdot}$. 
	Let $\mathcal{A}_i = \mathbb{R}_+ \times \prod_{j\neq i}[0,\phi_j^{-1})$ be the feasible set for this optimization problem. Clearly, $\mathcal{A}_i$ is a convex set. We aim to solve
	\begin{equation}
		\label{opt:dec_log}
		\begin{split}
			\sup_{(c_i, w_{i\cdot}) \in \mathcal{A}_i} \quad & (1 - c_{i})r + \sum_{j\neq i}w_{ij}\mu_j + \theta_i \bar F_i(c_i) \log(1-\eta_i)  + \sum_{j\neq i} \theta_j \bar F_j(c_j) \log(1-\phi_j w_{ij}).
		\end{split}
	\end{equation}

	Let $h(c_i, w_{i\cdot})$ denote the function to be maximized in~\eqref{opt:dec_log}. It is critical to observe that $h$ is additively separable in each of its optimization variables. Therefore, we can solve for each optimal control independently. Namely, all cross-derivatives of $h$ equal zero, which greatly simplifies the proof of strict concavity. We begin by computing partial derivatives of $h$ with respect to each variable, which gives
	\begin{equation}
		\label{eq:dec_derivs_log}	
		\begin{alignedat}{3}
			\frac{\partial h}{\partial c_i} &=  -r - \theta_i f_i(c_i) 	\log(1-\eta_i) \qquad \qquad && \frac{\partial^2 h}{\partial c_i^2} =   - \theta_i f_i'(c_i)  	\log(1-\eta_i) \\
			\frac{\partial h}{\partial w_{ij}} &=  \mu_j -\phi_j \theta_j  	\frac{\bar F_j(c_j)}{1-\phi_j w_{ij}} && \frac{\partial^2 h}{\partial w_{ij}^2} =  -\phi_j^2 \theta_j  	\frac{\bar F_j(c_j)}{(1-\phi_j w_{ij})^2} \qquad \forall j\neq i.
		\end{alignedat}
	\end{equation}

	Observe that within $\mathcal{A}_i$, we have $(1 - \phi_j w_{ij})^2 > 0$. Recall that by Assumption~\ref{assn}, the density function $f_j(\cdot)$ is fully supported on $\R_+$, and $f_i'(\cdot) < 0$. Therefore, it must be the case that $\bar F_j(c_j) > 0$ for any admissible $c_j$ and $\partial^2_{w_{ij}, w_{ij}} h < 0$. Additionally, $\partial^2_{c_{i}, c_{i}} h < 0$ because $\eta_j > 0$.

	As a result, the Hessian matrix of the objective function is negative definite in the feasible region, i.e. $\nabla^2 h \prec 0$ everywhere in $\mathcal{A}_i$. Hence $h$ is a strictly concave function; if an optimal solution to problem~\eqref{opt:dec_log} exists, it is unique~\citep{boyd2004convex}.

	\paragraph{Optimality of Given Solution:}
	To conclude, we must prove that~\eqref{sol:dec} is optimal for bank $i$.

	Note that $\frac{-r}{\theta_i\log(1-\eta_i)} > 0$. Since $f_i$ is monotonically decreasing and positive valued on $\R_+$, its inverse $f_i^{-1}\left(\frac{-r}{\theta_i\log(1-\eta_i)}\right)$ is well-defined if and only if $\frac{-r}{\theta_i\log(1-\eta_i)} \le f_i(0)$.

	Since optimization problem~\eqref{opt:dec_log} is convex, the first-order condition for constrained optimization is sufficient. We need only check that $y^* = (\hat c_i, w_{i\cdot}^*) \in \mathcal{A}_i$ satisfies
	\begin{equation}
		\label{eq:convex_foc_dec}
		\nabla{h}(y^*)^T(y - y^*) \le 0, \ \forall y \in \mathcal{A}_i.
	\end{equation} 
	The optimization problem for $h$ is additively separable, so this condition is equivalent to the following:
	\begin{equation}	
		\label{eq:dec_foc_sep}
		\begin{split}
			\partial_{c_i}h(\hat c_i)(c_i - \hat c_i) &\le 0, \ \forall c_i \in \R_+, \\
			\partial_{w_{ij}} h(\hat w_{ij})(w_{ij} - \hat w_{ij}) &\le 0, \ \forall w_{ij} \in \left[\left.0, \phi_j^{-1} \right)\right., \quad \forall j\neq i.
		\end{split}
	\end{equation}
	Note that the partial derivative $\partial_{c_i} h$ in~\eqref{eq:dec_derivs_log} is a function of only $c_i$. The same holds for the partials with respect to each $w_{ij}$. Note that these derivatives will depend on $c_j$, but this value is not controlled by bank $i$. Therefore, we will omit the dependence of these derivatives on the other optimization variables.

	We begin with optimality of the proposed $\hat c_i$. Consider the case where $\frac{-r}{\theta_i\log(1-\eta_i)} \le f_i(0)$, and observe that $\partial_{c_i}h(\hat c_i) = 0$ using~\eqref{eq:dec_derivs_log}. As a result, this choice of $\hat c_i$ satisfies the first-order condition for $\hat c_i$ in~\eqref{eq:dec_foc_sep}. Conversely, let us have $\frac{-r}{\theta_i\log(1-\eta_i)} > f_i(0)$. Since $f_i$ is assumed to be monotone decreasing, it must be the case that $\frac{-r}{\theta_i\log(1-\eta_i)} > \max_{c \in \R_+} f_i(c)$. Using again~\eqref{eq:dec_derivs_log}, we obtain that $\partial_{c_i}h(c) < 0$ for every $c\in \R_+$. In particular, we will have $\partial_{c_i}h(0) < 0$, and the first-order condition~\eqref{eq:dec_foc_sep} is satisfied by $\hat c_i = 0$. The proof of optimality for $\hat w_{ij}$ in~\eqref{sol:dec} follows exactly the same steps. If it is non-zero, then the proposed value solves $\partial_{w_{ij}}h(\hat w_{ij}) = 0$. If not, then we know that this partial derivative is negative everywhere in the feasible region for $w_{ij}$. Choosing $\hat w_{ij} = 0$ satisfies the corresponding equation in~\eqref{eq:dec_foc_sep}.

	Concluding, we have shown that the solution given in~\eqref{sol:dec} satisfies~\eqref{eq:dec_foc_sep}. Since it lies within $\mathcal{A}_i$, it is optimal for problem~\eqref{opt:dec_log}. Recall that strict concavity provides uniqueness of this solution. Finally, since all banks optimize concurrently,~\eqref{sol:dec} is obtained by plugging the optimal value $c_j^*$ into $\hat w_{ij}$.

	\vspace{1em}
	$(ii)$: The proof of this result will largely mirror that of part $(i)$. We first check separability of the PDE. If $V_i(t,x) = g_i(t)\frac{x^{1-\gamma_i}}{1-\gamma_i}$, then we have:
	\begin{equation}
		\label{eq:dec_pwr_simplification}
\partial_t V_i(t,x) = g_i'(t) \frac{x^{1-\gamma_i}}{1-\gamma_i}, \quad 
			\partial_x V_i(t,x) = g_i(t)  \frac{x^{1-\gamma_i}}{x}, \quad 
V_i(t, (1-c)x) = g_i(t) 	\frac{x^{1-\gamma_i}}{1-\gamma_i}(1-c)^{1-\gamma_i},
	\end{equation}
for all $c<1$.	Plugging these expressions into~\eqref{eq:dec_hjb} and dividing by $x^{1-\gamma_i}$ removes any spatial variables, and we are left with the following ordinary differential equation for $g_i$:
	\begin{equation}
		\label{eq:dec_g_ode_pwr}
		\begin{alignedat}{3}
			0 &=  \frac{g_i'(t)}{1-\gamma_i} + g_i(t)\sup_{c_i, w_{i\cdot}} \Bigg \{ && (1 - c_{i})r + \sum_{j\neq i}w_{ij}\mu_j  + \frac{\eta_i \mu_i}{\phi_i} + \theta_i \bar F_i(c_i) \frac{(1-\eta_i)^{1-\gamma_i} - 1}{1-\gamma_i}  \\
			& && + \sum_{j\neq i} \theta_j \bar F_j(c_j) \frac{(1-\phi_j w_{ij})^{1-\gamma_i} - 1 }{1-\gamma_i} \Bigg\}, 
		\end{alignedat}
	\end{equation}
with $g_i(T)= 1$.
	
	Let $\hat c_i$ and $\hat w_{ij}$ be the optimal solutions to the maximization. Then we see that $g_i$ will solve $g_i'(t) = -(1-\gamma_i)J_i^* g_i(t)$ with $g_i(T) = 1$, whose solution is $g_i(t) = \exp((1-\gamma_i)(T-t)J_i^*)$.

	The optimality and uniqueness of the solution in~\eqref{sol:dec} will be proved analogously to part $(i)$, but by analyzing a different objective function. We are now interested in:	
	\begin{equation}
		\label{opt:dec_pwr}
		\sup_{(c_i, w_{i\cdot}) \in \mathcal{A}_i} \quad (1 - c_{i})r + \sum_{j\neq i}w_{ij}\mu_j + \theta_i \bar F_i(c_i) \frac{(1-\eta_i)^{1-\gamma_i} - 1}{1-\gamma_i}  + \sum_{j\neq i} \theta_j \bar F_j(c_j) \frac{(1-\phi_j w_{ij})^{1-\gamma_i} - 1 }{1-\gamma_i}.
	\end{equation}
	Again, this optimization problem is additively separable, which will simplify the proof of strict concavity. As before, let $h(c_i,w_{i\cdot})$ denote the function to be maximized. We compute its partial  derivatives to be:
	\begin{equation}
		\label{eq:dec_derivs_pwr}
		\begin{alignedat}{3}
			\frac{\partial h}{\partial c_i} &=  -r - \theta_i f_i(c_i) 	\frac{(1-\eta_i)^{1-\gamma_i} - 1}{1-\gamma_i} && \frac{\partial^2 h}{\partial c_i^2} =  - \theta_i f_i'(c_i) 	\frac{(1-\eta_i)^{1-\gamma_i} - 1}{1-\gamma_i}\\
			\frac{\partial h}{\partial w_{ij}} &=  \mu_j - \phi_j \theta_j 	\bar F_j(c_j) (1-\phi_j w_{ij})^{-\gamma_i} \qquad && \frac{\partial^2 h}{\partial w_{ij}^2} =  \phi_j^2 \theta_j 	\bar F_j(c_j) (-\gamma_i)(1-\phi_j w_{ij})^{-\gamma_i-1} 
		\end{alignedat}
	\end{equation}

	Under Assumption~\ref{assn}, we will have both $\partial^2_{c_i, c_i}h < 0$ and $\partial^2_{w_{ij}, w_{ij}}h < 0$, since $w_{ij} < \phi_j^{-1}$ everywhere in $\mathcal{A}_i$. Therefore, $h$ is strictly concave on $\mathcal{A}_i$ and the optimization problem is convex. As a result, any optimal solution must be unique.

	The remaining part of the proof mirrors that of part $(i)$. Computing the gradient of $h$ at the candidate solution in~\eqref{sol:dec} and using the same argument will show that the first-order conditions in~\eqref{eq:dec_foc_sep} are satisfied. Since this point is feasible, it must be optimal.
\end{proof}

\begin{proof}[Proof of Corollary~\ref{cor:dec_verification}]
	We proceed with a standard verification argument. We need to show that if $\psi$ is a solution to the PDE~\eqref{eq:dec_hjb} and it is $\mathcal{C}^{1,1}\left( [0,T), \mathbb{R}_+\right)$, then it is equal to the value function. Since the proposed solutions solve the PDE and they are indeed $\mathcal{C}^{1,1}$, this will conclude.

	Fix $t< T$, and choose $\{c^i_s, w^{i\cdot}_s\}_{s\in [t,T]}$ be some admissible controls. We apply It\^{o}'s formula to $\psi(s,X^i_s)$ between $t$ and some stopping time $\tau^n$ -- to be chosen optimally later. This yields, using the notation introduced in the proof of Proposition~\ref{prop:dec_opt_hjb}, the following:	
	\begin{equation}
		\label{eq:dec_ver_ito}
		\begin{split}
			\psi(\tau^n,X^i_{\tau^n}) = \ &\psi(t,X^i_t) + \int_{t}^{\tau^n} 	\mathcal{L}^{c^i_s,w^{i\cdot}_s}\psi(s,X^i_s)ds +  \int_{t}^{\tau^n} \left[ \psi\left(s,X_{s-}^i - \eta_i 	X_{s-}^i\right) - \psi\left(s,X_{s-}^i\right) \right] dM_s^i \\
			&+ \sum_{j\neq i}  \int_{t}^{\tau^n} \left[ \psi\left(s,X_{s-}^i - 	\phi_j w^{ij}_s X_{s-}^i\right) - \psi\left(s,X_{s-}^i\right) \right] dM_s^j.
		\end{split}
	\end{equation}	
	Recall that the compensated jump process $\left\{M_t^k\right\}_{t\ge 0}$ is a martingale. Taking the expectation conditioned on $X^i_t = x$, we obtain:
	\begin{equation}
		\label{eq:dec_ver_exp}
		\begin{split}
			\mathbb{E}_{t,x} \left[ \psi(\tau^n,X^i_{\tau^n})\right] = \ 	&\psi(t,x) + \E_{t,x} \left[ \int_{t}^{\tau^n} \mathcal{L}^{c^i_s,w^{i\cdot}_s}\psi(s,X^i_s)ds \right]\\
			& +  \mathbb{E}_{t,x} \left[ \int_{t}^{\tau^n} \left[ 	\psi\left(s,X_{s-}^i - \eta_i X_{s-}^i\right) - \psi\left(s,X_{s-}^i\right) \right] dM_s^i \right]\\
			&+ \sum_{j\neq i}  \mathbb{E}_{t,x} \left[ \int_{t}^{\tau^n} \left[ 	\psi\left(s,X_{s-}^i - \phi_j w^{ij}_s X_{s-}^i\right) - \psi\left(s,X_{s-}^i\right) \right] dM_s^j \right].
		\end{split}
	\end{equation}

	If we choose $\tau^n = \left(T - \frac{1}{n}\right) \wedge \inf \left\{ s \in [t,T], X^i_s \le \frac{1}{n} \text{ or } X^i_s \ge n \right\}$, then for every $n$ the expectation of each stochastic integral is zero and we have:
	\begin{equation*}
		\mathbb{E}_{t,x} \left[ \psi(\tau^n,X^i_{\tau^n})\right] = \psi(t,x) + 	\E_{t,x} \left[ \int_{t}^{\tau^n} \mathcal{L}^{c^i_s,w^{i\cdot}_s}\psi(s,X^i_s)ds \right].
	\end{equation*}

	Taking the limit as $n\to \infty$, we will have $\tau^n \to T$. Furthermore, since $\psi$ satisfies the terminal condition (by assumption) and everything is bounded, an application of dominated convergence yields:	
	\begin{equation}
		\label{eq:dec_ver_general}
		\mathbb{E}_{t,x} \left[ U_i(X^i_T)\right] = \psi(t,x) + \E_{t,x} \left[ 	\int_{t}^{T} \mathcal{L}^{c^i_s,w^{i\cdot}_s}\psi(s,X^i_s)ds \right].
	\end{equation}

	First, we choose the controls in~\eqref{eq:dec_ver_general} to be given by the optimal solution of Proposition~\ref{prop:dec_hjb_soln}. Then, we will have $\mathcal{L}^{\hat c^{i}_s,\hat w^{i\cdot}_s}\psi(s,X^i_s) = 0$ for all $s \in [t,\tau^n]$, and consequentially:
	\begin{equation*}
		\psi(t,x) = \E_{t,x} \left[ U_i(X^i_T)\right].
	\end{equation*}
	Note that only the terminal wealth $X^i_T$ in the right-hand side depends on the controls $(\hat c^{i}_s,\hat w^{i\cdot}_s)$. After taking the supremum we obtain
	\begin{equation}
		\label{eq:dec_ver_ub}
		\psi(t,x) \le \sup_{\left\{c^i_s, w_s^{i\cdot}\right\}_{s\in [t,T]}} 	\E_{t,x} \left[ U_i(X_T^i) \right] = V_i(t,x).
	\end{equation}

	Next, we fix any control $(c^{i}_s,w^{i\cdot}_s)$. Then, in~\eqref{eq:dec_ver_general} we will have $\mathcal{L}^{c^{i}_s,w^{i\cdot}_s}\psi(s,X^i_s) \le 0$, and the result is: 
	\begin{equation*}
		\psi(t,x) \ge \E_{t,x} \left[ U_i(X^i_T)\right].
	\end{equation*}
	Note again that only $X^i_T$ depends on the controls. However, since this inequality holds for any admissible control we can take the supremum over both sides to give
	\begin{equation}
		\label{eq:dec_ver_lb}
		\psi(t,x) \ge \sup_{\left\{c^i_s, w_s^{i\cdot}\right\}_{s\in [t,T]}} 	\E_{t,x} \left[ U_i(X_T^i) \right] = V_i(t,x).
	\end{equation}

	Combining~\eqref{eq:dec_ver_ub} and~\eqref{eq:dec_ver_lb} shows that $\psi = V_i$. This implies that the optimal values to the maximization problem in the PDE for $\psi$ are indeed the optimal controls.
	Since the explicit solutions given by Proposition~\ref{prop:dec_hjb_soln} are once continuously differentiable in both time and space, then they are equal to the value function.
\end{proof}

\subsection{Centralized Network}
\label{ssec:pfs-cent}

\begin{proof}[Proof of Proposition~\ref{prop:cent_opt_hjb}]
	This proof is only a minor adaptation of the proof of Proposition~\ref{prop:dec_opt_hjb}. 
	First, the application of It\^{o}'s formula to the value function $V(t,X_t^1,\dots,X_t^n)$ yields more terms, but remains simple as the jump processes are mutually independent. Namely, the generator is given by
	\begin{equation*}
		\label{eq:cent_generator}
		\begin{split}
			\mathcal{L}^{c_\cdot, w_{\cdot \cdot}}\psi =& \  
			\partial_t \psi + \sum_{i=1}^n \Bigg( \left[ (1-c_i)r + \sum_{j\neq i}w_{ij}\mu_j + \frac{\eta_i \mu_i}{\phi_i}\right] x_i \partial_{x_i}\psi \\
			& + \theta_i \bar F_i(c_i)\Big[ \psi(t,x_1(1-\phi_i w_{1i}),\dots,x_i(1-\eta_i),\dots,x_n(1-\phi_i w_{ni})) - \psi \Big] \Bigg),
		\end{split}
	\end{equation*}
	where $\psi$ is evaluated at $(t, x_1,\dots,x_n)$ where unspecified.

	Next, to apply dominated convergence, the choice of the stopping time $\tau$ must ensure that all stopped processes $X_\tau^1,\dots X_\tau^n$ are bounded away from zero. We can therefore choose:
	$$\tau = (t+\delta) \wedge \min_i \left\{ \inf \left \{s \in [t,T], X^i_s \le \epsilon \text{ or } X^i_s \ge \frac{1}{\epsilon} \right\} \right\},$$
	and conclude as in the previous result.
\end{proof}

\begin{proof}[Proof of Proposition~\ref{prop:cent_hjb_soln}]
	The outline of this proof is similar to that of Prop.~\ref{prop:dec_hjb_soln}, but with greater complexity, and hence requiring additional assumptions to establish our results. We begin by discussing each of these.
First, logarithmic utility functions are needed so that~\eqref{eq:cent_hjb} admits a separable solution. We note that if the planner sought to maximize the product of banks' utilities, it would be necessary to assume that $\gamma_i \neq 1$ for all $i$. This assumption is used for existence of a separable solution to~\eqref{eq:cent_hjb}.

	The first condition in Assumption~\ref{assn:cent_uniq} concerns the shock densities $f_i$. In particular,~\eqref{eq:cent_f_cond} is satisfied by the family of exponential distributions ($f_i(x) = \lambda_i^{-1} e^{-\frac{x}{\lambda_i}}$, for some parameter $\lambda_i > 0$) and power distributions $\Big(f_i(x) = \frac{\left(\alpha_i^{-1} - 1\right) x_0^{\alpha_i^{-1}-1}}{(x+x_0)^{\alpha_i^{-1}}}$, for any $x_0 > 0$ and $\alpha_i < 1\Big)$. We note that this condition is not necessary for uniqueness, but is used for establishing monotonicity of a first-order condition for optimality by bounding the second derivative with an exponentially decaying function.

	Finally, the inequalities on $\Gamma(\eta_i;1)$ will ensure that either $(i)$: strict concavity of the objective function holds, or $(ii)$ there exists only a single solution to the necessary first-order conditions. However, these inequalities do not rule out the possibility of a corner solution of $c_i^* = 0$ or $w_{\cdot i}^* = 0$ -- as shown in~\eqref{sol:cent}. Of particular interest, the optimal decentralized and centralized allocations for $c_i$ and $w_{\cdot i}$ will coincide whenever either $c_i^* = 0$ or $w_{\cdot i}^* = 0$ in the planner's optimum.

	\paragraph{Separability of PDE and Maximization:}
	Recall that the PDE for the value function derived in Proposition~\ref{prop:cent_opt_hjb} is:
	\begin{equation}
		\label{eq:cent_hjb_pf}
		\begin{alignedat}{2}
			0  &= \partial_t V+ \sup_{c_\cdot,  w_{\cdot\cdot}} \Bigg\{ \sum_{i=1}^n \Bigg( \left[\left(1-c_i\right)r + \sum_{j\neq i}w_{ij}\mu_j + \frac{\eta_i \mu_i}{\phi_i} \right] x_i \partial_{x_i} V 
			\\
			& \quad  + \theta_i \bar F_i(c_i)\Big[ V(t,x_1(1-\phi_i w_{1i}),.., 	x_i(1-\eta_i),.., x_n(1-\phi_i w_{ni})) - V \Big] \Bigg)\Bigg\} 
			\\
			V(T,x_1,\dots,x_n) &= \sum_{i=1}^n U_i(x_i).  
		\end{alignedat}	
	\end{equation}
	By assumption, each bank's utility function is given by $U_i(x_i) = \log x_i$, i.e. $\gamma_i = 1$ for all $i$. Consider the following ansatz: $V(t,x_1,..,x_n) = g(t) + \sum_{i} \log x_i $. Substituting into~\eqref{eq:cent_hjb_pf}, we obtain:
	\begin{equation}
		\label{eq:cent_ode_g}
		\begin{alignedat}{3}
			0  &= g'(t) + \sup_{c_\cdot,  w_{\cdot\cdot}}  &&\sum_{i=1}^n \left(1-c_i\right) r + \sum_{j\neq i}w_{ij}\mu_j + \frac{\eta_i \mu_i}{\phi_i}  - \theta_i \bar F_i(c_i)\left[ \Gamma(\eta_i;1) + \sum_{j\neq i} \Gamma(\phi_i w_{ji};1)  \right] \\
		\end{alignedat}
	\end{equation}
	with $g(T) = 0$. The spatial variables will cancel and we are left with an ordinary differential equation for $g$.
	We now rewrite the following sum:
	$$\sum_{i=1}^n \sum_{j\neq i} w_{ij}\mu_j = \sum_{i=1}^n \sum_{j\neq i} w_{ji} \mu_i.$$
	Observe that for $k,j\neq i$, we will have $w_{ji} = w_{ki}$. That is, all $j\neq i$ banks will invest the same fraction of their wealth in bank $i$'s project. This can be seen in two ways. First, in the decentralized setting, the amount $w_{ji}$ depended on bank $j$ only through their risk aversion coefficient $\gamma_j$. Since in this Proposition we have assumed that $\gamma_i = 1$ for all $i$, the result follows. This can also be seen by computing the first-order conditions in~\eqref{eq:cent_ode_g} for $w_{ji}$ and $w_{ki}$, and noticing that they are identical. Let this fraction be denoted by $w_{\cdot i}$. This allows us to further simplify~\eqref{eq:cent_ode_g} and obtain
	\begin{equation}
		\begin{split}
			0 = g'(t) + \sum_{i=1}^n \frac{\eta_i \mu_i}{\phi_i} + \sup_{c_\cdot,  w_{\cdot\cdot}} \sum_{i=1}^n \left(1-c_i\right) r + (n-1)w_{\cdot i}\mu_i    - \theta_i \bar F_i(c_i)\Big[  \Gamma(\eta_i;1) + (n-1) \Gamma(\phi_i w_{\cdot i};1)  \Big]  .\\
		\end{split}
	\end{equation}
	This maximization is additively separable between each pair $(c_i, w_{\cdot i})$, indexed by $i$. Let $\A_i = \R_+ \times [0,\phi_i^{-1})$ denote the admissible values for $(c_i, w_{\cdot i})$. Then, the optimal allocation is found by solving:
	\begin{equation}
		\label{opt:cent_sep}
		\sum_{i=1}^n \sup_{(c_i, w_{\cdot i}) \in \A_i} h_i(c_i, w_{\cdot i}),
	\end{equation}
	where $h_i(c_i, w_{\cdot i}) = -rc_i + (n-1)\mu_i w_{\cdot i} - \theta_i \bar F_i(c_i)\Big[  \Gamma(\eta_i;1) + (n-1) \Gamma(\phi_i w_{\cdot i};1)  \Big]$ for each $i$.

	\paragraph{Reduction to Univariate Optimization:}
	We first maximize over $w_{\cdot i}$ and then $c_i$ given the optimal $w_{\cdot i}$. Given a value of $c_i$, we seek to find the optimal value of $w_{\cdot i}$. We can compute
	\begin{equation}
		\label{eq:partials_w}
		\begin{split}
			\frac{\partial h_i}{\partial w_{\cdot i}}(c_i, w_{\cdot i}) &= (n-1)\mu_i - (n-1) \frac{\phi_i \theta_i\bar F_i(c_i)}{1-\phi_i w_{\cdot i}}
			\\
			\frac{\partial^2 h_i}{\partial w_{\cdot i}^2}(c_i, w_{\cdot i}) &= -(n-1) \frac{\phi_i^2 \theta_i \bar F_i(c_i)}{(1-\phi_i w_{\cdot i})^2}.
		\end{split}
	\end{equation}
	Notice that the second derivative in this expression is always strictly negative. Hence, given $c_i$, the optimization problem over $w_{\cdot i}$ is strictly concave. This implies that the first-order conditions are sufficient, and that any optimal solution is unique. Let $w^*_{\cdot i}(c_i)$ denote the optimal solution given $c_i$. It must satisfy the following necessary first-order condition:
	$$\frac{\partial h_i}{\partial w_{\cdot i}}(c_i, w^*_{\cdot i}(c_i)) (w_{\cdot i} - w^*_{\cdot i}(c_i)) \le 0, \ \forall w_{\cdot i} \in [0,\phi_i^{-1}).$$
	Using~\eqref{eq:partials_w}, it is easy to check that this condition is satisfied by the following:
	\begin{equation}
		\label{sol:cent_w_given_c}
		w_{\cdot i}^*(c_i) = \begin{cases}
			\frac{1}{\phi_i}\left(1-\frac{\phi_i\theta_i 		\bar F_i(c_i)}{\mu_i}\right) & \text{ if } \frac{\phi_i\theta_i \bar F_i(c_i)}{\mu_i} \le 1 
			\\
			0 & \text{ otherwise.}
		\end{cases}
	\end{equation}
	This value is uniquely defined, and exists for any choice of $c_i$. We then rewrite each maximization in~\eqref{opt:cent_sep} as
	\begin{equation}
		\label{opt:cent_univar}
		\sup_{(c_i, w_{\cdot i}) \in \A_i} h_i(c_i, w_{\cdot i}) = \ \sup_{c_i \ge 0} \ h_i^*(c_i),
	\end{equation}
	where $h_i^*(c_i) = h_i(c_i, w_{\cdot i}^*(c_i))$.

	\paragraph{Existence of an Optimal Solution:}
	We now prove existence of an optimal solution to~\eqref{opt:cent_univar}. Observe that for large enough $c_i$, we will have $w_{\cdot i}^*(c_i) = \frac{1}{\phi_i}\left(1-\frac{\phi_i\theta_i 		\bar F_i(c_i)}{\mu_i}\right)$. For such $c_i$ we obtain
	\begin{equation}
		\label{eq:h_large_c}
		\begin{split}
			h^*_i(c_i) =& -rc_i + (n-1)\mu_i \left[\frac{1}{\phi_i}\left(1-\frac{\phi_i\theta_i 		\bar F_i(c_i)}{\mu_i}\right)\right] \\
			&  - \theta_i \bar F_i(c_i) \left[\Gamma(\eta_i;1) - (n-1) \log\left(\frac{\phi_i\theta_i 		\bar F_i(c_i)}{\mu_i}\right)\right].
		\end{split}
	\end{equation}
	As $c_i \to \infty$, we will have $\bar F_i(c_i) \to 0$. Since we can write
	
	$$\bar F_i(c_i)\log\left(\frac{\phi_i\theta_i \bar F_i(c_i)}{\mu_i}\right) = \bar F_i(c_i)\left[ \log\left(\frac{\phi_i\theta_i}{\mu_i}\right) + \log\bar F_i(c_i)\right],$$
	and $x \log x \underset{x\to 0}{\to} 0$, we will have $\lim_{c_i \to \infty} h^*_i(c_i) = -\infty$.

	This limit is sufficient for existence of an optimal solution to~\eqref{opt:cent_univar}. Fix some $K< 0$. Since we have shown $h^*_i(c_i) \underset{c_i \to \infty}{\to} -\infty$, we know that $\exists C \in \R_+: h^*_i(c_i) < K, \ \forall c_i > C$. By continuity of $h_i^*$, the set $\mathcal{B} = \left\{c_i \in \R_+ : h_i^*(c_i) \ge K\right\}$ is compact. We can conclude by the Extreme Value Theorem that there exists a globally optimal value of $h^*_i$ within $\mathcal{B}$. Moreover, as long as $\mathcal{B}$ is non-empty, any point in $\mathcal{B}$ achieves higher objective value than any point in its compliment. By taking $K$ to be a large enough negative number, we can ensure that $\mathcal{B} \neq \emptyset$.

	\paragraph{System of Equations for Optimum:}
	The expression~\eqref{sol:cent_w_given_c} gives us the second equation in the system~\eqref{sol:cent}. For the other equation, we must analyze the first-order condition for $c_i$ in~\eqref{opt:cent_sep}. Taking derivatives with respect to $c_i$, we obtain
	\begin{equation}
		\label{eq:partials_c}
		\begin{split}
			\frac{\partial h_i}{\partial c_i}(c_i, w_{\cdot i}) &= -r + \theta_i f_i(c_i)\Big[  \Gamma(\eta_i;1) + (n-1) \Gamma(\phi_i w_{\cdot i};1)  \Big]
			\\
			\frac{\partial^2 h_i}{\partial c_i^2}(c_i, w_{\cdot i}) &= - \theta_i f_i'(c_i)\Big[  \Gamma(\eta_i;1) + (n-1) \Gamma(\phi_i w_{\cdot i};1)  \Big].
		\end{split}
	\end{equation}
	Notice that the second derivative is also negative everywhere -- although this does not imply that the objective function $h_i$ is concave. We proceed similarly as before, seeking to define an optimal value of $c_i$ for any given $w_{\cdot i}$. Let this be denoted $c_i^*(w_{\cdot i})$. It must satisfy:
	$$\frac{\partial h_i}{\partial c_i}(c_i^*(w_{\cdot i}), w_{\cdot i}) (c_i - c_i^*(w_{\cdot i})) \le 0, \ \forall c_i \in \R_+.$$
	Using~\eqref{eq:partials_c}, we can see that this will be satisfied whenever
	\begin{equation}
		\label{sol:cent_c_given_w}
		c_i^*(w_{\cdot i}) = 
		\begin{cases}
			f_i^{-1}\left(\frac{r}{\theta_i \left[ \Gamma(\eta_i;1) + (n-1) \Gamma(\phi_i w_{\cdot i};1)\right]}\right) &\text{ if } f_i(0) \le \frac{r}{\theta_i \left[ \Gamma(\eta_i;1) + (n-1) \Gamma(\phi_i w_{\cdot i};1)\right]} \\
			0 & \text{ otherwise.}
		\end{cases}  
	\end{equation}
	With~\eqref{sol:cent_w_given_c}, we obtain the system~\eqref{sol:cent}.

	\paragraph{Uniqueness:}
	It remains only to show that the optimal solution to~\eqref{opt:cent_univar} is unique. We return to our analysis of the univariate optimization problem in~\eqref{opt:cent_univar}. The necessary first-order condition for optimality of $c_i^*$ is
	\begin{equation}
		\label{eq:cent_foc_c}
		\frac{d h_i^*}{dc_i}(c_i^*)(c_i - c_i^*) \le 0, \ \forall c_i \in \R_+.
	\end{equation}
	We proceed by showing that there exists only a single $c_i^*$ satisfying this expression, and since existence has been proved, it must be the optimal solution. Recall that $\tilde c_i = F_i^{-1} \left(\left[1 - \frac{\mu_i}{\phi_i \theta_i}\right]_+\right)$, and we have $w_{\cdot i}^*(c_i) = 0$ if and only if $c_i \le \tilde c_i$.

	The reduced objective function $h_i^*(c_i)$, after substituting in~\eqref{sol:cent_w_given_c}, can be written as:
	\begin{equation}
		\label{eq:h*i}
		\begin{split}
			h_i^*(c_i) & = -rc_i - \theta_i\bar F_i(c_i) \Gamma(\eta_i;1) \\
			& \quad + 
			\begin{cases}
				(n-1) \left[\frac{\mu_i}{\phi_i}  + \theta_i \bar F_i(c_i)\left(\log\left(\frac{\phi_i \theta_i \bar F_i(c_i)}{\mu_i}\right) -1 \right)\right]& \text{ if } c_i \ge \tilde c_i \\
				0 & \text { otherwise.}	
			\end{cases}
		\end{split}
	\end{equation}
	Taking the derivative with respect to $c_i$, we obtain
	\begin{equation}
		\label{eq:dh*i_c}
		\begin{split}
			\frac{d h_i^*}{ d c_i}(c_i) &= -r + \theta_i f_i(c_i)\Gamma(\eta_i;1) - \begin{cases}
				\theta_i f_i(c_i) (n-1) \log\left(\frac{\phi_i \theta_i \bar F_i(c_i)}{\mu_i}\right)& \text{ if } c_i \ge \tilde c_i \\
				0 & \text { otherwise,}	
			\end{cases}
		\end{split}
	\end{equation}
	and the second derivative equals
	\begin{equation}
		\label{eq:d2h*i_c}
		\begin{split}
			\frac{d^2 h_i^*}{ d c_i^2}(c_i) &= \theta_i f_i'(c_i)\Gamma(\eta_i;1) + \theta_i (n-1)
			\begin{cases}
				\frac{f_i(c_i)^2}{1 - F_i(c_i)}  - f_i'(c_i)  	\log\left(\frac{\phi_i \theta_i \bar F_i(c_i)}{\mu_i}\right) & \text{ if } c_i \ge \tilde c_i \\
				0 & \text { otherwise.}	
			\end{cases}
		\end{split}
	\end{equation}
	Note that we are evaluating the right derivatives at $c_i = \tilde c_i$, where this function is not differentiable.

	In the regime $c_i < \tilde c_i$, we will always have $\frac{d^2 h_i^*}{ d c_i2}(c_i) < 0$. If this were also true for $c_i \ge \tilde c_i$, then the objective function would be strictly concave, and uniqueness would follow. We now prove that if $\frac{d^2 h_i^*}{ d c_i^2}(x) < 0$, then $h_i^*(\cdot)$ is strictly concave on $[x,\infty)$. In particular, by plugging in $x = \tilde c_i$ we conclude uniqueness of the optimum.

	Let us compute an additional derivative of $h_i^*(\cdot)$:
	\begin{equation}
		\begin{split}
			\frac{d^3 h_i^*}{dc^3_i}(c_i) &= \theta_i f_i''(c_i) \Gamma(\eta_i;1) \\
			& \quad + \theta_i (n-1)
			\begin{cases}
				\frac{f_i(c_i)^2}{\bar F_i(c_i)}\left[\frac{f_i(c_i)}{\bar F_i(c_i)} + 3 \frac{f_i'(c_i)}{f_i(c_i)}\right] -  f_i''(c_i)  \log\left(\frac{\phi_i\theta_i \bar F_i(c_i)}{\mu_i}\right) &\text{ if } c_i \ge \tilde c_i \\
				0 & \text{ otherwise.}
			\end{cases}
		\end{split}
	\end{equation}
	Observe that when we have $c_i \ge \tilde c_i$, a bit of algebra yields
	\begin{equation}
		\label{eq:h'''_from_h''}
		\frac{d^3 h_i^*}{dc^3_i}(c_i) = \frac{f_i''(c_i)}{f_i'(c_i)} \frac{d^2 h_i^*}{dc^2_i}(c_i)+ \frac{(n-1) \theta_i f_i(c_i)^2}{\bar F_i(c_i)} \left[\frac{f_i(c_i)}{\bar F_i(c_i)} + 3 \frac{f_i'(c_i)}{f_i(c_i)} - \frac{f_i''(c_i)}{f_i'(c_i)}\right].
	\end{equation}
	If, as assumed in this Proposition, we have $\frac{f_i(c_i)}{\bar F_i(c_i)} + 3 \frac{f_i'(c_i)}{f_i(c_i)} - \frac{f_i''(c_i)}{f_i'(c_i)} < 0$ for all $c_i \ge 0$, then it will follow that
	$$\frac{d^3 h_i^*}{dc^3_i}(c_i) < \frac{f_i''(c_i)}{f_i'(c_i)} \frac{d^2 h_i^*}{dc^2_i}(c_i).$$
	Applying Gr\"{o}nwall's inequality, we see that 
	$$\frac{d^2 h_i^*}{dc^2_i}(b) < \frac{d^2 h_i^*}{dc^2_i}(a) \exp\left(\int_{a}^b \frac{f_i''(s)}{f_i'(s)} ds\right),$$
	for any $\tilde c_i \le a < b$. As a consequence, if $\frac{d^2 h_i^*}{dc^2_i}(a) \le 0$, then $\frac{d^2 h_i^*}{dc^2_i}(b) < 0$ for all $b > a$.

	Rewriting~\eqref{eq:d2h*i_c}, we obtain:
	\begin{equation}
		\label{eq:d2h*i_c_tilde}
		\frac{d^2 h_i^*}{ d c_i^2}(\tilde c_i) = 
		\begin{cases}
			\theta_i f_i'(0)\left[\Gamma(\eta_i;1) - (n-1) \log\left(\frac{\phi_i \theta_i}{\mu_i}\right)\right] + \theta_i (n-1) f_i(0)^2 & \text{ if } \tilde c_i = 0 \\
			\theta_i f_i'(\tilde c_i)\Gamma(\eta_i;1) + \theta_i (n-1) \frac{\phi_i \theta_i f_i(\tilde c_i)^2}{\mu_i} & \text{ otherwise.}
		\end{cases}
	\end{equation}
	For $i$ satisfying
	\begin{equation}
		\label{eq:cent_gamma_cond_d2h}
		\Gamma(\eta_i;1) > 
		\begin{cases}
			(n-1) \left[ \log\left(\frac{\phi_i \theta_i}{\mu_i}\right) - \frac{f_i(0)^2}{f_i'(0)}\right] & \text{ if } \tilde c_i = 0 \\[2ex]
			- (n-1) \frac{\phi_i \theta_i f_i(\tilde c_i)^2}{\mu_i f_i'(\tilde c_i)} & \text{ otherwise.}
		\end{cases}
	\end{equation}
	in the assumption~\eqref{eq:cent_gamma_cond}, we see that $\frac{d^2 h_i^*}{ d c_i^2}(\tilde c_i) < 0$. By our application of Gr\"{o}nwall's inequality, we can conclude that $ h_i^*$ must be strictly concave, and hence the optimum is unique.

	Now, we turn to the banks $i$ satisfying
	\begin{equation}
		\label{eq:cent_gamma_cond_dh}
		\Gamma(\eta_i;1) > 
		\begin{cases}
			\frac{r}{\theta_i f_i(0)} + (n-1) \log\left(\frac{\phi_i \theta_i}{\mu_i}\right) & \text{ if } \tilde c_i = 0 \\[1ex]
			\frac{r}{\theta_i f_i(\tilde c_i)} & \text{ otherwise.}
		\end{cases}
	\end{equation}
	We can compute:
	\begin{equation}
		\label{eq:dh*i_c_tilde}
		\frac{d h_i^*}{ d c_i}(\tilde c_i) = -r + 
		\begin{cases}
			\theta_i f_i(0)\left[\Gamma(\eta_i;1)  - 
			(n-1) \log\left(\frac{\phi_i \theta_i}{\mu_i}\right) \right]& \text{ if } \tilde c_i = 0 \\
			\theta_i f_i(\tilde c_i)\Gamma(\eta_i;1) & \text { otherwise.}	
		\end{cases}
	\end{equation}
	By~\eqref{eq:cent_gamma_cond_dh}, we have $\frac{d h_i^*}{ d c_i}(\tilde c_i) > 0$. Since $\frac{d^2 h_i^*}{ d c_i^2}(c_i) < 0$ for all $c_i < \tilde c_i$, we cannot have any points satisfying the first-order condition~\eqref{eq:cent_foc_c} in $[0,\tilde c_i]$. However, we do know that there must exist an optimal solution, so therefore it must lie within $(\tilde c_i, \infty)$. At such a point $c_i^*$, we must have $\frac{d h_i^*}{ d c_i}(c_i^*) = 0$, and also $\frac{d^2 h_i^*}{ d c_i^2}(c_i^*) \le 0$, which are the two necessary conditions for optimality of $c_i^*$ when it lies in the interior of the feasible region. By the same conclusion using Gr\"{o}nwall's inequality, we must have $\frac{d^2 h_i^*}{ d c_i^2}(c_i) < 0$, and hence $\frac{d h_i^*}{ d c_i}(c_i) < 0$ for any $c_i > c_i^*$. Hence, only this choice of $c_i^*$ will satisfy the necessary first-order conditions, and as a result it must be unique.

	Since we require all $i$ to satisfy at least one of~\eqref{eq:cent_gamma_cond_dh} or~\eqref{eq:cent_gamma_cond_d2h}, the optimal solutions to each of the $n$ optimization problems in~\eqref{opt:cent_sep} must be unique.
\end{proof}

\begin{proof}[Proof of Corollary~\ref{cor:cent_verification}]
	The proof of this result mirrors the proof of Corollary~\ref{cor:dec_verification}, and therefore we omit many details.

	Fix some time $t< T$, at which we have $X^i_t = x_i$. We again choose some admissible controls $\left\{c^{\cdot}_s, w^{\cdot \cdot}_s\right\}_{s\in [t,T]}$. We then apply It\^{o}'s formula, which only differs in yielding a few more terms. Namely, we will need to use the generator defined in Section~\eqref{eq:cent_generator}, and the stochastic integrands will be slightly more complex. Next, to apply dominated convergence, our choice of the stopping time $\tau^n$ must ensure that each of the wealth processes $\left\{X^1_s\right\}_{s\ge 0}\dots \left\{X^n_s\right\}_{s\ge 0}, $ is bounded at time $\tau^n$. Therefore, we choose
	$$\tau^n = \left(T - \frac{1}{n}\right) \wedge \min_i \left\{ \inf \left\{ s \ge t, |X^i_s - X^i_t| \ge n  \right\}\right\}$$
	and conclude identically.
\end{proof}

\subsection{Differences in Optima}
\label{ssec:pfs-diff}

\begin{lemma}
	\label{lem:PoA_exp_c}
	Assume that $w_{\cdot i}^* > 0$. If the shock density satisfies:
	$f_i(x)  \ge \kappa_{i,L} e^{-\frac{x}{\lambda_{i,L}}},$
	for all $x$ and fixed constants $\lambda_{i,L} > 0$ and $\kappa_{i,L} > 0$, then 
	
	$$c_i^* \ge \lambda_{i,L}\log \left(\frac{\theta_i \kappa_{i,L} \Gamma(\phi_i \hat w_{\cdot i};1)}{r}\right) + \lambda_{i,L} \log(n-1).$$
	In particular, the planner's optimal \change{ cash reserves} asymptotically grow at least logarithmically in $n$.
	
	
Furthermore, if for all $x$ we also have:
	$$f_i(x) \le \kappa_{i,U} e^{-\frac{x}{\lambda_{i,U}}}$$
	for $\lambda_{i,L} \le \lambda_{i,U}$ and $\kappa_{i,L} \le \kappa_{i,U}$, then 
	\begin{enumerate}[(i)]
		\item \textbf{Upper Bound:}
		\begin{equation}
			c_i^* \le \lambda_{i,U} \log\left(\frac{\theta_i \kappa_{i,U} C_U}{r}\right) + \lambda_{i,U} \log\left((n-1) \log(n) \right),
		\end{equation}
		where $C_U > 3$ depends on all model parameters (including $\lambda_{i,L}$ and $\lambda_{i,U}$), but does not explicitly grow with $n$. 
		As a result, $\lim_{n\to \infty}\frac{c_i^*}{\log(n)} \le \lambda_{i,U}.$
		
		\item \textbf{Lower Bound:}
		\begin{equation}
			c_i^* \ge \lambda_{i,L} \log\left(\frac{\theta_i \kappa_{i,L} \lambda_{i,L}}{r\lambda_{i,U}}\right) + \lambda_{i,L} \log\left((n-1) \Big[\log(n-1) - \frac{\lambda_{i,U}}{\lambda_{i,L}} \log\left(C_L \right)\Big]\right),
		\end{equation}
		for $C_L > 0$ depending only on $i$'s parameters. 
		Hence, $\lim_{n\to \infty}\frac{c_i^*}{\log(n)} \ge \lambda_{i,L}$.
		
	\end{enumerate}
	
	Combining the two limiting bounds, we have $c_i^* = \Theta\big(\log(n)\big).$
\end{lemma}

The proof follows from iterating through upper (and lower) bounds for $c_i^*$ using the system of equations in~\eqref{sol:cent}, and beginning from crude estimates. It is possible to use the same techniques in this proof to obtain bounds when the density has power-law tails. While the results are not qualitatively different, we are unable to achieve the tight bound that appears in Proposition~\ref{prop:PoA_exp_exact_c} when the shock distribution is itself a power-law. The main result can be seen in Appendix~\ref{app:PoA_pwr}.

\begin{proof}[Proof of Lemma~\ref{lem:PoA_exp_c}]
	The main idea in this proof is to first establish crude bounds of:
	$$\hat c_i \le c_i^* \le K n^2,$$
	for a suitable choice of $K$. This then allows us to improve the bounds on $c_i^*$ itself through the relationship
	$$c_i^* = f_i^{-1}\left(\frac{r}{\theta_i \left[\Gamma(\eta_i;1) - (n-1) \log\left(\frac{\phi_i \theta_i \bar F_i(c_i^*)}{\mu_i} \right)\right]}\right),$$
	using the assumptions of a super- and sub-exponential density.

	Through a direct computation with the explicit solutions in Propositions~\ref{prop:dec_hjb_soln} and~\ref{prop:cent_hjb_soln}, we can write
	\begin{equation}
		\begin{alignedat}{3}
			V(t, x_1,\dots,x_n) - \sum_{i=1}^n V_i(t, x_i) &= (T-t) && \left[J^*_C - \sum_{i=1}^n J^*_i\right] \\
			&= (T-t) &&\sum_{i=1}^n\Big[-r(c_i^* - \hat c_i) + (n-1) \mu_i (w_{\cdot i}^* - \hat w_{\cdot i}) \\
			& &&  - \theta_i  \bar F_i(c_i^*) \big[\Gamma(\eta_i;1) + (n-1) \Gamma(\phi_i w_{\cdot i}^*) \big] \\
			& &&  + \theta_i  \bar F_i(\hat c_i) \big[\Gamma(\eta_i;1) + (n-1) \Gamma(\phi_i \hat w_{\cdot i}) \big] \\
		\end{alignedat}
	\end{equation}
	Observe that using the definitions, we have $w_{\cdot i}^* - \hat w_{\cdot i} = \frac{\theta_i}{\mu_i}\left(\bar F_i(\hat c_i) - \bar F_i( c_i^*)\right)$. Plugging this expression in and rearranging terms, we obtain:
	\begin{equation}
		\label{eq:val_gap_reduced}
		\begin{split}
			\frac{g(t) - \sum_{i=1}^n g_i(t)}{T-t} =&  \sum_{i=1}^n\Bigg[-r(c_i^* - \hat c_i) + \theta_i \left(\bar F_i(\hat c_i) - \bar F_i( c_i^*)\right) \big[(n-1) + \Gamma(\eta_i;1)\big] \\
			& + \theta_i (n-1) \Big[\bar F_i(\hat c_i)\Gamma(\phi_i \hat w_{\cdot i};1) - \bar F_i(c_i^*)\Gamma(\phi_i w_{\cdot i}^*;1)\Big] \Bigg].
		\end{split}
	\end{equation}
	Since we know the gap in~\eqref{eq:val_gap_reduced} must be positive, we can write:
	\begin{equation}
		\begin{split}
			\sum_{i=1}^n r c_i^* \le &  \sum_{i=1}^n\Bigg[r \hat c_i + \theta_i \left(\bar F_i(\hat c_i) - \bar F_i( c_i^*)\right) \big[(n-1) + \Gamma(\eta_i;1)\big] \\
			& + \theta_i (n-1) \Big[\bar F_i(\hat c_i)\Gamma(\phi_i \hat w_{\cdot i};1) - \bar F_i(c_i^*)\Gamma(\phi_i w_{\cdot i}^*;1)\Big] \Bigg] \\
			\le &  \sum_{i=1}^n\Bigg[r \hat c_i + \theta_i \bar F_i(\hat c_i) \big[(n-1) + \Gamma(\eta_i;1)\big] \\
			& + \theta_i (n-1) \Big[\bar F_i(\hat c_i)\Gamma(\phi_i \hat w_{\cdot i};1) \Big] \Bigg], \\
		\end{split}
	\end{equation}
	which follows by dropping the final term and since $\bar F_i(c_i^*)\ge 0$. A crude bound implies that 
	\begin{equation}
		\begin{split}
			rc_i^* & \le \sum_{i = 1}^n (n-1) \Big[r \hat c_i + \theta_i \bar F_i(\hat c_i) \big[1 + \Gamma(\eta_i;1) + \Gamma(\phi_i \hat w_{\cdot i};1)\big]\Big] \\
			c_i^* &\le K n^2,
		\end{split}
	\end{equation}
	where $K = \max_i \Big\{ \hat c_i + \frac{\theta_i}{r} \bar F_i(\hat c_i) \big[1 + \Gamma(\eta_i;1) + \Gamma(\phi_i \hat w_{\cdot i};1)\big] \Big\}$ does not depend explicitly on $n$. Since $w_{\cdot i}^* \ge 0$, it is also easy to see that $c_i^* \ge \hat c_i$. Both these bounds will be useful starting points for the proof.

	\begin{enumerate}[(i)]
		\item 
		\textbf{Upper Bound:} 
		We first prove the upper bound for $c_i^*$. First, since $f_i(x) \le \kappa_{i,U} e^{-\frac{x}{\lambda_{i,U}}}$ and both functions are decreasing, we will have $f_i^{-1} (y) \le \lambda_{i,U} \log\left(\frac{\kappa_{i,U}}{y}\right)$, and it follows from the system of equations~\eqref{sol:cent} that 
		\begin{equation}
			c_i^* \le \lambda_{i,U} \log\left(\frac{\theta_i \kappa_{i,U} \left[\Gamma(\eta_i;1) - (n-1) 	\log\left(\frac{\phi_i \theta_i}{\mu_i}\bar F_i(c_i^*)\right)\right]}{r}\right).
		\end{equation}
		Now, using $f_i(x) \ge \kappa_{i,L} e^{-\frac{x}{\lambda_{i,L}}}$, we know that $\bar F_i(c_i^*) = \int_{c_i^*}^\infty f_i(u) du \ge \kappa_{i,L} \lambda_{i,L} e^{-\frac{c_i^*}{\lambda_{i,L}}}$, and write:
		\begin{equation}
			\label{eq:c_exp_ub}
			\begin{split}
				c_i^* & \le  \lambda_{i,U} \log\left(\frac{\theta_i \kappa_{i,U} \left[\Gamma(\eta_i;1) - (n-1) \log\left(\frac{\phi_i \theta_i \kappa_{i,L} \lambda_{i,L}}{\mu_i}\right) +(n-1) \frac{c_i^*}{\lambda_{i,L}}\right]}{r}\right)\\
				&\le \lambda_{i,U} \log\left(\frac{\theta_i \kappa_{i,U} \left[\Gamma(\eta_i;1) - (n-1) \log\left(\frac{\phi_i 	\theta_i \kappa_{i,L} \lambda_{i,L}}{\mu_i} \wedge 1\right) +(n-1) \frac{c_i^*}{\lambda_{i,L}}\right]}{r}\right).
			\end{split}
		\end{equation}
		Since each of the three terms in the brackets is non-negative, we can upper bound this quantity by:	
		\begin{equation}
			c_i^* \le \lambda_{i,U} \log\left(\frac{\theta_i \kappa_{i,U} \left[\Gamma(\eta_i;1) -  \log\left(\frac{\phi_i 	\theta_i \kappa_{i,L} \lambda_{i,L}}{\mu_i} \wedge 1\right) + \lambda_{i,L}^{-1} \right]n c_i^*}{r}\right),
		\end{equation}
		and we define $D = \Gamma(\eta_i;1) -  \log\left(\frac{\phi_i \theta_i \kappa_{i,L} \lambda_{i,L}}{\mu_i} \wedge 1\right) + \lambda_{i,L}^{-1}$ for convenience. Recall that we obtained a crude upper bound of $c_i^* \le K n^2$, which, when plugged in, yields:
		\begin{equation}
			\begin{split}
				c_i^* & \le \lambda_{i,U} \log\left(\frac{\theta_i \kappa_{i,U} D K n^3}{r}\right).
			\end{split}
		\end{equation}
		This is a significantly tighter bound than $Kn^2$. Therefore, we plug it back into~\eqref{eq:c_exp_ub}. By simplifying and bounding the term in the logarithm, we compute:
		\small
		\begin{equation}
			\begin{split}	
				\frac{c_i^*}{\lambda_{i,U}} &\le   \log\left(\frac{\theta_i \kappa_{i,U} \left[\Gamma(\eta_i;1) - (n-1) \log\left(\frac{\phi_i 	\theta_i \kappa_{i,L} \lambda_{i,L}}{\mu_i} \right) +(n-1) \frac{\lambda_{i,U}}{\lambda_{i,L}} \log\left(\frac{\theta_i \kappa_{i,U} D K n^3}{r}\right)\right]}{r}\right) \\
				&\le  \log\left(\frac{\theta_i \kappa_{i,U} \left[\Gamma(\eta_i;1) + (n-1) \left[ \log\left(\frac{\mu_i}{\phi_i 	\theta_i \kappa_{i,L} \lambda_{i,L}} \left(\frac{\theta_i \kappa_{i,U} D K}{r}\right)^{\frac{\lambda_{i,U}}{\lambda_{i,L}}} \vee 1 \right) + 3\frac{\lambda_{i,U}}{\lambda_{i,L}} \log(n)\right]\right]}{r}\right).
			\end{split}	
		\end{equation}
		\normalsize
		Notice that $\Gamma(\eta_i;1)\ge 0$, $\log\left(\frac{\mu_i}{\phi_i 	\theta_i \kappa_{i,L} \lambda_{i,L}} \left(\frac{\theta_i \kappa_{i,U} D K}{r}\right)^{\frac{\lambda_{i,U}}{\lambda_{i,L}}} \vee 1 \right) \ge 0$. Therefore, we can write
		\small
		\begin{equation}
			\frac{c_i^*}{\lambda_{i,U}} \le \log\left(\frac{\theta_i \kappa_{i,U} \left[\Gamma(\eta_i;1) + \log\left(\frac{\mu_i}{\phi_i 	\theta_i \kappa_{i,L} \lambda_{i,L}} \left(\frac{\theta_i \kappa_{i,U} D K}{r}\right)^{\frac{\lambda_{i,U}}{\lambda_{i,L}}} \vee 1 \right)+ 3\frac{\lambda_{i,U}}{\lambda_{i,L}}\right](n-1)\log(n)}{r}\right),
		\end{equation}
		\normalsize
		and after simplification we obtain the desired bound of:
		\begin{equation}
			c_i^* \le \lambda_{i,U} \log\left(\frac{\theta_i \kappa_{i,U} C_U}{r}\right) + \lambda_{i,U} \log\left((n-1) \log(n) \right),
		\end{equation}
		where $C_U = \Gamma(\eta_i;1) + \log\left(\frac{\mu_i}{\phi_i 	\theta_i \kappa_{i,L} \lambda_{i,L}} \left(\frac{\theta_i \kappa_{i,U} D K}{r}\right)^{\frac{\lambda_{i,U}}{\lambda_{i,L}}} \vee 1 \right)+ 3\frac{\lambda_{i,U}}{\lambda_{i,L}}$. Observe that $C_U$ does not depend explicitly on $n$, but through $K$ it will be a function of parameters throughout the system.

		Finally, it follows that $\lim_{n\to \infty}\frac{c_i^*}{\log(n)} \le \lambda_{i,U}$.
		\item 
		\textbf{Lower Bound:} 
		We proceed with the lower bound identically. With our assumption of $f_i(x) \ge \kappa_{i,L} e^{-\frac{x}{\lambda_{i,L}}}$, we know 
		\begin{equation}
			\label{eq:c_exp_lb_simple}
			c_i^*  \ge \lambda_{i,L} \log\left(\frac{\theta_i \kappa_{i,L} \left[\Gamma(\eta_i;1) - (n-1) \log\left(\frac{\phi_i \theta_i}{\mu_i}\bar F_i(c_i^*)\right)\right]}{r}\right).
		\end{equation}
		Moreover, since $\Gamma(\eta_i;1) \ge 0$ this term can be dropped to obtain:
		\begin{equation}
			\label{eq:c_exp_lb}
			c_i^* \ge \lambda_{i,L} \log\left(\frac{-\theta_i \kappa_{i,L}  (n-1) \log\left(\frac{\phi_i \theta_i}{\mu_i}\bar F_i(c_i^*)\right)}{r}\right).
		\end{equation}
		By plugging in the initial crude bound of $c_i^*\hat c_i$, and since $\Gamma(\phi_i \hat w_{\cdot i};1) = - \log\left(\frac{\phi_i \theta_i}{\mu_i}\bar F_i(\hat c_i)\right)$ by definition, we can compute a tighter lower bound for $c_i^*$ of
		\begin{equation}
			\label{eq:c_lb_superexp}
			c_i^* \ge \lambda_{i,L} \log\left(\frac{\theta_i \kappa_{i,L} (n-1) \Gamma(\phi_i \hat w_{\cdot i};1)}{r}\right).
		\end{equation}
		This is precisely the lower bound in the first part of Lemma~\ref{lem:PoA_exp_c}. Note that for this result, we needed only the lower bound on $f_i(\cdot)$, through which~\eqref{eq:c_exp_lb_simple} follows.

		We now continue and prove the tighter lower bound, which requires the upper bound on $f_i(\cdot)$. In particular, we assumed that $f_i(x) \le \kappa_{i,U} e^{-\frac{x}{\lambda_{i,U}}}$, and it follows that $\bar F_i(c_i^*) \le \kappa_{i,U} \lambda_{i,U} e^{-\frac{c_i^*}{\lambda_{i,U}}}$. With~\eqref{eq:c_lb_superexp}, we can compute an improved upper bound of:
		\begin{equation}
			\bar F_i(c_i^*) \le \kappa_{i,U} \lambda_{i,U} \left(\frac{r} {\theta_i  \kappa_{i,L} (n-1) \Gamma(\phi_i \hat w_{\cdot i})}\right)^{\frac{\lambda_{i,L}}{\lambda_{i,U}}}.
		\end{equation}
		This upper bound on $f_i(\cdot)$ also implies that $\hat c_i \le \lambda_{i,U} \log\left(\frac{\theta_i \kappa_{i,U} \Gamma(\eta_i;1)}{r}\right)$. Similarly, the assumed $f_i(x) \ge \kappa_{i,L} e^{-\frac{x}{\lambda_{i,L}}}$ will give us $\bar F_i(\hat c_i) \ge \kappa_{i,L} \lambda_{i,L} e^{-\frac{\hat c_i}{\lambda_{i,L}}}$. Putting the two together, we will have 
		$$\bar F_i(\hat c_i) \ge \kappa_{i,L}\lambda_{i,L} \left(\frac{r}{\theta_i \kappa_{i,U}\Gamma(\eta_i;1)}\right)^{\frac{\lambda_{i,U}}{\lambda_{i,L}}},$$
		and it follows that 
		$$\bar F_i(c_i^*) \le \bar F_i(\hat c_i)  \frac{\kappa_{i,U} \lambda_{i,U}}{\kappa_{i,L}\lambda_{i,L}} \left(\frac{r}{\theta_i}\right)^{\frac{\lambda_{i,L}}{\lambda_{i,U}} - \frac{\lambda_{i,U}}{\lambda_{i,L}}} \frac{\left(\kappa_{i,U}\Gamma(\eta_i;1)\right)^{\frac{\lambda_{i,U}}{\lambda_{i,L}}}}{\left(\kappa_{i,L}\Gamma(\phi_i \hat w_{\cdot i};1)\right)^{\frac{\lambda_{i,L}}{\lambda_{i,U}}}} (n-1)^{-\frac{\lambda_{i,L}}{\lambda_{i,U}}}.$$

		Let $C_L = \frac{\kappa_{i,U} \lambda_{i,U}}{\kappa_{i,L}\lambda_{i,L}} \left(\frac{r}{\theta_i}\right)^{\frac{\lambda_{i,L}}{\lambda_{i,U}} - \frac{\lambda_{i,U}}{\lambda_{i,L}}} \frac{\left(\kappa_{i,U}\Gamma(\eta_i;1)\right)^{\frac{\lambda_{i,U}}{\lambda_{i,L}}}}{\left(\kappa_{i,L}\Gamma(\phi_i \hat w_{\cdot i};1)\right)^{\frac{\lambda_{i,L}}{\lambda_{i,U}}}}$. Plugging this bound into~\eqref{eq:c_exp_lb}, we obtain:
		\begin{equation}
			\begin{split}
				c_i^* &\ge \lambda_{i,L} \log\left(\frac{-\theta_i \kappa_{i,L}  (n-1) \log\left(\frac{\phi_i \theta_i}{\mu_i}\bar F_i(\hat c_i)   C_L (n-1)^{-\frac{\lambda_{i,L}}{\lambda_{i,U}}} \right)}{r}\right) \\
				&\ge \lambda_{i,L} \log\left(\frac{-\theta_i \kappa_{i,L}  (n-1) \log\left( C_L (n-1)^{-\frac{\lambda_{i,L}}{\lambda_{i,U}}} \right)}{r}\right),
			\end{split}
		\end{equation}
		since $- \log\left(\frac{\phi_i \theta_i}{\mu_i}\bar F_i(\hat c_i)\right) = \Gamma(\phi_i \hat w_{\cdot i};1) \ge 0$, and hence this term can be dropped. Simplifying, we arrive at the desired bound of:
		\begin{equation}
			c_i^* \ge \lambda_{i,L} \log\left(\frac{\theta_i \kappa_{i,L}\lambda_{i,L}}{r\lambda_{i,U}}\right) + \lambda_{i,L} \log\left((n-1) \Big[\log(n-1) - \frac{\lambda_{i,U}}{\lambda_{i,L}} \log\left(C_L \right)\Big]\right),
		\end{equation}
		from which it follows that $\lim_{n\to \infty}\frac{c_i^*}{\log(n)} \ge \lambda_{i,L}$.
	\end{enumerate}

	Putting both $(i)$ and $(ii)$ together, we see that $c_i^* = \Theta\big(\log(n)\big)$.
\end{proof}

\begin{proof}[Proof of Proposition~\ref{prop:PoA_order}]
	Using Propositions~\ref{prop:dec_hjb_soln} and~\ref{prop:cent_hjb_soln}, we can compute
	\begin{equation}
		\label{eq:val_gap_exact}
		\begin{split}
			\frac{V - \sum_{i=1}^n V_i}{T-t} =&  \sum_{i=1}^n\Bigg[-r(c_i^* - \hat c_i) + \theta_i \left(\bar F_i(\hat c_i) - \bar F_i(c_i^*)\right) \big[(n-1) + \Gamma(\eta_i;1)\big] \\
			& + \theta_i (n-1) \Big[\bar F_i(\hat c_i)\Gamma(\phi_i \hat w_{\cdot i};1) - \bar F_i(c_i^*)\Gamma(\phi_i w_{\cdot i}^*;1)\Big] \Bigg],
		\end{split}
	\end{equation}
	where $V$ and $V_i$ are evaluated at $(t, x_1,\dots,x_n)$ and the difference becomes independent of wealths because of logarithmic utility. Notice that any of the terms in the sum will equal zero if $w_{\cdot i}^* = 0$ (in which case we also must also have $\hat w_{\cdot i} = 0$, and hence $\hat c_i = c_i^*$). If not, then using the results from Section~\ref{sec:diff} we see that
	$$\frac{-r(c_i^* - \hat c_i) + \theta_i \left(\bar F_i(\hat c_i) - \bar F_i(c_i^*)\right) \big[(n-1) + \Gamma(\eta_i;1)\big]}{n} \underset{n\to\infty}{\to} \theta_i \bar F_i(\hat c_i),$$
	since $c_i^* \asymp \log(n)$ and $\bar F_i(c_i^*) \to 0$. Moreover, we have seen that $(n-1)\bar F_i(c_i^*)\Gamma(\phi_i w_{\cdot i}^*;1) = \Theta(1)$. Since the sum is now of order $|\M_n|$, putting the two together yields
	$$\frac{V - \sum_{i=1}^n V_i}{T-t} = \Theta\left(n|\M_n|\right).$$
	In Proposition~\ref{prop:dec_hjb_soln}, it is easy to see that $V_i = (T-t) \Theta\left(|\M_n|\right)$, and therefore we obtain
	$$\frac{V}{\sum_{i=1}^n V_i} = 1 + \Theta(1),$$
	as desired.
\end{proof}

\begin{proof}[Proof of Corollary~\ref{cor:PoA_limit}]
	This proposition is proved easily by analyzing the value functions in Propositions~\ref{prop:dec_hjb_soln} and~\ref{prop:cent_hjb_soln}. We will use the notation of Section~\ref{sec:diff}, where $\hat c_\cdot$ indicates the decentralized optimum, and $c_{\cdot}^*$ indicates the centralized optimum (likewise for $w_{\cdot \cdot}$).

	We begin by analyzing the decentralized value function $V_i$. Using the explicit formula in Corollary~\ref{cor:dec_verification}, we write:
	\begin{equation}
		\frac{V_i}{|\M_n|(T-t)} = \frac{J_i^*}{|\M_n|} + \frac{\log x}{|\M_n|(T-t)},
	\end{equation}
	and see that the second term will go to zero as $n\to \infty$. Moreover, by assumption that all banks in $\M_n$ are homogeneous, we will have $\hat w_{ij} = \hat w_{ik}$ for any $j,k \in M_n$. This yields:
	\begin{equation}
		J_i^* = (1 - \hat c_i)r - \theta_i \bar F_i(\hat c_i) \Gamma(\eta_i;1) + |\M_n| \left[\mu \hat w - \theta \bar F(\hat c) \Gamma(\phi \hat w;1)\right],
	\end{equation}
	where $\hat c$ denotes the optimal liquidity supply held by any bank in $\M_n$, and $\hat w$ denotes the optimal investment made by any bank to those in $\M_n$. By using Eq~\eqref{sol:dec} to compute $\hat w$, we obtain:
	\begin{equation}
		\begin{split}
			J_i^* &= (1 - \hat c_i)r - \theta_i \bar F_i(\hat c_i) \Gamma(\eta_i;1) \\
			& \quad + |\M_n| \left[\frac{\mu}{\phi}\left(1 - \frac{\phi \theta \bar F(\hat c)}{\mu}\right) + \theta \bar F(\hat c) \log\left(\frac{\phi \theta \bar F(\hat c)}{\mu}\right)\right],
		\end{split}
	\end{equation}
	and the desired limit follows. We note that this expression for $J_i^*$ is only correct for $i \notin \M_n$, otherwise we would have a factor of $|\M_n|-1$ in front of the term in brackets. However, in the limit this difference disappears.

	The analysis of the centralized setting is almost identical, using the value function in Proposition~\ref{prop:cent_hjb_soln}, we have:

	\begin{equation}
		\frac{V}{n|\M_n|(T-t)} = \frac{J^*_C}{n|\M_n|} + \frac{\sum_{i = 1}^n \log x_i}{n|\M_n|(T-t)}.
	\end{equation}
	The only term of interest for large $n$ will be $J_C^*$, and by homogeneity within $\M_n$ we can see that:
	\begin{equation}
		J_C^* = |\M_n| (n-1) w^* \mu + \sum_{i=1}^n \Bigg( \left(1-c_i^*\right)r  - \theta_i \bar F_i(c_i^*)\Big[ \Gamma(\eta_i;1) + (n-1) \Gamma(\phi_i w^*;1) \Big] \Bigg),
	\end{equation}
	where $w^*$ denotes the optimal fractional amount invested into each bank in $\M_n$. Notice that only for bank in $\M_n$ will we have $c_i^*$ growing with $n$ (logarithmically). Moreover, from the analysis in Section~\ref{sec:diff}, we also know that $(n-1) \bar F_i(c_i^*) \Gamma(\phi_iw_{\cdot i}^*;1)$ is of constant order. Therefore, when dividing by $n |\M_n|$ and taking the limit, the sum will go to zero. Only the first term will remain, and we also know that $w^* \to \phi^{-1}$ as $n\to \infty$, which concludes.
	
	In order to show the limit for the price of anarchy, it is only necessary to sum $V_i$ over $n$ and divide.
\end{proof}

\section{Price of Anarchy: Super-/Sub-Power Distribution}
\label{app:PoA_pwr}

In this section, we perform similar calculations to the main result of Section~\ref{sec:diff}, but for shock size densities bounded by power law distributions. In particular, we have the following analogue of Lemma~\ref{lem:PoA_exp_c}:

\begin{prop}
	\label{aprop:PoA_pwr_c}
	If for all $x$ we have
	$f_i(x) \ge \kappa_{i,L} (\zeta_{i}^0+x)^{-\frac{1}{\alpha_{i,L}}},$
	for some constants $\alpha_{i,L} < 1$, $\kappa_{i,L} > 0$, and $\zeta_{i}^0 \ge 1$, then
	\begin{equation}
		\label{aeq:PoA_pwr_c_lb_simple}
		c_i^* \ge \left(\frac{- \kappa_{i,L} \theta_i (n-1) 	\log\left(\frac{\phi_i \theta_i}{\mu_i} \bar F_i(\hat c_i)\right) }{r}\right)^{\alpha_{i,L}} - \zeta_{i}^0.
	\end{equation}

	If, furthermore, the density satisfies
	$f_i(x) \le \kappa_{i,U} (\zeta_{i}^0+x)^{-\frac{1}{\alpha_{i,U}}},$
	with $\kappa_{i,U} \ge \kappa_{i,L}$ and $\alpha_{i,L} \le \alpha_{i,U} < 1$, then:
	\begin{enumerate}[(i)]
		\item \textbf{Upper Bound:}
		\begin{equation}
			c_i^* \le C_U \big[(n-1)\log(n)\big]^{\alpha_{i,U}} - \zeta_{i}^0,
		\end{equation}
		where $C_U$ depends on all model parameters, but does not explicitly grow with $n$. As a result, $\lim_{n\to \infty} \frac{c_i^*}{\big[(n-1)\log(n)\big]^{\alpha_{i,U}}} \le C_U$.
		
		\item \textbf{Lower Bound:}
		\begin{equation}
			c_i^* \ge 	\left(\frac{\kappa_{i,L}\theta_i}{r}(n-1)\left[\left(\frac{\alpha_{i,L}}{\alpha_{i,U}} - \alpha_{i,L}\right)\log(n-1) - \log(C_L)\right]\right)^{\alpha_{i,L}} - \zeta_i^0,
		\end{equation}
		for $C_L > 0$ depending only on $i$. Hence, $\lim_{n\to \infty} \frac{c_i^*}{\big[(n-1)\log(n)\big]^{\alpha_{i,L}}} \ge \left(\frac{\kappa_{i,L}\theta_i}{r} \left(\frac{\alpha_{i,L}}{\alpha_{i,U}} - \alpha_{i,L}\right)\right)^{\alpha_{i,L}}$.
	\end{enumerate}
\end{prop}

The proof follows an identical technique. In the special case where the shock density is indeed a power distribution, we have the following analogue of Proposition~\ref{prop:PoA_exp_exact_c}.

\begin{corollary}
	If $f_i(x) = \frac{\left(\frac{1}{\alpha_i}-1\right)\left(\zeta_{i}^0\right)^{\frac{1}{\alpha_i} - 1}}{(\zeta_i^0 + x)^{\frac{1}{\alpha_i}}}$, then 
	$$c_i^* = \Theta\left( \big[(n-1)\log(n)\big]^{\alpha_i} \right).$$
\end{corollary}

This result can be seen by simply plugging $\alpha_{i,L} = \alpha_{i,U} = \alpha_i$ into Proposition~\ref{aprop:PoA_pwr_c}.

This Corollary can be used to replicate the remaining analysis in Section~\ref{sec:diff}, but as the results are qualitatively similar, we omit these calculations.

\subsection{Proof of Proposition~\ref{aprop:PoA_pwr_c}}

\begin{proof}
	The proof of this result largely mirrors the proof of Lemma~\ref{lem:PoA_exp_c}. Recall that we have shown that 
	$$\hat c_i \le c_i^* \le Kn^2,$$
	for a suitable choice of $K$. By our assumptions on the density, it also follows that:
	$$\left(\frac{y}{\kappa_{i,L}}\right)^{-\alpha_{i,L}} - \zeta_{i}^0 \le f_i^{-1}(y) \le \left(\frac{y}{\kappa_{i,U}}\right)^{-\alpha_{i,U}} - \zeta_{i}^0,$$
	and
	$$\frac{\kappa_{i,L}}{\frac{1}{\alpha_{i,L}} - 1}(\zeta_{i}^0 + x)^{1 - \frac{1}{\alpha_{i,L}}} \le 1 - F_i(x) \le \frac{\kappa_{i,U}}{\frac{1}{\alpha_{i,U}} - 1}(\zeta_{i}^0 + x)^{1 - \frac{1}{\alpha_{i,U}}}.$$
	We can then follow the proof of Lemma~\ref{lem:PoA_exp_c} identically, but using these bounds instead.
\end{proof}

\change{

\section{Model Extensions}
\label{app:alternate_models}

\subsection{Liquidity Loss on Shock Arrival}
\label{subapp:liq_loss_shock}

In the main text of the paper, we make the assumption that the magnitude of a liquidity shock is important only if it exceeds the firm's cash reserves. 
We present an extension to our model that incorporates the magnitude of this loss directly into the dynamics of each firm's net capitalization.

Specifically, consider the event $d\tilde{N}^i_t = 1$, where a shock arrives for firm $i$.  In the main text, this event induces a jump in firm $i$'s capitalization process $X^i_t$ if and only if the size of the shock $\zeta^i_t$ is larger than the firm's available cash reserves $c^i_t$. However, we can build a model in which firms must themselves pay out any liquidity shocks, but only up to their maximum available supply of $c^i_t$. This feature is captured in the following equation for the dynamics of $X^i_t$.
\begin{align}
	\label{app:eq:liq_loss_shock_sde}
	\frac{dX^i_t}{X^i_t} &= \left( (1 - c^i_t) r  + \sum_{j\neq i}w_t^{ij} \mu_j + \frac{\eta_i }{\phi_i} \mu_i\right)dt - \sum_{j\neq i} w_t^{ij} \phi_j \,dN^j_t \\
	& \qquad  - (\eta_i \change{ \,+\, c^i_t}) \,dN^i_t \change{  \,-\, \zeta^i_t (d\tilde{N}^i_t - dN^i_t) }\, , \quad i=1,\cdots,n.
\end{align}
Notice that: i) when $dN^i_t = 1$, and a shock is larger than the firm's available liquidity, this liquidity is immediately lost to cover the shock, and ii) if $d \tilde{N}^i_t = 1$ but $dN^i_t = 0$, and the shock was smaller than the firm's supply of liquidity, then the firm only loses exactly $\zeta^i_t$ -- corresponding to their obligation to pay for the shock.

The value function for individual firm $i$, denoted $V_i(t,x)$ in the decentralized case, can be written identically to \eqref{eq:dec_val}, but under the same regularity conditions  of Proposition \ref{prop:dec_opt_hjb}, the associated non-local PDE it solves 
is now given by:
\begin{align}
	\label{app:eq:liq_loss_shock_pde}
	\begin{alignedat}{23}
		0&= \partial_t V_i  + \sup_{c_i,  w_{i\cdot}} \Bigg\{ 	&&\left[\left(1-c_i\right)r  + \sum_{j\neq i}w_{ij} \mu_j  + \frac{\eta_i \mu_i}{\phi_i}\right]x  \partial_x V_i + \theta_i \bar F_i(c_i)\Big[ V_i(t,x(1 - \eta_i \change{ \, -\,  c_i})) - V_i \Big] \\
		& && + \sum_{ j\neq i} \theta_j \bar F_j(c_j)\Big[ 	V_i(t,x(1-\phi_j w_{ij})) - V_i \Big] \\
		& && \change{ \, + \, \theta_i \int_{0}^{c_i} f_i(u)\left[V_i(t, x(1-u)) - V_i) du\right] }\Bigg\},
	\end{alignedat}
\end{align}
with terminal condition $V_i(T,x) = U_i(x)$. 

For simplicity, we will only consider the setting of a logarithmic utility function, (although the next step of separability can still be performed with power utility). 
We look for a separable solution to the PDE \eqref{app:eq:liq_loss_shock_pde} given by $V_i(t,x) = g_i(t) + \log(x)$. By plugging this ansatz into \eqref{app:eq:liq_loss_shock_pde} and simplifying, we can obtain the following ordinary differential equation for $g_i(t)$.
\begin{align}
	\label{app:eq:liq_loss_shock_ode_g}
	\begin{alignedat}{23}
		0&=  g_i'(t) &+ \sup_{c_i,  w_{i\cdot}} \Bigg\{ \left(1-c_i\right)r  + \sum_{j\neq i}w_{ij} \mu_j  + \frac{\eta_i \mu_i}{\phi_i}+ \theta_i \bar F_i(c_i) \log(1 - \eta_i \change{ \, - \, c_i}) \Big] \\
		&&+ \sum_{ j\neq i} \theta_j \bar F_j(c_j)\log(1-\phi_j w_{ij}) 
		\change{ \, + \, \theta_i \int_{0}^{c_i} f_i(u)\log(1-u) du }\Bigg\},
	\end{alignedat}
\end{align}
with terminal condition $g_i(T) = 0$.

Indeed the PDE solved by the value function has a separable solution, but now the first-order conditions for optimality solved by the optimal controls yield the following system of equations:
\begin{align}
	\label{app:eq:liq_loss_shock_opt_ctrls}
	0 &= \mu_j + \frac{\phi_j \theta_j \bar F_j(\hat c_j)}{1 - \phi_j \hat w_{ij}},\quad \text{ for all } j\neq i \\
	0 &= -r - \theta_i f_i(\hat c_i) \log(1-\eta_i \change{ \, -\, \hat c_i}) \change{ \, - \, \frac{\theta_i \bar F_i(\hat c_i)}{1 - \eta_i - \hat c_i} + \theta_i f_i(\hat c_i) \log(1-\hat c_i)}.
\end{align}
Although the second equation for the optimal control $\hat c_i$ still exists in isolation, this expression cannot be solved explicitly for the optimal controls -- certainly not for general choice of $\bar F(\cdot)$, and not even if we assume the liquidity shocks to be exponentially distributed. Hence, to study the optimal controls in this model, we would need to resort to computational methods.

Using the parameters in Table~\ref{tab:PoA_sim_params}, and varying the value of $\lambda$, we can see in Figure~\ref{fig:Liq_Loss_Shock_Arrival} that for large $\lambda$, and hence small expected size for liquidity shocks, the optimal cash reserves are largely identical. Since the first-order condition for the interbank investment amounts will only depend on the optimal cash reserves, this implies that the optimal allocations will coincide. As a result, the same results we obtain in the main text would hold.
\begin{figure}[!h]
	\centering
	\includegraphics[width=0.75\textwidth]{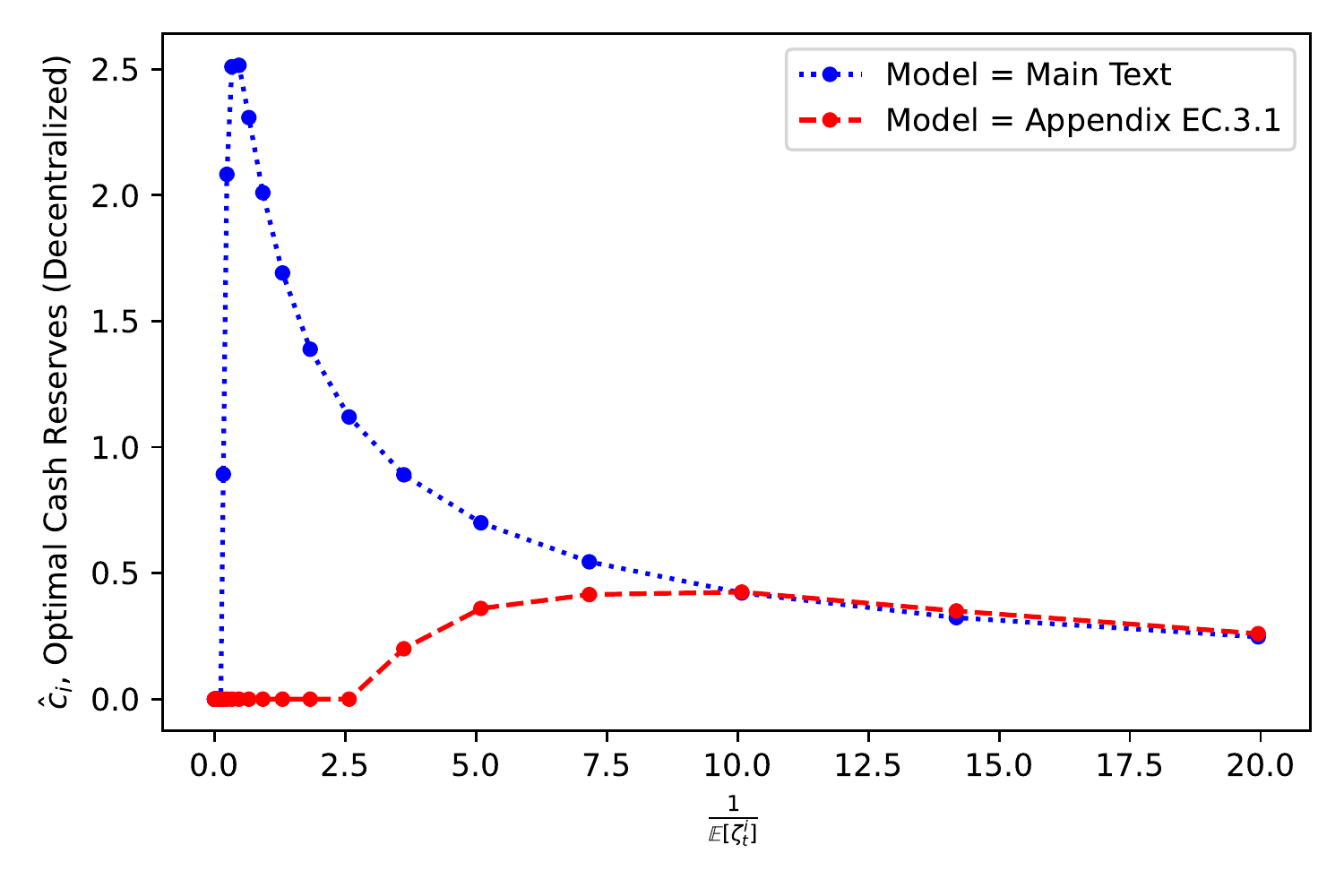}
	\captionsetup{width=.8\linewidth, font = small, justification=justified}
	\caption{Comparison of the optimal cash reserves when liquidity shocks directly induce losses regardless of size.}
	\label{fig:Liq_Loss_Shock_Arrival}
\end{figure}

\subsection{Partially Liquid Risk-Free Asset}
\label{subapp:partially_liq_rf_asset}

A further simplifying assumption made in this paper is that the risk-free asset, with constant rate of return $r$, does not provide a source of liquidity. In this section, we relax this assumption by allowing a firm's capital held in the risk-free asset to provide \textit{some} liquidity upon shock arrival.

The only substantial difference we must make, with respect to the model in the main text, concerns the construction of the thinned Poisson process $N^i_t$ from the shock arrival process $\tilde N^i_t$. We now define $dN^i_t$ as follows:
\begin{equation}
	\label{app:eq:partially_liq_rf_asset_thinned_cond}
	dN^i_t = 1 \iff d \tilde{N}^i_t = 1 \text{ and } \zeta^i_t > c^i_t \change{ \, +\, \alpha \left(1 - c^i_t - \sum_{j\neq i} w^{ij}_t\right)_+},
\end{equation}
where the parameter $\alpha \in (0,1)$ controls the fraction of value obtained by liquidating the risk-free asset immediately, and $(y)_+$ denotes the positive part of $y$. Note that some firms may borrow at the risk-free rate ($1 - c^i_t - \sum_{j\neq i} w^{ij}_t<0$), but obtain no additional liquidity. 

With $dN^i_t$ defined in this manner, the equation \eqref{eq:wealth_sde} for the dynamics of $X^i_t$ remains identical. We define the value function of firm $i$ in the same manner as the main text, and it can be shown that under regularity conditions, $V_i(t,x)$ solves:
\begin{align}
	\label{app:eq:partially_liq_rf_asset_pde}
	\begin{alignedat}{23}
		0&= \partial_t V_i  + \sup_{c_i,  w_{i\cdot}} \Bigg\{ 	&&\left[\left(1-c_i\right)r  + \sum_{j\neq i}w_{ij} \mu_j  + \frac{\eta_i \mu_i}{\phi_i}\right]x  \partial_x V_i \\
		& &&
		+ \theta_i \bar F_i\left(c_i \change{ \, + \,  \alpha \Big(1 - c_i - \sum_{j\neq i} w_{ij} \Big)_+} \right)\Big[ V_i(t,x(1 - \eta_i)) - V_i \Big] \\
		& && + \sum_{ j\neq i} \theta_j \bar F_j\left(c_j \change{ \, +\,  \alpha \Big(1 - c_j - \sum_{k\neq j} w_{jk} \Big)_+} \right)\Big[ 	V_i(t,x(1-\phi_j w_{ij})) - V_i \Big] \Bigg\}.
	\end{alignedat}
\end{align}

Indeed, if we assume firm $i$ has logarithmic utility, and plug in the ansatz of $V_i(t,x) = g_i(t) + \log(x)$, the PDE admits a separable solution and we obtain the following ODE for $g_i(t)$.
\begin{align}
	\label{app:eq:partially_liq_rf_asset_ode}
	\begin{alignedat}{23}
		0&= g_i'(t)  + \sup_{c_i,  w_{i\cdot}} \Bigg\{ 	&&\left[\left(1-c_i\right)r  + \sum_{j\neq i}w_{ij} \mu_j  + \frac{\eta_i \mu_i}{\phi_i}\right] 
		+ \theta_i \bar F_i\left(c_i \change{ \, + \,  \alpha \Big(1 - c_i - \sum_{j\neq i} w_{ij} \Big)_+} \right)\log(1 - \eta_i)\\
		& && + \sum_{ j\neq i} \theta_j \bar F_j\left(c_j \change{ \, +\,  \alpha \Big(1 - c_j - \sum_{k\neq j} w_{jk} \Big)_+} \right)\log(1-\phi_j w_{ij})\Bigg\},
	\end{alignedat}
\end{align}
with terminal condition $g_i(T) = 0$.

We note that the first order conditions may not be sufficient for the optimization problem in Eq.~\eqref{app:eq:partially_liq_rf_asset_ode} since the objective function is not differentiable whenever $c_i + \sum_{j\neq i} w_{ij} = 1$.
However, we can still differentiate the optimization problem with respect to the decision variables, and obtain the following system of equations for the optimal controls of firm $i$:
\begin{align}
	\label{app:eq:partially_liq_rf_asset_opt_ctrls}
	0 &= \mu_j \change{ - \theta_i f_i \left( \hat c_i + \alpha \Big( 1 - \hat c_i - \sum_{k\neq i} \hat w_{ik} \Big)_+ \right) \left[ -\alpha \cdot \mathbf{1} \Big\{	1 - \hat c_i - \sum_{k\neq i} \hat w_{ik} > 0\Big\} \right] }\\
	& \qquad - \frac{\phi_j \theta_j \bar F_j \left( \hat c_j \change{ + \alpha \Big( 1 - \hat c_j - \sum_{k\neq j} \hat w_{jk} \Big)_+} \right)}{1 - \phi_j \hat w_{ij}}  \quad \text{ for all } j\neq i \\
	0 &= -r - \theta_i f_i\left(\hat c_i \change{ + \alpha \Big( 1 - \hat c_i - \sum_{k\neq i} \hat w_{ik} \Big)_+  } \right)\log(1-\eta_i) \change{ \left( 1 - \alpha \cdot \mathbf{1} \Big\{	1 - \hat c_i - \sum_{k\neq i} \hat w_{ik} > 0\Big\} \right). }
\end{align}
Observe that this system of equations differs from the system in Eq.~\eqref{sol:dec} if and only if $1 - \hat c_i - \sum_{k\neq i} \hat w_{ik} > 0$, and hence if firm $i$ holds a long position in the risk-free bond.

We have numerically obtained solutions for the optimal controls $(\hat c_i, \hat w_{i \cdot})$ in the symmetric case, where all banks have identical projects and utility functions. In doing so, we have seen that for large enough $n$, these optimal controls will exactly equal those in the main text (Section~\ref{ssec:opts-dec}). The rationale for this is intuitive. 
Under the assumption of identical banks, the total amount invested in other banks' projects $ \sum_{k\neq i} \hat w_{ik}$ grows linearly in $n$. This quantity eventually becomes large enough for bank $i$ to short the risk-free bond and use this to finance their interbank project investments, which implies that the system of equations in Eq.~\eqref{app:eq:partially_liq_rf_asset_opt_ctrls} will simplify to the system of equations that gives rise to the optimal controls in Eq.~\eqref{sol:dec}.

We note that this observation will easily extend to the centralized optimum, since the planner will always choose to hold no less cash than decentralized banks. Hence under the planner's optimum, the optimal project investment amounts are no smaller than those in the decentralized setting, and it will still be the case that banks hold short positions in the risk-free bond.

Regardless, we have obtained numerical solutions for the optimal controls for decentralized banks. Figure~\ref{fig:Partially_Liq_RF_Asset} compares the optimal cash reserves for a bank as $\alpha$, the fraction of liquidity available due to a long position in the risk-free asset, varies. The parameters used in these simulations can be found in Table~\ref{tab:Partially_Liq_RF_Asset_params}. We note that even for a 2-bank system, cash reserves quickly drop to zero as $\alpha$ grows. And of course, for $\alpha = 0$, we obtain the optimal controls from the model in the main text. In Figure~\ref{fig:Partially_Liq_RF_Asset_n_6}, with $n=6$ we find that for the majority of values of $\alpha$, the optimal cash allocations also coincide, since there is a sufficiently large number of return-bearing projects for a bank to invest in to justify shorting the risk-free asset. 
\begin{figure}[!h]
	\centering
	\begin{subfigure}{0.45\linewidth}
		\includegraphics[width = \linewidth]{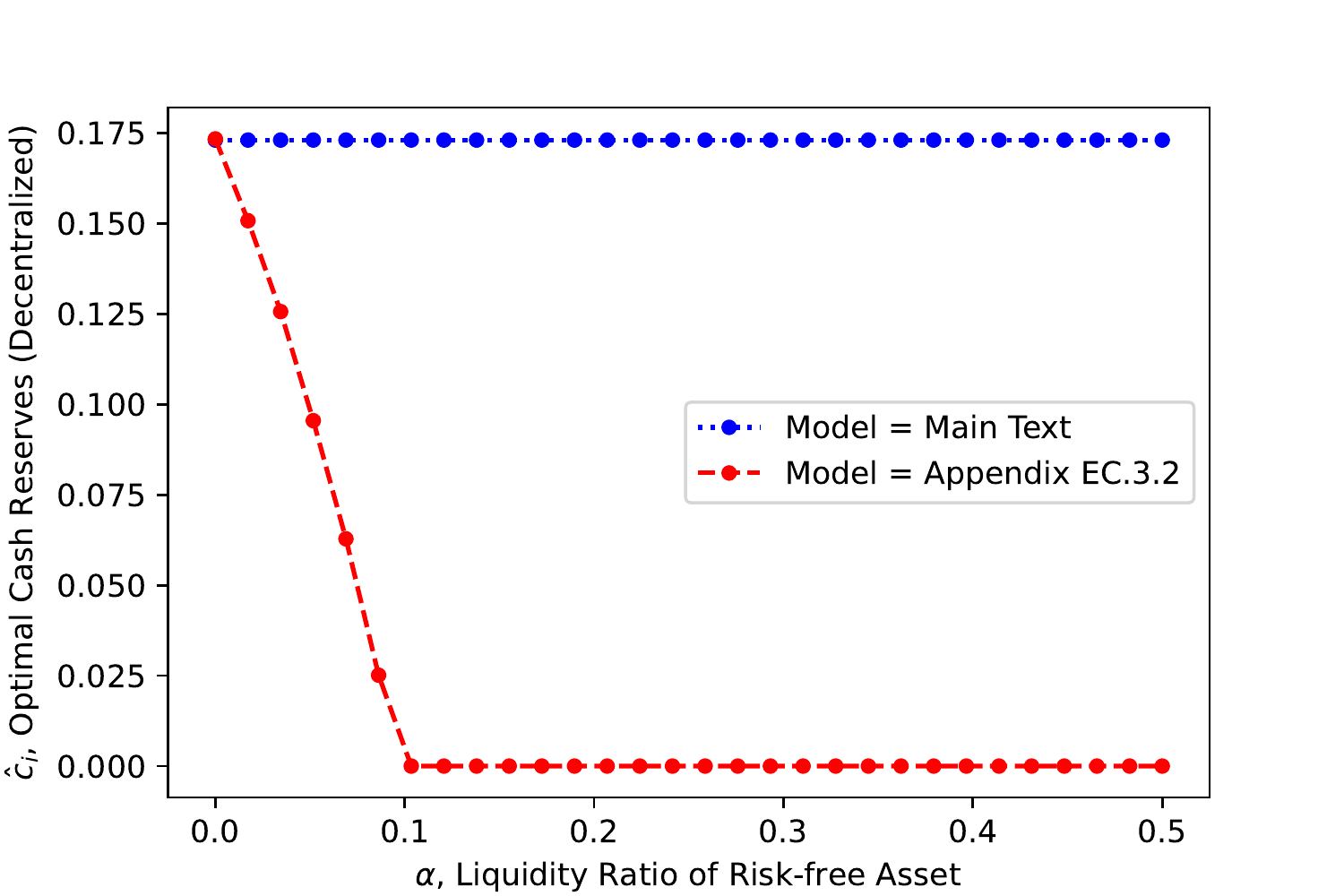}
		\captionsetup{font = small}
		\caption{$n=2$}
		\label{fig:Partially_Liq_RF_Asset_n_2}
	\end{subfigure}
	\hspace{2em}
	\begin{subfigure}{0.45\linewidth}
		\includegraphics[width = \linewidth]{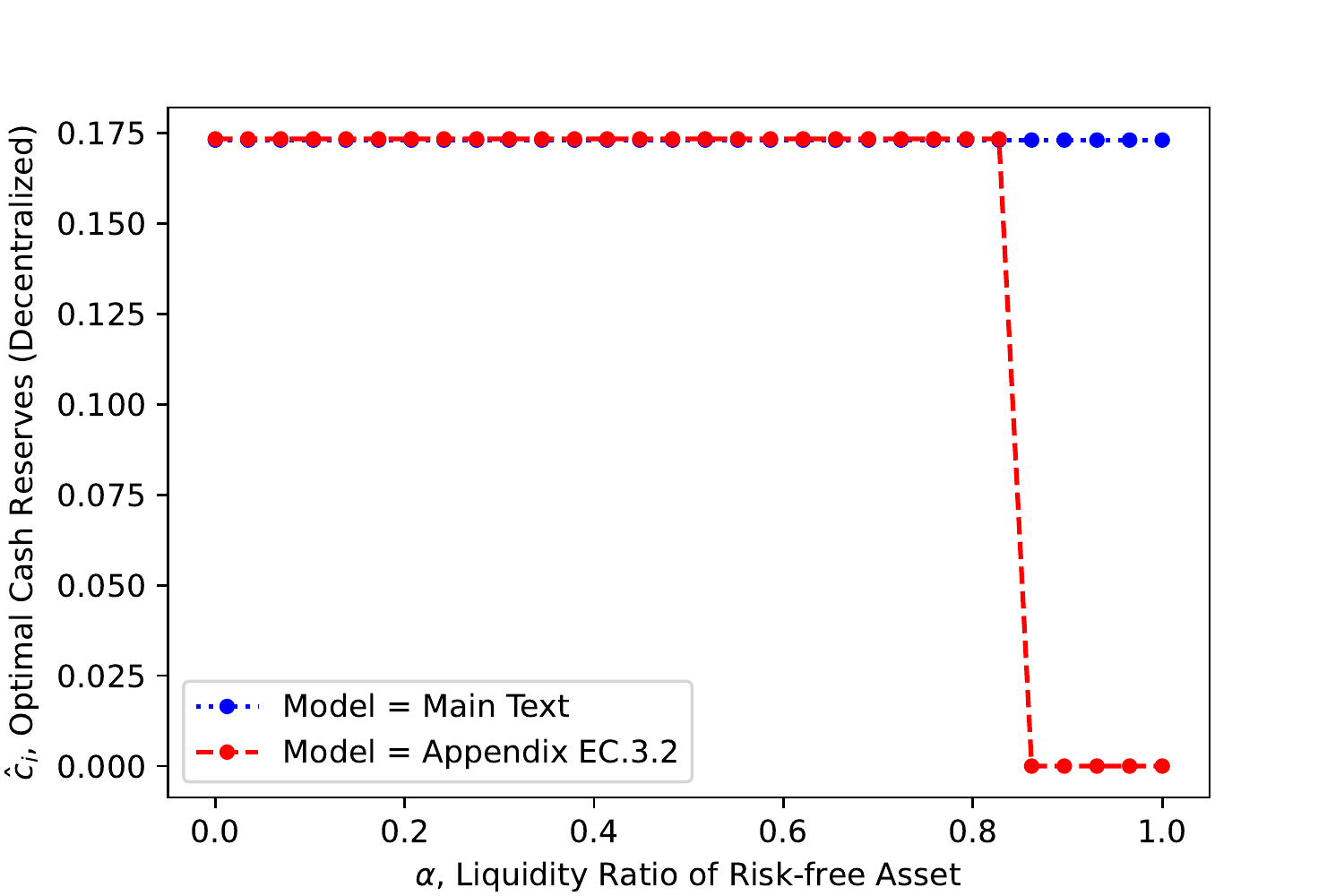}
		\captionsetup{font = small}
		\caption{$n=6$}
		\label{fig:Partially_Liq_RF_Asset_n_6}
	\end{subfigure}
	\captionsetup{width=.8\linewidth, font = small, justification=justified}
	\caption{Comparison of the optimal cash reserves when liquidity shocks directly induce losses regardless of size. Parameters (excluding $n$) are detailed in Table~\ref{tab:Partially_Liq_RF_Asset_params}.}
	\label{fig:Partially_Liq_RF_Asset}
\end{figure}

\begin{table}[h]
	\centering
	\captionsetup{width=.8\linewidth, font = small, justification=justified}
	\caption{Parameters used for obtaining optimal cash reserves shown in Figures~\ref{fig:Partially_Liq_RF_Asset_n_2} and~\ref{fig:Partially_Liq_RF_Asset_n_6}. Code is available \href{https://github.com/drigobon/rigobon-sircar-2022}{\color{blue}here}.}
	\label{tab:Partially_Liq_RF_Asset_params}
	\begin{tabular}{ c  c  l }
		\toprule
		Notation & Value & Description \\
		\midrule
		$r$ & 0.03 & Risk-free rate \\
		$\mu$ & 0.045 & Excess drift\\
		$\phi$ & 0.4 & Losses to External Investors\\
		$\eta$ & 0.3 & Losses to Associated Bank\\
		$F(x)$ & $1 - e^{-\frac{x}{\lambda}}$ & CDF of shock size\\
		$\lambda$ & 1 & Parameter of $F(\cdot)$ \\
		$\theta$ & 0.1 & Shock arrival rate\\
		$\gamma$ & 1 & Relative risk aversion coefficient\\
		\bottomrule
	\end{tabular}
\end{table}

\subsection{Endogenizing Self-Investment Parameter $\eta_i$}
\label{subapp:endog_eta}

A third assumption made in the main text concerns the fact that the parameter $\eta_i$, which governs the fraction of bank $i$'s wealth lost upon a liquidity shortage, is exogenous. In this section, we will relax this assumption and study the optimal control problem for bank $i$ in a decentralized setting, when $\eta_i$ can also be chosen freely.

The derivation of the HJB is largely identical, and after plugging in a separable form for the value function, we would obtain the following ODE for $g_i(t)$:
\begin{align}
	\label{app:eq:endog_eta_ode_g}
	\begin{alignedat}{23}
		0&=  g_i'(t) + \sup_{c_i,  \eta_i, w_{i\cdot}} \Bigg\{ 	&&\left(1-c_i\right)r  + \sum_{j\neq i}w_{ij} \mu_j  + \frac{\eta_i \mu_i}{\phi_i}+ \theta_i \bar F_i(c_i) \log(1 - \eta_i) \Big] \\
		& && + \sum_{ j\neq i} \theta_j \bar F_j(c_j)\log(1-\phi_j w_{ij})\Bigg\}.
	\end{alignedat}
\end{align}
The only difference between this equation and its analogue in the main text is that $\eta_i$ is included as a control variable. We will only analyze the optimization problem for both $c_i$ and $\eta_i$, since the first-order conditions for each $w_{ij}$ remain identical to that of the main text. 

Now, the first-order conditions for these two controls are as follows:
\begin{align}
	\label{app:eq:endog_eta_opt_ctrls}
	0 = -r - \theta_i f_i(\hat c_i) \log(1-\hat\eta_i), \qquad
	0 = \frac{\mu_i}{\phi_i} - \frac{\theta_i \bar F_i(\hat c_i)}{1-\hat\eta_i},
\end{align}
where $\hat\eta_i$ denotes the optimal choice of $\eta_i$ for bank $i$ in the decentralized setting. Substituting $1-\hat\eta_i$ by solving the second expression yields the following expression in only $\hat c_i$:
\begin{equation}
	0 = -r - \theta_i f_i(\hat c_i) \log\left(\frac{\phi_i \theta_i \bar F_i(\hat c_i)}{\mu_i}\right). 
\end{equation}
In general, this expression cannot be explicitly solved without further assumptions on the distribution of liquidity shocks. 

If, however, we assume that the size of liquidity shocks follows an exponential distribution, with $\bar F_i(x) = e^{-\frac{x}{\lambda_i}}$ and $f_i(x) = \lambda_i^{-1} e^{-\frac{x}{\lambda_i}}$, we can simplify and isolate $\hat c_i$ to obtain:
\begin{equation}
	\label{app:eq:endog_eta_opt_ci}
	\hat c_i = -\lambda_i \log\left( \frac{-r  \lambda_i}{\theta_i \mathbb{W}\left(\frac{-r \lambda_i \phi_i}{ \mu_i }\right) } \right),
\end{equation}
where $\mathbb{W}(x)$ denotes the Lambert-W function defined as the inverse function of $x e^x$, restricted to the range $[-1,\infty)$ and the domain $[-e^{-1}, \infty)$. We note that this expression holds only when there is an interior solution to the optimization problem, i.e. when the optimal $\hat c_i$ and $\hat\eta_i$ are strictly positive. 

The analogue of this expression in the main text is:
\begin{equation}
	\label{app:eq:main_text_opt_ci}
	\hat c_i = -\lambda_i \log\left( \frac{-r \lambda_i}{\theta_i \log(1-\eta_i)} \right),
\end{equation}
when the parameter $\eta_i$ is assumed to be exogenous and given. The two are equal if and only if $\log(1-\eta_i) = \mathbb{W}\left(\frac{-r \lambda_i \phi_i}{ \mu_i }\right)$, or equivalently when $(1-\eta_i) \log(1-\eta_i) = \frac{-r\lambda_i \phi_i}{\mu_i }$.

We can numerically compute the value of the objective function for various choices of $c_i$, where we use the second expression in Eq.~\eqref{app:eq:endog_eta_opt_ctrls} to compute the optimal value of $\eta_i$ for each possible $c_i$ as follows:
\begin{equation}
	\hat\eta_i(c_i) = \left(1 - \frac{\phi_i \theta_i \bar F_i(c_i)}{\mu_i}\right)_+,
\end{equation}
since we assume that firm $i$ still cannot short their own project. Figure~\ref{fig:Endog_Eta_Obj} shows, for different values of $c_i$, the quantity:
\begin{equation}
	\label{app:eq:obj_c_i_with_endog_eta}
	\mathrm{Obj}(c_i) = (1-c_i)r + \frac{\hat\eta_i(c_i) \mu_i}{\phi_i} + \theta_i \bar F_i(c_i) \log(1-\hat\eta_i(c_i)),
\end{equation}
which corresponds to the portion of the objective function in Eq.~\eqref{app:eq:endog_eta_ode_g} tied to the two controls $\eta_i$ and $c_i$. 

We have chosen the parameters in Table~\ref{tab:Endog_Eta_Obj_params} to show the existence of two optima. In one `corner' solution, firm $i$ chooses to hold zero cash reserves and zero stake in their own project. In the second `interior' solution, both of these optimal quantities are strictly positive fractions of $i$'s wealth. This can be contrasted with the second curve, which shows Eq~\eqref{app:eq:obj_c_i_with_endog_eta} when $\eta_i$ is fixed at a value of 0.5, and for which a unique maximum exists.

Note that the interior solution in Figure~\ref{fig:Endog_Eta_Obj} with endogenous $\eta_i$ represents an equilibrium in which projects are \textit{less} risky than their analogue in the main text model. Cash reserves are higher, and hence bank $i$'s degree of investment in their own project is also greater than the fixed value of 0.5. In contrast, the corner solution with zero cash reserves and zero investment in project $i$ is not particularly interesting nor realistic, as it represents a setting in which bank $i$ is wholly immune to all liquidity shocks. 
\newpage
\begin{figure}[htbp]
	\centering
	\includegraphics[width=0.75\textwidth]{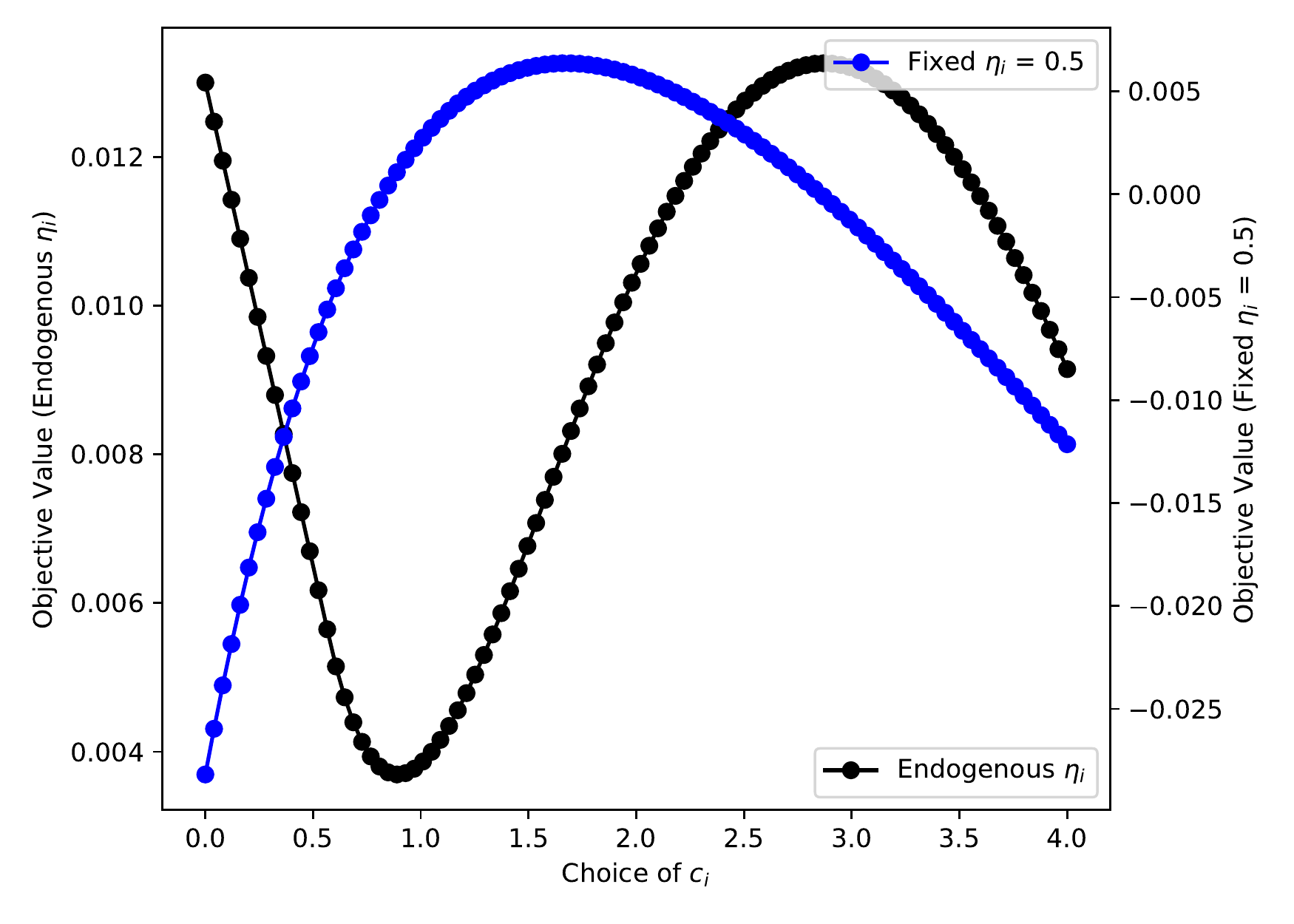}
	\captionsetup{width=.8\linewidth, font = small, justification=justified}
	\caption{Comparison of the optimal cash reserves when liquidity shocks directly induce losses regardless of size.}
	\label{fig:Endog_Eta_Obj}
\end{figure}
\begin{table}[h]
	\centering
	\captionsetup{width=.8\linewidth, font = small, justification=justified}
	\caption{Parameters used for generating Figure~\ref{fig:Endog_Eta_Obj}. Code is available \href{https://github.com/drigobon/rigobon-sircar-2022}{here}.}
	\label{tab:Endog_Eta_Obj_params}
	\begin{tabular}{ c  c  l }
		\toprule
		Notation & Value & Description \\
		\midrule
		$r$ & 0.013 & Risk-free rate \\
		$\mu_i$ & 0.045 & Excess drift\\
		$\phi_i$ & 0.8 & Losses to External Investors\\
		$F_i(x)$ & $1 - e^{-\frac{x}{\lambda_i}}$ & CDF of shock size\\
		$\lambda_i$ & 1 & Parameter of $F(\cdot)$ \\
		$\theta_i$ & 0.1 & Shock arrival rate\\
		$\gamma_i$ & 1 & Relative risk aversion coefficient\\
		\bottomrule
	\end{tabular}
\end{table}

}

\end{document}